\documentclass[12pt, a4paper]{article}
\usepackage[utf8]{inputenc}
\usepackage{tensor}
\usepackage[margin=0.8in]{geometry}
\usepackage{mathtools}
\usepackage{authblk}
\usepackage{amssymb}
\usepackage{url}
\usepackage{amsmath}
\usepackage{appendix}
\usepackage{hyperref}
\usepackage{caption}
\usepackage{amsthm}
\usepackage{stmaryrd}
\usepackage{bigints}
\newtheorem{theorem}{Theorem}[section]
\newtheorem{corollary}{Corollary}[theorem]
\newtheorem{lemma}[theorem]{Lemma}
\newtheorem{definition}[theorem]{Definition}

\newcommand{\overbar}[1]{\mkern 1.8mu\overline{\mkern-1.8mu#1\mkern-1.8mu}\mkern 1.8mu}
\usepackage{graphicx}
\numberwithin{equation}{section}

\title{\textbf{On energy bounds in asymptotically locally AdS spacetimes}}
\author{Virinchi Rallabhandi\thanks{v.v.rallabhandi@sms.ed.ac.uk}}
\affil{School of Mathematics and Maxwell Institute for Mathematical Sciences, University of Edinburgh, King's Buildings, Edinburgh, UK, EH9 3FD}
\date{19th of December, 2025}

\begin{document}

\maketitle

\begin{abstract}
\noindent This work considers positive energy theorems in asymptotically, locally AdS spacetimes. Particular attention is given to spacetimes where conformal infinity has  compact, Einstein cross-sections admitting Killing or parallel spinors; a positive energy theorem is derived for such spacetimes in terms of geometric data intrinsic to the cross-section. This is followed by the first complete proofs of the BPS inequalities in (the bosonic sectors of) 4D and 5D minimal, gauged supergravity, including with magnetic fields, provided the Maxwell field is exact. The BPS inequalities are proven for asymptotically AdS spacetimes, but also generalised to the aforementioned class of asymptotically, locally AdS spacetimes.

\end{abstract}

\tableofcontents

\section{Introduction}
The positive energy theorem stands as one of the most treasured and significant results in mathematical general relativity - originally proved by Schoen \& Yau based on minimal surface methods \cite{Schoen1979} and soon after by Witten \cite{Witten1981} based on spinor techniques. Witten's method suggested a number of extensions, including allowing a negative cosmological constant - the focus of the present work. The first positive energy theorems for asymptotically AdS spacetimes \cite{Gibbons1983, GibbonsHullWarner1983} followed soon after Witten's original work and were based on the Abbott-Deser definition of energy \& asymptotics \cite{Abbott1982}.

However, in the age of holography, a more natural choice of asymptotics is one based on a Fefferman-Graham expansion \cite{Fefferman1985, deHaro2001, Skenderis2002}. In particular, the Einstein equation is solved order by order from a timelike conformal boundary and the geometry of the boundary itself is arbitrary; the case of a static $\mathbb{R}\times S^2$ boundary reduces to the asymptotically (globally) AdS case. Rigorous definitions of energy were given in the latter context by \cite{Wang2001, Chrusciel2003, Chrusciel2001} and corresponding positive energy theorems were subsequently proven\footnote{Note that in the former two references, the asymptotics considered are Riemannian, not Lorentzian, and should be viewed as asymptotics for an initial data slice.}.

Having understood the ``global" case, the next logical extension is the ``local" case. The example of a toroidal boundary was considered in \cite{Chrusciel2006} and a more general analysis was performed in \cite{Cheng2005}. One of the main aims of this work is to built upon the latter. The present work will adopt a few conceptual differences though. Most saliently, the holographic renormalisation \cite{Skenderis2002} approach pursued by \cite{Cheng2005} will not be followed. Instead, energy will be defined using the background subtraction and Hamiltonian methods of \cite{Henneaux1985, Chrusciel2001, Chrusciel2003}. Furthermore, Killing spinors will play a crucial role in the analysis. To this end, a general formula is developed for imaginary Killing spinors on time-symmetric metrics with cross-sections admitting either parallel or real Killing spinors. This formula allows a derivation of a positive energy theorem based on data intrinsic to the cross-section. The theorem decomposes the ``Witten-Nester" energy \cite{Nester1981} of \cite{Cheng2005} into further ``conserved quantities" built from symmetries of the boundary geometry.

Given the deep connections between Witten's method and supergravity \cite{Horowitz1983}, another natural extension is to try prove BPS inequalities for (the bosonic sectors of) supergravity theories. This was realised very soon after Witten's original work to prove global mass-charge inequalities in asymptotically flat spacetimes in four and five dimensions \cite{Gibbons1982, Gibbons1994}. Some results already exist along these lines \cite{London1995, Kostelecky1996, Wang2015b, Nozawa2014b} in the context of asymptotically AdS spacetimes - i.e. in gauged supergravity theories. However, the magnetic field is set to zero in \cite{London1995} and a non-gauge-covariant connection is used in \cite{Kostelecky1996, Wang2015b}, thereby leading to some unnatural assumptions and different results to the present work when incorporating magnetic fields. \cite{Nozawa2014b} has the closest results to the present work, but still has some dependence on \cite{Kostelecky1996} in its reasoning\footnote{These issues are all discussed further in section \ref{sec:BPS}.}. This paper aims to build upon the literature by providing a more complete treatment of magnetic fields in the study of classical energy-charge inequalities with negative cosmological constant. This includes BPS inequalities with the standard asymptotically (globally) AdS asymptotics, but also a study of how BPS inequalities may be changed in the asymptotically, locally AdS case.

The paper begins in section \ref{sec:setup} by setting up the formalism of asymptotically, locally AdS spacetimes and applying Witten's method to prove a very general positive energy theorem - theorem \ref{thm:generalPositveEnergy}. Section \ref{sec:boundaryGeometry} follows with an analysis of the boundary geometry and specialises theorem \ref{thm:generalPositveEnergy} to boundaries with the properties discussed above. Section \ref{sec:examples} illustrates these results through various example boundary geometries - namely a squashed $S^7$, the torus, the round sphere and the lens space, $L(p, 1)$. Finally, section \ref{sec:BPS} presents the BPS inequalities. One additional appendix justifies that a particular Dirac operator can be inverted, as required for Witten's method. The main results are theorems \ref{thm:generalPositveEnergy}, \ref{thm:crossSectionKilling}, \ref{thm:crossSectionPositiveEnergy0}, \ref{thm:crossSectionPositiveEnergy3} and corollary \ref{thm:magneticCancellation}.

\subsection{Conventions}
$(M, g)$ is a smooth, $n$-dimensional spacetime with mostly pluses metric signature. $a, b, \cdots$, $I, J, \cdots$, $M, N, \cdots$ and $A, B, \cdots$ denote vielbein indices running $\{0, 1, \cdots, n - 1\}$, $\{1, 2, \cdots, n - 1\}$, $\{0, 2, 3, \cdots, n - 1\}$ and $\{2, 3, \cdots, n - 1\}$ respectively\footnote{Although in most cases, equations with $a, b, \cdots$ will remain valid if these were abstract indices.} while $\mu, \nu, \cdots$, $i, j ,\cdots$, $m, n, \cdots$ and $\alpha, \beta, \cdots$ denote coordinate indices running $\{0, 1, \cdots, n - 1\}$, $\{1, 2, \cdots, n - 1\}$, $\{0, 2, 3, \cdots, n - 1\}$ and $\{2, 3, \cdots, n - 1\}$ respectively. It's assumed $n \geq 4$ throughout. The gamma matrices are chosen so that $\gamma^a\gamma^b + \gamma^b\gamma^a = -2g^{ab}I$, $(\gamma^0)^\dagger = \gamma^0$ and $(\gamma^I)^\dagger = -\gamma^I$. $\gamma^{a_1\cdots a_k}$ denotes the antisymmetrisation, $\gamma^{[a_1}\cdots\gamma^{a_k]}$. Only Dirac spinors will be used in this work and the Dirac conjugate of a Dirac spinor, $\Psi$, is $\overline{\Psi} = \Psi^\dagger\gamma^0$. From section \ref{sec:boundaryGeometry} onwards, it will be necessary to consider both Dirac spinors on an $(n-2)$-dimensional surface and Dirac spinors on the full spacetime. Thus, the spinors on the $(n-2)$-dimensional surface exist in a vector space which has half the dimension of the spacetime spinors. In these cases, the spacetime spinor space can be viewed as a direct sum of the $(n-2)$-dimensional surface's spinor space with itself. Furthermore, the gamma matrices for the spacetime will be chosen as
\begin{align}
    \gamma^0 &= \begin{bmatrix}
        I & 0 \\
        0 & -I
    \end{bmatrix}, \,\, \gamma^1 = \begin{bmatrix}
        0 & -I \\
        I & 0
    \end{bmatrix} \,\,\mathrm{and}\,\,\gamma^A = \begin{bmatrix}
        0 & \hat{\gamma}^A \\
        \hat{\gamma}^A & 0
    \end{bmatrix},
\end{align}
where $\{\hat{\gamma}^A\}$ are the gamma matrices (in any representation) of the $(n-2)$-dimensional surface. Indeed, hats will be placed on all quantities intrinsic to the surface. The Levi-Civita connection of $g$ is denoted by $D$ and the Riemann tensor convention is $[D_a, D_b]V^c = R\indices{^c_d_a_b}V^d$. Thus, the Einstein equation is
\begin{align}
    R_{ab} - \frac{1}{2}g_{ab}R + \Lambda g_{ab} = 8\pi T_{ab}.
\end{align}
The units/length scales are chosen so that $\Lambda = -\frac{1}{2}(n - 1)(n - 2)$.

\section{Set-up and general positive energy theorem}
\label{sec:setup}
We begin with the set-up that will be used throughout. The central notion underpinning this work is the asymptotically, locally AdS spacetime.
\begin{definition}[Asymptotically, locally AdS]
    \label{def:asymptoticallyAdS}
    A spacetime, $(M, g)$, is said to be asymptotically, locally AdS if and only if the following conditions hold. $(M, g)$ must admit a conformal compactification such that conformal infinity, $\mathcal{I}$, is a timelike hypersurface. Then, in an open neighbourhood of $\mathcal{I}$, there must exist coordinates, $(r, x^m) = (r, t, x^\alpha)$, such that $\{r = \infty\}$ is $\mathcal{I}$ itself and $g$ admits a Fefferman-Graham expansion \cite{Fefferman1985},
    \begin{align}
        g &= \mathrm{e}^{2r}\left(f_{(0)mn} + \mathrm{e}^{-r}f_{(1)mn} + \cdots + \mathrm{e}^{-(n-1)r}f_{(n-1)mn} + r\mathrm{e}^{-(n-1)r}\tilde{f}_{(n-1)} + \cdots\right)\mathrm{d}x^m\otimes\mathrm{d}x^n \nonumber \\
        &\,\,\,\,\,\,\, + \mathrm{d}r\otimes\mathrm{d}r,
        \label{eq:feffermanGraham}
    \end{align}
    with $f_{(k)mn}$ and $\tilde{f}_{(k)mn}$ independent of $r$ for any $k$. The series, $f_{(0)mn} + \mathrm{e}^{-r}f_{(1)mn} + \cdots$, will be denoted $f_{mn}$ (when summed). It will be assumed $\mathcal{I}$ is diffeomeorphic to $\mathbb{
    R}\times S$ for some $(n-2)$D, spacelike, compact manifold, $S$. $t$ will be chosen as the coordinate along $\mathbb{R}$ and $x^\alpha$ are coordinates along $S$. Without loss of generality, the  coordinates are chosen so that $f_{(0)0\alpha} = 0$. Finally, the coordinates will always be ordered so that $t$ is the 0th coordinate, $r$ is the 1st coordinate and $x^\alpha$ are coordinates $2$ to $n - 1$.
    \label{def:locallyAdS}
\end{definition}
The lowest order $n$ terms of $f_{mn}$, i.e $f_{(0)}$, $\cdots$, $f_{(n-2)}$ and $\tilde{f}_{(n-1)}$, are uniquely determined by $f_{(0)mn}$ and its curvature \cite{Fefferman1985, deHaro2001, Skenderis2002}, i.e. specifying $f_{(0)}$ specifies $g$ up to $O\big(\mathrm{e}^{-(n-3)r}\big)$. $f_{(n-1)mn}$ is the first distinguishing term and is often thought of as a ``boundary stress tensor" \cite{Skenderis2002}, although the procedure followed in this work will not require holographic renormalisation like \cite{Skenderis2002}.

Typically, $f_{(2k+1)mn} = 0$ and there are further conditions on the trace and divergence of $f_{(n-1)mn}$, but these won't matter for the present discussion.

A striking feature of this expansion is that $f_{(0)}$ - the metric on $\mathcal{I}$ itself - can be freely chosen. A major focus of this work will be exploring the effects of the boundary geometry freedom on positive energy theorems. Energy itself will be defined using the ``background subtraction" method - see \cite{Hollands2005} for alternatives.
\begin{definition}[Background metric, $\bar{g}$]
    Any asymptotically, locally AdS metric, $\bar{g}$, is said to be a background metric for $g$ when $\bar{f}_{(0)mn} = f_{(0)mn}$.
\end{definition}
Hence, $g$ and $\bar{g}$ can only differ from $O\big(\mathrm{e}^{-(n-3)r}\big)$ onwards. The most commonly considered background metric will be AdS itself,
\begin{align}
    \bar{g} &= \mathrm{d}r\otimes\mathrm{d}r + \mathrm{e}^{2r}\left(-\bigg(1 + \frac{1}{4}\mathrm{e}^{-2r}\bigg)^2\mathrm{d}t\otimes\mathrm{d}t + \left(1 - \frac{1}{4}\mathrm{e}^{-2r}\right)^2g_{S^{n-2}}\right).
    \label{eq:ads}
\end{align}
In terms of Fefferman-Graham expansions, this has $S = S^{n-2}$ cross-section, $f_{(n-1)} = 0$ and $f_{(0)} = -\mathrm{d}t\otimes\mathrm{d}t + g_{S^{n-2}}$.

Let $\Sigma_t$ be any $(n-1)$D, spacelike hypersurface intersecting $\mathcal{I}$ such that $t$ is a constant on $\Sigma_t$ in an open neighbourhood of $\mathcal{I}$. Then, let $\Sigma_{t, r}$ be a constant $r$ cross-section of $\Sigma_t$ and let $\Sigma_{t, \infty} = \Sigma_t\cap\mathcal{I} = S$. Furthermore, throughout this work, $\int_{\Sigma_{t, \infty}}$ should be interpreted as $\lim_{r\to\infty}\int_{\Sigma_{t,r}}$. Now, the energy of $(M, g)$ can be defined as follows.
\begin{definition}[Energy]
    In asymptotically, locally AdS spacetimes, the energy relative to a background metric is defined to be
    \begin{align}
        E &= \frac{n-1}{16\pi}\int_{\Sigma_{t, \infty}}\hat{f}_{(0)}^{mn}(f_{(n-1)mn} - \bar{f}_{(n-1)mn})\,\sqrt{\mathrm{det}(f_{(0)\alpha\beta})}\,\mathrm{d}^{n-2}x,
        \label{eq:kottlerEnergy}
    \end{align}
    where $\hat{f}_{(0)}^{mn} = f_{(0)}^{mn} + n_{(0)}^mn_{(0)}^n$ is the induced (inverse) metric on $\Sigma_{t, \infty}$ times $\mathrm{e}^{2r}$ and $n_{(0)}^m$ is the unit normal to such cross-sections\footnote{The $\mathrm{e}^r$ factor scalings effectively just remove all $\mathrm{e}^r$ factors on the boundary, as would happen in the conformal compactification. Similarly, note that $\sqrt{\mathrm{det}(f_{(0)\alpha\beta})}\,\mathrm{d}^{n-2}x$ is simply the measure/volume form on $\Sigma_{t, \infty}$ after compactification.}.
\end{definition}
This definition can be derived by background subtraction from ``first principles" using the procedure described in \cite{Regge1974, Henneaux1985} to ensure a well-defined Hamiltonian formulation. In the examples studied later, only $\mathrm{AdS}_5$ and its variations will actually have a non-zero $\bar{f}_{(n-1)mn}$. The full geometric set-up described here is illustrated in figure \ref{fig:locallyAdSPenrose}.
\begin{figure}
    \centering
    \includegraphics[width=0.8\linewidth]{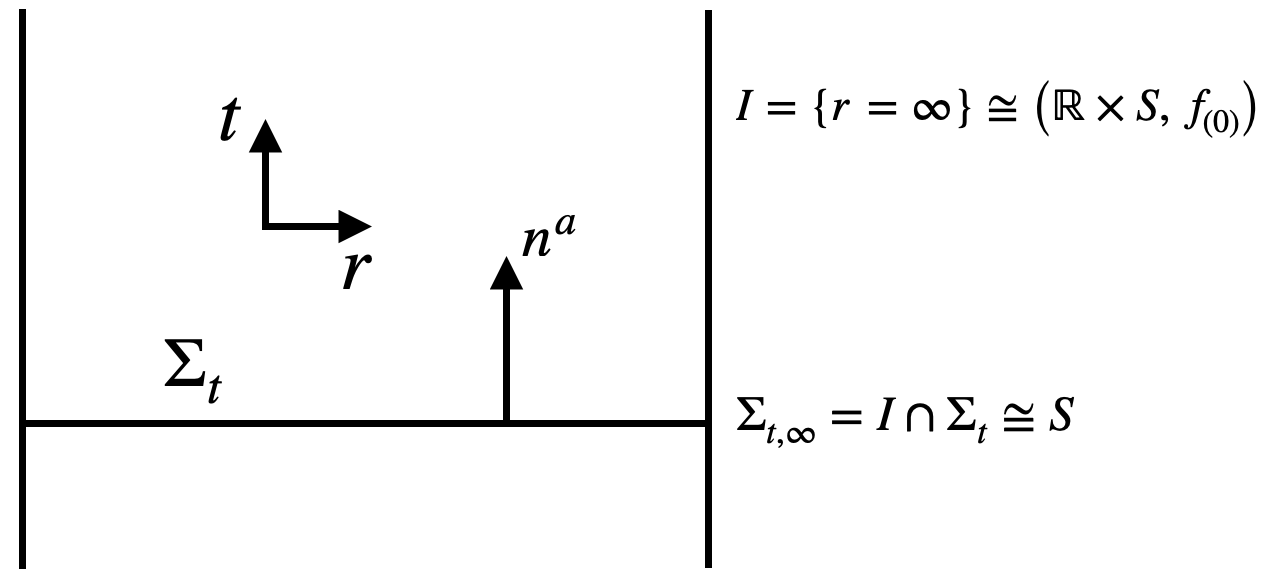}
    \caption{This is a Penrose diagram for defining asymptotically, locally AdS spacetimes and energy within them. Conformal infinity, $\mathcal{I}$, is $\{r = \infty\}$, where $r$ is the Fefferman-Graham coordinate. $\mathcal{I}$ is topologically $\mathbb{R}\times S$ with metric $f_{(0)mn}$. $t$ is a coordinate along the $\mathbb{R}$ direction, $\Sigma_t$ is a spacelike hypersurface which has constant $t$ near $\mathcal{I}$ and $n^a$ is the future-directed, timelike, unit normal to $\Sigma_t$. Energy is measured at $\Sigma_{t, \infty} = \mathcal{I}\cap \Sigma_t$.}
    \label{fig:locallyAdSPenrose}
\end{figure}

Key to Witten's method will be backgrounds that admit imaginary Killing spinors\footnote{The subsequent arguments would still work as long as the Killing spinor equation is satisfied up to $O(\mathrm{e}^{-(n-1)r})$ corrections.}.
\begin{definition}[Background Killing spinor]
    A spinor, $\varepsilon_k$, is called a Killing spinor of the background metric, $\bar{g}$, if and only if it satisfies
    \begin{align}
        \overbar{D}_a\varepsilon_k + \frac{\mathrm{i}}{2}\gamma_a\varepsilon_k = 0,
        \label{eq:backgroundKillingSpinor}
    \end{align}
    where $\overbar{D}_a$ is the Levi-Civita connection of $\bar{g}$. Similarly, denote the vielbeins associated to $\bar{g}$ as $\bar{e}^a$ and $\bar{e}_a$.
    \label{def:background}
\end{definition}
Of course, not every background metric admits a background Killing spinor. However, Witten's method works most naturally when there does exist a non-zero $\varepsilon_k$ - see \cite{Leitner2003, Baum1989} for a broad discussion on the admissible metrics. Furthermore, when no background Killing spinors exist, negative energies are possible - as illustrated by the examples in \cite{Chrusciel2023}.

Note that $\varepsilon_k$ may only be defined in an open neighbourhood of the ``boundary" at infinity or equation \ref{eq:backgroundKillingSpinor} may only have a solution in such a region. This is not a problem because equation \ref{eq:backgroundKillingSpinor} will only really be required in an open neighbourhood of infinity, say $\overbar{M}$, and $\varepsilon_k$ can be extended to a spinor on all of $\Sigma_t$ by multiplying it with a smooth function that's one near infinity but falls to zero within $\overbar{M}$.

\begin{definition}[$\nabla_a$, $\mathcal{A}_a$, $\mathbb{M}$, Witten-Nester 2-form and $Q(\varepsilon)$]
    \label{def:setup}
    When acting on any Dirac spinor, $\Psi$, of the spacetime, define the modified connection, $\nabla$, by
    \begin{align}
        \nabla_a\Psi &= D_a\Psi + \frac{\mathrm{i}}{2}\gamma_a\Psi + \mathcal{A}_a\Psi\,\,\,\,\,\,\,\mathrm{and} \\
        \nabla_a\overline{\Psi} &= D_a\overline{\Psi} - \frac{\mathrm{i}}{2}\overline{\Psi}\gamma_a + \overline{\Psi}\gamma^0\mathcal{A}_a^\dagger\gamma^0 = (\nabla_a\Psi)^\dagger\gamma^0,
    \end{align}
    where $\mathcal{A}_a$ is some Clifford algebra valued one-form\footnote{In the interests of generality, $\mathcal{A}_a$ is left quite unconstrained for now. But, in the examples, $\mathcal{A}_a$ will either be zero or a function of Maxwell fields such that $\nabla_a$ describes the gravitino transformation in a gauged supergravity.}. It will always be assumed that $\gamma^{IJ}\mathcal{A}_J$ is a hermitian matrix and that 
    \begin{align}
        \mathbb{M} &= 4\pi T^{0a}\gamma_0\gamma_a + \gamma^{IJ}D_I\mathcal{A}_J + \frac{\mathrm{i}(n-2)}{2}(\gamma^I\mathcal{A}_I + \mathcal{A}_I^\dagger\gamma^I) - \mathcal{A}^\dagger_I\gamma^{IJ}\mathcal{A}_J
    \end{align}
    is a non-negative definite matrix. Furthermore, it will be assumed $||\mathcal{A}_a||_0$ and $||\mathbb{M}||_0$ decay as $O(\mathrm{e}^{-(n-1)r})$ and $o(\mathrm{e}^{-(n-1)r})$ respectively near $\Sigma_{t, \infty}$, where $||\cdot||_0$ denotes the (pointwise) operator norm of a matrix\footnote{In the latter case, this is equivalent to saying $||\mathbb{M}||_0 \in L^1$.}. Finally, the Witten-Nester \cite{Witten1981, Nester1981} 2-form\footnote{Some authors would refer to this as the Hodge dual of a Witten-Nester $(n-2)$-form.} is defined to be
    \begin{align}
        E^{ab}(\varepsilon) &= \overline{\varepsilon}\gamma^{abc}\nabla_c\varepsilon + \mathrm{c.c} = \overline{\varepsilon}\gamma^{abc}\nabla_c\varepsilon - \nabla_c(\overline{\varepsilon})\gamma^{abc}\varepsilon.
        \label{eq:eAB}
    \end{align}
    Given a unit normal, $n_a$, to $\Sigma_t$, it is used to define a functional, $Q(\varepsilon)$, by 
    \begin{align}
        Q(\varepsilon) &= \int_{\Sigma_t}n_aD_b(E^{ba}(\varepsilon))\mathrm{d}V.
    \end{align}
\end{definition}
Having established these preliminaries, it's now possible to state the Lichnerowicz identity associated to the connections in the definition above\footnote{A similar Lichnerowicz identity is also given in \cite{Nozawa2014a}.}. Note that a variant of the Lichnerowicz identity always underpins any Witten-style positive energy theorem.
\begin{theorem}[Lichnerowicz identity]
    \label{thm:lichnerowicz}
    \begin{align}
        n_aD_b(E^{ba}(\varepsilon)) &= 2\left((\nabla_I\varepsilon)^\dagger\nabla^I\varepsilon - (\gamma^I\nabla_I\varepsilon)^\dagger\gamma^J\nabla_J\varepsilon + \varepsilon^\dagger\mathbb{M}\varepsilon\right).
    \end{align}
\end{theorem}
\begin{proof}
    First note the standard Lichnerowicz identity, 
    \begin{align}
        \gamma^{abc}D_bD_c\varepsilon &= -\frac{1}{2}\left(R^{ab} - \frac{1}{2}g^{ab}R\right)\gamma_b\varepsilon,
        \label{eq:lichnerowicz}
    \end{align}
    which is proven using $\gamma^{abc}$'s antisymmetry to convert $D_bD_c\varepsilon$ into $\frac{1}{2}[D_b, D_c]\varepsilon = -\frac{1}{8}R\indices{^d^e_b_c}\gamma_{de}\varepsilon$ and then applying the identity, $\gamma^{abc}\gamma_{de} = \gamma\indices{^a^b^c_d_e} - 6\gamma\indices{^[^a^b_[_e}\delta\indices{^c^]_d_]} + 6\gamma^{[a}\delta\indices{^b_[_e}\delta\indices{^c^]_d_]}$.
    
    By directly expanding and recombining terms, applying the standard Lichnerowcz identity, the Einstein equation and standard gamma matrix identities, one finds
    \begin{align}
        D_bE^{ba}(\varepsilon) &=  \bar{\varepsilon}(8\pi T^{ab}\gamma_b - \mathrm{i}(n-2)\gamma^{ab}\mathcal{A}_b - \mathrm{i}(n-2)\gamma^0\mathcal{A}_b^\dagger\gamma^0\gamma^{ab} - 2\gamma^0\mathcal{A}_b^\dagger\gamma^0\gamma^{bac}\mathcal{A}_c \nonumber \\
        &\,\,\,\,\,\,\,\,\,\,\, + \gamma^{bac}D_b\mathcal{A}_c - \gamma^0 D_b(\mathcal{A}_c^\dagger)\gamma^0\gamma^{bac})\varepsilon + 2\nabla_b(\bar{\varepsilon})\gamma^{bac}\nabla_c\varepsilon \nonumber \\
        &\,\,\,\,\,\,\, + \bar{\varepsilon}(\gamma^{bac}\mathcal{A}_c - \gamma^0\mathcal{A}_c^\dagger\gamma^0\gamma^{cab})D_b\varepsilon + D_b(\bar{\varepsilon})(\gamma^{cab}\mathcal{A}_c - \gamma^0\mathcal{A}_c^\dagger\gamma^0\gamma^{bac})\varepsilon.
    \end{align}
    In vielbein indices $n_a = -\delta_{a0}$. Then, in conjunction with the assumption that $\gamma^{IJ}\mathcal{A}_J$ is hermitian, the previous equation reduces to
    \begin{align}
        n_aD_bE^{ba}(\varepsilon) &= 2\varepsilon^\dagger\mathbb{M}\varepsilon + 2\nabla_I(\varepsilon)^\dagger\gamma^{IJ}\nabla_J\varepsilon = 2\left((\nabla_I\varepsilon)^\dagger\nabla^I\varepsilon - (\gamma^I\nabla_I\varepsilon)^\dagger\gamma^J\nabla_J\varepsilon + \varepsilon^\dagger\mathbb{M}\varepsilon\right),
    \end{align}
    which is the required version of the Lichnerowicz identity.
\end{proof}
Spinors are naturally defined by first choosing a frame. In the present context, this will lead to vielbeins associated to various different metrics - e.g. background, foreground, with $\mathrm{e}^r$ factors, without $\mathrm{e}^r$ factors etc. Thus, it will often be necessary to explicitly state which metric a given vielbein is associated to. A convenient notation for this will be $e^{(h)}$ for the vielbein (or inverse vielbein) associated to a metric, $h$. 

Given some background metric, $\bar{g}$, with vielbein, $\{\partial_r, \mathrm{e}^{-r}\bar{e}_M^{(\bar{f})m}\partial_m\}$, a natural vielbein for $g$ is $\partial_r$ together with
\begin{align}
    e_M &= \mathrm{e}^{-r}\bar{e}^{(\bar{f})m}_M\left(\partial_m - \frac{1}{2}\mathrm{e}^{-(n-1)r}(f_{(n-1)mp} - \bar{f}_{(n-1)mp})\bar{f}^{pn}_{(0)}\partial_n + O(\mathrm{e}^{-nr})\right).
    \label{eq:gVielbein}
\end{align}
In particular, this ansatz has $g(e_M, e_N) = \eta_{MN} + O(\mathrm{e}^{-nr})$. The specific $O(\mathrm{e}^{-nr})$ corrections won't be relevant for this work. Furthermore, given a background metric, the vielbein for $g$ will be fixed as $\{\partial_r, e_M\}$ throughout\footnote{Without this, it will not be possible to meaningfully talk about notions such as ``constant spinor" later.}. Background Killing spinors interact with this choice of vielbein to yield useful decay properties on $\Sigma_t$. The most common measure of decay will be inclusion in $L^2$, whose inner product is defined as
\begin{align}
    \langle A, B\rangle_{L^2} &= \int_{\Sigma_t}A^\dagger_{I_1\cdots I_k}B^{I_1\cdots I_k}\mathrm{d}V 
\end{align}
for any spinor-valued, rank-$k$ tensors, $A$ and $B$.
\begin{lemma}
    If a background Killing spinor, $\varepsilon_k$, is $O(\mathrm{e}^{r/2})$ near $\Sigma_{t, \infty}$, then $\nabla_I\varepsilon_k \in L^2$.
    \label{thm:killingSpinorL2}
\end{lemma}
\begin{proof}
    First note that to be in $L^2$, an object must decay faster than $O(\mathrm{e}^{-(n-2)r/2})$ because the integration measure over $\Sigma_t$ is $O(\mathrm{e}^{(n-2)r})$. Next, recall that given a vielbein, $e\indices{_a^\mu}\partial_\mu$, the spin connection coefficients are defined as
    \begin{align}
        \omega_{bca} &= \frac{1}{2}\left(g(e_a, [e_b, e_c]) - g(e_b, [e_c, e_a]) + g(e_c, [e_b, e_a])\right).
        \label{eq:spinConnection}
    \end{align}
    Since $\varepsilon_k$ is a background Killing spinor, with the vielbein chosen in equation \ref{eq:gVielbein} (along with $e_1 = \partial_r)$, one finds
    \begin{align}
        \nabla_M\varepsilon_k &= \left(-\frac{1}{2}\mathrm{e}^{-nr}(f_{(n-1)np} - \bar{f}_{(n-1)np})\bar{f}^{pm}\bar{e}^{(\bar{f})n}_M + O(\mathrm{e}^{-(n+1)r})\right)\partial_m\varepsilon_k \nonumber \\
        &\,\,\,\,\,\,\, - \frac{1}{4}(\omega_{ab M} - \overbar{\omega}_{ab M})\gamma^{ab}\varepsilon_k + \mathcal{A}_M\varepsilon_k \,\,\,\,\,\mathrm{and} \\
        \nabla_1\varepsilon_k &= -\frac{1}{4}(\omega_{ab 1} - \overbar{\omega}_{ab 1})\gamma^{ab}\varepsilon_k + \mathcal{A}_1\varepsilon_k.
    \end{align}
    $\mathcal{A}_a\varepsilon_k$ is comfortably in $L^2$ from the decay assumption on $||\mathcal{A}_a||_0$. Likewise, the term proportional to $\mathrm{e}^{-nr}\partial_m\varepsilon_k$ also easily decays fast enough to be in $L^2$. Finally, for the connection terms, using equations \ref{eq:spinConnection} and \ref{eq:gVielbein}, one finds
    \begin{align}
        \omega_{NPM} - \bar{\omega}_{NPM} &= O(\mathrm{e}^{-nr}),
        \label{eq:wNPMDecay} \\
        \omega_{MN1} - \bar{\omega}_{MN1} &= O(\mathrm{e}^{-(n-1)r}), \\
        \omega_{1NM} - \bar{\omega}_{1NM} &= O(\mathrm{e}^{-(n-1)r}) \,\,\,\mathrm{and} 
        \label{eq:w1NMDecay} \\
        \omega_{1M1} - \bar{\omega}_{1M1} &= 0.
    \end{align}
    In summary, all the terms decay quickly enough to have $\nabla_I\varepsilon_k \in L^2$.
\end{proof}
\begin{definition}[$p_M$]
    \label{def:pM}
    For future notational convenience, define
    \begin{align}
        p_{M} &= e_M^{(f_{(0)})m}e_0^{(f_{(0)})n}(f_{(n-1)mn} - \bar{f}_{(n-1)mn}) + \delta_{M0}f_{(0)}^{mn}(f_{(n-1)mn} - \bar{f}_{(n-1)mn}) \\
        &= \delta_{M0}\hat{f}_{(0)}^{mn}(f_{(n-1)mn} - \bar{f}_{(n-1)mn}) + \delta\indices{^A_M}e_A^{(f_{(0)})m}n_{(0)}^n(f_{(n-1)mn} - \bar{f}_{(n-1)mn}).
    \end{align}
\end{definition}
\begin{theorem}[Positive energy theorem]
    \label{thm:generalPositveEnergy}
    If the Einstein equation holds and $\exists$ a non-zero $\varepsilon_k$ with $\varepsilon_k$ being $O(\mathrm{e}^{r/2})$ near $\Sigma_{t, \infty}$, then $\exists \epsilon$ such that $\gamma^I\nabla_I\varepsilon = 0$ and 
    \begin{align}
        Q(\varepsilon) &= \frac{n-1}{2}\int_{\Sigma_{t, \infty}}\mathrm{e}^{-r}p_M\bar{\varepsilon}_k\gamma^M\varepsilon_k\sqrt{\mathrm{det}(f_{(0)\alpha\beta})}\,\mathrm{d}^{n-2}x \nonumber \\
        &\,\,\,\,\,\,\, + \int_{\Sigma_{t, \infty}}\mathrm{e}^{(n-2)r}\varepsilon^\dagger_k\left(\gamma^1\gamma^A\mathcal{A}_A + \mathcal{A}_A^\dagger\gamma^A\gamma^1\right)\varepsilon_k\sqrt{\mathrm{det}(f_{(0)\alpha\beta})}\,\mathrm{d}^{n-2}x  
        \label{eq:qBoundary} \\
        &= 2\int_{\Sigma_t}\left((\nabla_I\varepsilon)^\dagger\nabla^I\varepsilon + \varepsilon^\dagger\mathbb{M}\varepsilon\right)\mathrm{d}V \\
        &\geq 0.
    \end{align}
\end{theorem}
Note that the growth/decay of $\varepsilon_k$ and $||\mathcal{A}_a||_0$ ensure equation \ref{eq:qBoundary} is convergent.
\begin{proof}
    From lemma \ref{thm:killingSpinorL2} and the constructions of appendix \ref{sec:invertDirac}, there exists a functional space, $\mathcal{H}$, and a spinor, $\Psi \in \mathcal{H}$, such that $\gamma^I\nabla_I\Psi = \gamma^I\nabla_I\varepsilon_k$. Choose $\varepsilon = \varepsilon_k - \Psi$.

    From appendix \ref{sec:invertDirac}, there exists a Cauchy sequence of compactly supported, smooth spinors, $\{\psi_m\}_{m = 0}^{\infty}$, whose limit (under an appropriate norm) is $\Psi$. Next, let $\varepsilon_m = \varepsilon_k - \psi_m$ so that $\lim_{m\to\infty}\varepsilon_m = \varepsilon$. Then, since $\psi_m$ is compactly supported, in a vielbein where $n_a = -\delta_{a0}$ and $\partial_r = e_1$ in the asymptotic end, lemma \ref{thm:antisymmetricDerivative} implies
    \begin{align}
        Q(\varepsilon_m) &= \int_{\Sigma_{t, \infty}}E^{01}(\varepsilon_m)\mathrm{d}A = \int_{\Sigma_{t, \infty}}E^{01}(\varepsilon_k)\mathrm{d}A.
    \end{align}
    Since the RHS doesn't depend on $m$,
    \begin{align}
        \lim_{m\to\infty} Q(\varepsilon_m) &= \int_{\Sigma_{t, \infty}}E^{01}(\varepsilon_k)\mathrm{d}A \\
        &= \int_{\Sigma_{t, \infty}}\big(\varepsilon_k^\dagger\gamma^1\gamma^AD_A\varepsilon_k + D_A(\varepsilon_k)^\dagger\gamma^A\gamma^1\varepsilon_k - \mathrm{i}(n-2)\varepsilon_k^\dagger\gamma^1\varepsilon_k \nonumber \\
        &\,\,\,\,\,\,\, + \varepsilon_k^\dagger\gamma^1\gamma^A\mathcal{A}_A\varepsilon_k + \varepsilon_k^\dagger \mathcal{A}_A^\dagger\gamma^A\gamma^1\varepsilon_k\big)\mathrm{d}A.
        \label{eq:energyLimit}
    \end{align}
    To leading order, the measure, $\mathrm{d}A$, is $\mathrm{e}^{(n-2)r}\sqrt{\mathrm{det}(f_{(0)\alpha\beta})}\,\mathrm{d}^2x\cdots\mathrm{d}^{n-1}x$. This $\mathrm{e}^{(n-2)r}$ growth and the $O(\mathrm{e}^{r/2})$ growth of $\varepsilon_k$ means it suffices to keep only terms that decay as $O(\mathrm{e}^{-(n-1)r})$ or slower in the matrices in $E^{01}(\varepsilon_k)$. It's assumed $\mathcal{A}_A$ decays as $O(\mathrm{e}^{-(n-1)r})$, so those terms are kept as they are. Next, consider the derivative terms. Since $\varepsilon_k$ is a background Killing spinor,
    \begin{align}
        D_A\varepsilon_k &= (e\indices{_A^m} - \bar{e}\indices{_A^m})\partial_m\varepsilon_k - \frac{1}{4}(\omega_{abA} - \overbar{\omega}_{abA})\gamma^{ab}\varepsilon_k - \frac{\mathrm{i}}{2}\gamma_A\varepsilon_k.
    \end{align}
    From equation \ref{eq:gVielbein}, $(e\indices{_A^m} - \bar{e}\indices{_A^m})$ is $O(\mathrm{e}^{-nr})$ and hence can be ignored. Likewise, from equations \ref{eq:wNPMDecay} and \ref{eq:w1NMDecay}, only the connection coefficient terms where either $a$ or $b$ is 1 need to be retained. That leaves
    \begin{align}
        \varepsilon_k^\dagger\gamma^1\gamma^AD_A\varepsilon_k &\to -\frac{1}{2}(\omega_{1MA} - \overbar{\omega}_{1MA})\varepsilon_k^\dagger\gamma^A\gamma^M\varepsilon_k + \frac{\mathrm{i}}{2}(n-2)\varepsilon_k^\dagger\gamma^1\varepsilon_k.
    \end{align}
    From equation \ref{eq:spinConnection}, the connection difference is
    \begin{align}
        2(\omega_{1MA} - \overbar{\omega}_{1MA}) &= (g(e_A, [\partial_r, e_M]) + g(e_M, [\partial_r, e_A]) - \bar{g}(\bar{e}_A, [\partial_r, \bar{e}_M]) - \bar{g}(\bar{e}_M, [\partial_r, \bar{e}_A])).
    \end{align}
    This is symmetric in $A$ and $M$. Hence, it follows that
    \begin{align}
        &-\frac{1}{2}(\omega_{1MA} - \overbar{\omega}_{1MA})\varepsilon_k^\dagger\gamma^A\gamma^M\varepsilon_k \nonumber \\
        &= \frac{1}{2}\delta^{AB}(g(e_A, [\partial_r, e_B]) - \bar{g}(\bar{e}_A, [\partial_r, \bar{e}_B]))\varepsilon_k^\dagger\varepsilon_k \nonumber \\
        &\,\,\,\,\,\,\, -\frac{1}{4}(g(e_A, [\partial_r, e_0]) + g(e_0, [\partial_r, e_A]) - \bar{g}(\bar{e}_A, [\partial_r, \bar{e}_0]) - \bar{g}(\bar{e}_0, [\partial_r, \bar{e}_A]))\varepsilon_k^\dagger\gamma^A\gamma^0\varepsilon_k.
        \label{eq:connectionDifferenceAM}
    \end{align}
    Each term in equation \ref{eq:connectionDifferenceAM} can be calculated using equation \ref{eq:gVielbein}. In particular,
        \begin{align}
        &g(e_M, [\partial_r, e_N]) \nonumber \\
        &= g\bigg(\mathrm{e}^{-r}\bar{e}_M^{(\bar{f})m}\partial_m - \frac{1}{2}\mathrm{e}^{-nr}\bar{e}_M^{(\bar{f})m}(f_{(n-1)mp} - \bar{f}_{(n-1)mp})\bar{f}^{pn}_{(0)}\partial_n + O(\mathrm{e}^{-(n+1)r}), \nonumber \\
        &\,\,\,\,\,\,\, -\mathrm{e}^{-r}\bar{e}_N^{(\bar{f})q}\partial_q + \mathrm{e}^{-r}\partial_r\big(\bar{e}_N^{(\bar{f})q}\big)\partial_q + \frac{n}{2}\mathrm{e}^{-nr}\bar{e}_N^{(\bar{f})q}(f_{(n-1)qr} - \bar{f}_{(n-1)qr})\bar{f}^{rs}_{(0)}\partial_s + O(\mathrm{e}^{-(n+1)r})\bigg) \\
        &= \bar{e}_M^{(\bar{f})m}\partial_r\big(\bar{e}_N^{(\bar{f})q}\big)\Big(\bar{f}_{mq} + O(\mathrm{e}^{-(n-1)r})\Big) \nonumber \\
        &\,\,\,\,\,\,\, - \bar{e}_M^{(\bar{f})m}\bar{e}_N^{(\bar{f})q}\Big(\bar{f}_{mq} + \mathrm{e}^{-(n-1)r}(f_{(n-1)mq} - \bar{f}_{(n-1)mq}) - \frac{1}{2}\mathrm{e}^{-(n-1)r}(f_{(n-1)mp} - \bar{f}_{(n-1)mp})\bar{f}^{pn}_{(0)}\bar{f}_{(0)nq} \nonumber \\
        &\,\,\,\,\,\,\, - \frac{n}{2}\mathrm{e}^{-(n-1)r}(f_{(n-1)qr} - \bar{f}_{(n-1)qr})\bar{f}^{rs}_{(0)}\bar{f}_{(0)ms} + O(\mathrm{e}^{-nr})\Big) \\
        &= \bar{e}_M^{(\bar{f})m}\partial_r\big(\bar{e}_N^{(\bar{f})q}\big)\bar{f}_{mq} + O(\mathrm{e}^{-nr}) -\eta_{MN} \nonumber \\
        &\,\,\,\,\,\,\, + \frac{n-1}{2}\mathrm{e}^{-(n-1)r}\bar{e}_M^{(\bar{f})m}\bar{e}_N^{(\bar{f})n}(f_{(n-1)mn} - \bar{f}_{(n-1)mn}) + O(\mathrm{e}^{-nr})
    \end{align}
    and likewise
    \begin{align}
        \bar{g}(\bar{e}_M, [\partial_r, \bar{e}_N])
        &= \bar{g}\left(\mathrm{e}^{-r}\bar{e}_M^{(\bar{f})m}\partial_m, -\mathrm{e}^{-r}\bar{e}_N^{(\bar{f})n}\partial_n + \mathrm{e}^{-r}\partial_r\big(\bar{e}_N^{(\bar{f})n}\partial_n\big)\right) \\
        &= -\eta_{MN} + \bar{e}_M^{(\bar{f})m}\partial_r\big(\bar{e}_N^{(\bar{f})n}\big)\bar{f}_{mn}.
    \end{align}
    Putting them together yields
    \begin{align}
        g(e_M, [\partial_r, e_N]) - \bar{g}(\bar{e}_M, [\partial_r, \bar{e}_N]) &= \frac{n-1}{2}\mathrm{e}^{-(n-1)r}\bar{e}_M^{(\bar{f})m}\bar{e}_N^{(\bar{f})n}(f_{(n-1)mn} - \bar{f}_{(n-1)mn})
    \end{align}
    up to $O(\mathrm{e}^{-nr})$. Substituting this back into equation \ref{eq:connectionDifferenceAM} then implies
    \begin{align}
        -\frac{1}{2}(\omega_{1MA} - \overbar{\omega}_{1MA})\varepsilon_k^\dagger\gamma^A\gamma^M\varepsilon_k
        &= \delta^{AB}\frac{n-1}{4}\mathrm{e}^{-(n-1)r}\bar{e}_A^{(\bar{f})m}\bar{e}_B^{(\bar{f})n}(f_{(n-1)mn} - \bar{f}_{(n-1)mn})\varepsilon_k^\dagger\varepsilon_k \nonumber \\
        &\,\,\,\,\,\,\, + \frac{n-1}{4}\mathrm{e}^{-(n-1)r}\bar{e}_A^{(\bar{f})m}\bar{e}_0^{(\bar{f})n}(f_{(n-1)mn} - \bar{f}_{(n-1)mn})\bar{\varepsilon}_k\gamma^A\varepsilon_k \\
        &= \frac{n-1}{4}\mathrm{e}^{-(n-1)r}\bar{e}_M^{(\bar{f})m}\bar{e}_0^{(\bar{f})n}(f_{(n-1)mn} - \bar{f}_{(n-1)mn})\bar{\varepsilon}_k\gamma^M\varepsilon_k \nonumber \\
        &\,\,\,\,\,\,\, + \frac{n-1}{4}\mathrm{e}^{-(n-1)r}\bar{f}^{mn}(f_{(n-1)mn} - \bar{f}_{(n-1)mn})\varepsilon_k^\dagger\varepsilon_k
    \end{align}
    to leading order. The $\mathrm{e}^{-(n-1)r}$ factor and $\varepsilon_k = O(\mathrm{e}^{r/2})$ mean everything else can be kept to $O(1)$; anything lower order will go to zero in equation \ref{eq:energyLimit}. Therefore ultimately,
    \begin{align}
        -\frac{1}{2}(\omega_{1MA} - \overbar{\omega}_{1MA})\varepsilon_k^\dagger\gamma^A\gamma^M\varepsilon_k &= \frac{n-1}{4}\mathrm{e}^{-(n-1)r}p_M\bar{\varepsilon}_k\gamma^M\varepsilon_k + O(\mathrm{e}^{-(n-1)r}).
    \end{align}
    Substituting this back implies
    \begin{align}
        \varepsilon_k^\dagger\gamma^1\gamma^AD_A\varepsilon_k &\to \frac{n-1}{4}\mathrm{e}^{-(n-1)r}p_M\bar{\varepsilon}_k\gamma^M\varepsilon_k + \frac{\mathrm{i}}{2}(n-2)\varepsilon_k^\dagger\gamma^1\varepsilon_k
    \end{align}
    and consequently equation \ref{eq:energyLimit} reduces to
    \begin{align}
        \lim_{m \to \infty}Q(\varepsilon_m) &= \frac{n-1}{2}\int_{\Sigma_{t, \infty}}\mathrm{e}^{-r}p_M\bar{\varepsilon}_k\gamma^M\varepsilon_k\sqrt{\mathrm{det}(f_{(0)\alpha\beta})}\,\mathrm{d}^{n-2}x \nonumber \\
        &\,\,\,\,\,\,\, + \int_{\Sigma_{t, \infty}}\mathrm{e}^{(n-2)r}\varepsilon^\dagger_k\left(\gamma^1\gamma^A\mathcal{A}_A + \mathcal{A}_A^\dagger\gamma^A\gamma^1\right)\varepsilon_k\sqrt{\mathrm{det}(f_{(0)\alpha\beta})}\,\mathrm{d}^{n-2}x.
    \end{align}
    Finally, the LHS converges to $Q(\varepsilon)$ by theorem \ref{thm:bulkConvergence}.
\end{proof}

\section{Energy bounds and boundary geometry}
\label{sec:boundaryGeometry}
While theorem \ref{thm:generalPositveEnergy} applies somewhat generally, more progress can be made by restricting the boundary geometry further.
\begin{definition}[AdS with cross-section, $(S, h)$]
A metric is defined to be AdS with cross-section, $(S, h)$, if and only if
\begin{align}
    g &= \mathrm{d}r\otimes\mathrm{d}r + \mathrm{e}^{2r}\left(-\left(1 + \frac{c}{4}\mathrm{e}^{-2r}\right)^2\mathrm{d}t\otimes\mathrm{d}t + \left(1 - \frac{c}{4}\mathrm{e}^{-2r}\right)^2h\right)
    \label{eq:crossSectionMetric}
\end{align}
and $h$ is a Riemannian metric on a compact, $(n - 2)$D manifold, $S$, such that $R_{AB}^{(h)} = c(n-3)\delta_{AB}$ for $c = -1$, $0$ or $1$.
\end{definition}
It can be checked that any metric of this form satisfies the vacuum Einstein equation. While these metrics may be geodesically incomplete or contain conical singularities when $(S, h)$ is not the round sphere, these global issues don't prevent them from serving as natural background metrics for asymptotically, locally AdS spacetimes with static $\mathbb{R}\times S$ boundary geometry. The main objective of this section is to prove that if $h$ is ``symmetric" in some sense, then there is a positive energy theorem for spacetimes with these asymptotics. To apply the general result in theorem \ref{thm:generalPositveEnergy}, one then seeks a Killing spinor for this background.
\begin{theorem}
    The most general solution to $D_a\varepsilon_k + \frac{\mathrm{i}}{2}\gamma_a\varepsilon_k = 0$ for a metric that is AdS with cross-section, $(S, h)$, is
    \begin{align}
        \varepsilon_k = \begin{cases}
            \mathrm{e}^{r/2}P_1^-\varepsilon_h & \mathrm{for}\,\, c = 0 \\
            \mathrm{e}^{r/2}P_1^-\left(\mathrm{e}^{\mathrm{i}\gamma^0t/2} - \mathrm{i}\mathrm{e}^{-\mathrm{i}\gamma^0t/2}\right)\varepsilon_h + \frac{1}{2}\mathrm{e}^{-r/2}P_1^+\left(\mathrm{e}^{\mathrm{i}\gamma^0t/2} + \mathrm{i}\mathrm{e}^{-\mathrm{i}\gamma^0t/2}\right)\varepsilon_h & \mathrm{for}\,\,c = 1 \\
            0 & \mathrm{for}\,\,c = -1
        \end{cases}
        \label{eq:killingSpinorCrossSection}
    \end{align}
    where $P_1^\pm = \frac{1}{2}(I \pm \mathrm{i}\gamma^1)$, $\varepsilon_h$ solves $D_A^{(h)}\varepsilon_h = \frac{c}{2}\gamma_A\varepsilon_h$ and $\partial_t\varepsilon_h = 0$. In the former equation, $D_A^{(h)}$ is the formal expression for the Levi-Civita connection of $h$, e.g. if $\{e^{(h)A}\}_{A = 2}^{n - 1}$ is a vielbein for $h$ and $\omega_{BCA}^{(h)}$ are the corresponding spin coefficients, then\footnote{This is only a formal Levi-Civita connection because the gamma matrices used are from the spacetime, not the cross-section. In effect, this is a Levi-Civita connection that results from a reducible representation of the Clifford algebra.} $D_A^{(h)}\varepsilon_h = e_A^{(h)\alpha}\partial_\alpha\varepsilon_h - \frac{1}{4}\omega_{BCA}^{(h)}\gamma^{BC}\varepsilon_h$.
    \label{thm:crossSectionKilling}
\end{theorem}
\begin{proof}
    In coordinates, the equation to solve is
    \begin{align}
        e\indices{_a^\mu}\partial_\mu\varepsilon_k - \frac{1}{4}\omega_{bca}\gamma^{bc}\varepsilon_k + \frac{\mathrm{i}}{2}\gamma_a\varepsilon_k = 0.
        \label{eq:epsilonKCoordinates}
    \end{align}
    The most natural vielbein to choose in this context is
    \begin{align}
        e^0 &= \left(\mathrm{e}^r + \frac{c}{4}\mathrm{e}^{-r}\right)\mathrm{d}t, \,\, e^1 = \mathrm{d}r \,\,\mathrm{and} \,\,e^A = \left(\mathrm{e}^r - \frac{c}{4}\mathrm{e}^{-r}\right)e^{(h)A},
    \end{align}
    where $\{e^{(h)A}\}_{A = 2}^{n - 1}$ is a vielbein for $h$. Then, one finds 
    \begin{align}
        \omega_{01} &= -\frac{\mathrm{e}^r - \frac{c}{4}\mathrm{e}^{-r}}{\mathrm{e}^r + \frac{c}{4}\mathrm{e}^{-r}}e^0, \,\,
        \omega_{A1} = \frac{\mathrm{e}^r + \frac{c}{4}\mathrm{e}^{-r}}{\mathrm{e}^r - \frac{c}{4}\mathrm{e}^{-r}}e^A \,\,\mathrm{and} \,\,
        \omega_{AB} = \omega_{AB}^{(h)},
    \end{align}
    where $\omega_{AB}^{(h)}$ are the connection 1-forms of $h$. 
    
    Thus, the $a = 1$ component of equation \ref{eq:epsilonKCoordinates} says $0 = \partial_r\varepsilon_k + \frac{\mathrm{i}}{2}\gamma_1\varepsilon_k$, which immediately integrates to $\varepsilon_k = \mathrm{e}^{-\mathrm{i}\gamma^1r/2}\varepsilon_r$ for some spinor, $\varepsilon_r$, that doesn't depend on $r$. Projecting $\varepsilon_r$ onto eigenspaces of $\gamma^1$ using $P^\pm_1 = \frac{1}{2}(I \pm \mathrm{i}\gamma^1)$ then yields 
    \begin{align}
        \varepsilon_k = \mathrm{e}^{-\mathrm{i}\gamma^1r/2}(P_1^-\varepsilon_- + P_1^+\varepsilon_+) = \mathrm{e}^{r/2}P_1^-\varepsilon_- + \mathrm{e}^{-r/2}P_1^+\varepsilon_+
    \end{align}
    for spinors, $\varepsilon_\pm$, that don't depend on $r$.

    Next, with this expression for $\varepsilon_k$, one finds the $a = 0$ component of equation \ref{eq:epsilonKCoordinates} reduces to
    \begin{align}
        0 &= \mathrm{e}^{r/2}P_1^-\left(\partial_t\varepsilon_- - \mathrm{i}\gamma^0\varepsilon_+\right) + \mathrm{e}^{-r/2}P_1^+\left(\partial_t\varepsilon_+ - \frac{\mathrm{i}c}{4}\gamma^0\varepsilon_-\right).
    \end{align}
    Since the two $\gamma^1$ eigenspaces have no non-trivial intersection, it follows that
    \begin{align}
        \partial_t\varepsilon_- = \mathrm{i}\gamma^0\varepsilon_+\,\,\,\mathrm{and}\,\,\,\, \partial_t\varepsilon_+ = \frac{\mathrm{i}c}{4}\gamma^0\varepsilon_-.
        \label{eq:dtEpsilonPM}
    \end{align}
    Finally, noting $\omega_{AB} = \omega_{AB}^{(h)} \implies \omega_{ABC} = \frac{1}{\mathrm{e}^r - \frac{c}{4}\mathrm{e}^{-r}}\omega_{ABC}^{(h)}$ along the way, the $a = A$ component of equation \ref{eq:epsilonKCoordinates} reduces to
    \begin{align}
        0 &= \mathrm{e}^{r/2}P_1^-\left(D_A^{(h)}\varepsilon_-+ \mathrm{i}\gamma_A\varepsilon_+\right) + \mathrm{e}^{-r/2}P_1^+\left(D_A^{(h)}\varepsilon_+ - \frac{\mathrm{i}c}{4}\gamma_A\varepsilon_-\right),
    \end{align}
    from which it follows that 
    \begin{align}
        D_A^{(h)}\varepsilon_- = -\mathrm{i}\gamma_A\varepsilon_+\,\,\, \mathrm{and} \,\,\, D_A^{(h)}\varepsilon_+ = \frac{\mathrm{i}c}{4}\gamma_A\varepsilon_-.
        \label{eq:dAEpsilonPM}
    \end{align}
    
    First consider equations \ref{eq:dtEpsilonPM} and \ref{eq:dAEpsilonPM} for $c = 0$. From the former it immediately follows that $\varepsilon_- = \mathrm{i}t\gamma^0\varepsilon_+ + \varepsilon_h$ for some spinor, $\varepsilon_h$, that (like $\varepsilon_+$) doesn't depend on $t$. Then, the latter implies $D^{(h)A}D_A^{(h)}\varepsilon_h = -\mathrm{i}\gamma^AD_A^{(h)}\varepsilon_+ - 0 = 0$.

    Therefore $D_A^{(h)}\varepsilon_h = 0$ because integrating on the cross-section, $S$, says
    \begin{align}
        0 &= \int_{S}\varepsilon_h^\dagger D^{(h)A}D_A^{(h)}(\varepsilon_h)\mathrm{d}A(h) = -\int_{S}\left(D^{(h)A}\varepsilon_h\right)^\dagger D_A^{(h)}(\varepsilon_h)\mathrm{d}A(h).
    \end{align}
    Then $-\mathrm{i}\gamma_A\varepsilon_+ = D_A^{(h)}\varepsilon_- = 0 \implies \varepsilon_+ = 0$ and leaves $\varepsilon_- = \varepsilon_h$, completing the $c = 0$ analysis.

    Next, consider $c = -1$. Equation \ref{eq:dAEpsilonPM} now implies $D^{(h)A}D_A^{(h)}\varepsilon_- = \frac{n-2}{4}\varepsilon_-$. Therefore,
    \begin{align}
        \int_{S}\varepsilon_-^\dagger\varepsilon_-\mathrm{d}A(h) &= \frac{4}{n-2}\int_{S}\varepsilon_-^\dagger D^{(h)A}D_A^{(h)}(\varepsilon_-)\mathrm{d}A(h) \\
        &= -\frac{4}{n-2}\int_{S}(D^{(h)A}\varepsilon_-)^\dagger D_A^{(h)}(\varepsilon_-)\mathrm{d}A(h).
    \end{align}
    As the LHS is non-negative and the RHS is non-positive, it must be that both are zero. Hence, $\varepsilon_- = 0$ from the LHS and subsequently $\varepsilon_+ = 0$ from $D_A^{(h)}\varepsilon_- = -\mathrm{i}\gamma_A\varepsilon_+$.

    Finally, consider $c = 1$. Let
    \begin{align}
        \psi = \varepsilon_- + 2\varepsilon_+ \,\,\,\mathrm{and}\,\,\,\varphi = \varepsilon_- - 2\varepsilon_+ \iff \varepsilon_- = \frac{1}{2}(\psi + \varphi) \,\,\,\mathrm{and}\,\,\, \varepsilon_+ = \frac{1}{4}(\psi - \varphi).
    \end{align}
    Then, equations \ref{eq:dtEpsilonPM} and \ref{eq:dAEpsilonPM} are
    \begin{align}
        \partial_t\psi = \frac{\mathrm{i}}{2}\gamma^0\psi,\,\, \partial_t\varphi = -\frac{\mathrm{i}}{2}\gamma^0\varphi, \,\, D_A^{(h)}\psi = \frac{\mathrm{i}}{2}\gamma_A\varphi\,\, \mathrm{and}\,\, D_A^{(h)}\varphi = -\frac{\mathrm{i}}{2}\gamma_A\psi.
    \end{align}
    The first two immediately integrate to $\psi = 2\mathrm{e}^{\mathrm{i}\gamma^0t/2}\psi_t$ and $\varphi = 2\mathrm{e}^{-\mathrm{i}\gamma^0t/2}\varphi_t$ for some spinors, $\psi_t$ and $\varphi_t$, that don't depend on $t$ or $r$. Equivalently,
    \begin{align}
        \varepsilon_- = \mathrm{e}^{\mathrm{i}\gamma^0t/2}\psi_t + \mathrm{e}^{-\mathrm{i}\gamma^0t/2}\varphi_t\,\,\, \mathrm{and}\,\,\, \varepsilon_+ = \frac{1}{2}(\mathrm{e}^{\mathrm{i}\gamma^0t/2}\psi_t - \mathrm{e}^{-\mathrm{i}\gamma^0t/2}\varphi_t).
    \end{align}
    By construction, $P_1^\pm\varepsilon_\pm = \varepsilon_\pm \iff \varepsilon_\pm = \pm\mathrm{i}\gamma^1\varepsilon_\pm$ without loss of generality. Therefore,
    \begin{align}
        \varepsilon_- = \mathrm{e}^{\mathrm{i}\gamma^0t/2}\psi_t + \mathrm{e}^{-\mathrm{i}\gamma^0t/2}\varphi_t = -\mathrm{i}\gamma^1(\mathrm{e}^{\mathrm{i}\gamma^0t/2}\psi_t + \mathrm{e}^{-\mathrm{i}\gamma^0t/2}\varphi_t) = -\mathrm{i}\mathrm{e}^{-\mathrm{i}\gamma^0t/2}\gamma^1\psi_t - \mathrm{i}\mathrm{e}^{\mathrm{i}\gamma^0t/2}\gamma^1\varphi_t.
    \end{align}
    Setting $t = 0, \,\pi$ this equation implies
    \begin{align}
        \psi_t + \varphi_t = -\mathrm{i}\gamma^1\psi_t - \mathrm{i}\gamma^1\varphi_t \,\,\,\mathrm{and}\,\,\,\psi_t - \varphi_t = \mathrm{i}\gamma^1\psi_t - \mathrm{i}\gamma^1\varphi_t
    \end{align}
    respectively. Putting the two equations together, it follows that $\psi_t = -\mathrm{i}\gamma^1\varphi_t$. Hence, 
    \begin{align}
        \varepsilon_- &= \mathrm{e}^{\mathrm{i}\gamma^0t/2}\psi_t + \mathrm{e}^{-\mathrm{i}\gamma^0t/2}\varphi_t = (I - \mathrm{i}\gamma^1)\mathrm{e}^{\mathrm{i}\gamma^0t/2}\psi_t \,\,\,\mathrm{and} \\
        \varepsilon_+ &= \frac{1}{2}(\mathrm{e}^{\mathrm{i}\gamma^0t/2}\psi_t - \mathrm{e}^{-\mathrm{i}\gamma^0t/2}\varphi_t) = \frac{1}{2}(I + \mathrm{i}\gamma^1)\mathrm{e}^{\mathrm{i}\gamma^0t/2}\psi_t.
    \end{align}
    Let $\varepsilon_h = \frac{1}{2}(I + \gamma^1)\psi_t \iff \psi_t = (I - \gamma^1)\varepsilon_h$. Then, equation \ref{eq:dAEpsilonPM} just says $D_A^{(h)}\varepsilon_h = \frac{1}{2}\gamma_A\varepsilon_h$, while
    \begin{align}
        \varepsilon_- &= (I - \mathrm{i}\gamma^1)\mathrm{e}^{\mathrm{i}\gamma^0t/2}(I -\gamma^1)\varepsilon_h = (I - \mathrm{i}\gamma^1)(\mathrm{e}^{\mathrm{i}\gamma^0t/2} - \mathrm{i}\mathrm{e}^{-\mathrm{i}\gamma^0t/2})\varepsilon_h \,\,\,\mathrm{and} \\
        \varepsilon_+ &= \frac{1}{2}(I + \mathrm{i}\gamma^1)\mathrm{e}^{\mathrm{i}\gamma^0t/2}(I -\gamma^1)\varepsilon_h = \frac{1}{2}(I + \mathrm{i}\gamma^1)(\mathrm{e}^{\mathrm{i}\gamma^0t/2} + \mathrm{i}\mathrm{e}^{-\mathrm{i}\gamma^0t/2})\varepsilon_h.
    \end{align}
    It can now be checked directly that both equations in \ref{eq:dAEpsilonPM} hold with these $\varepsilon_\pm$, thereby completing the $c = 1$ analysis.
\end{proof}
Solutions to $D_A^{(h)}\varepsilon_h = \frac{c}{2}\gamma_A\varepsilon_h$ are well-studied mathematical problems \cite{Wang1989, Bar1993}. However, one subtlety in comparing with the literature is that $\{\gamma^A\}_{A = 2}^{n-1}$ don't form an irreducible representation of $h$'s Clifford algebra; an irreducible representation of a (Riemannian) Clifford algebra with $n - 2$ elements would have $2^{\lfloor(n-2)/2\rfloor}\times2^{\lfloor(n-2)/2\rfloor}$ matrices, not $2^{\lfloor n/2\rfloor}\times2^{\lfloor n/2\rfloor}$ matrices like $\{\gamma^A\}_{A = 2}^{n-1} \subset \{\gamma^a\}_{a = 0}^{n-1}$. The doubled size means there are effectively two irreducible representations summed in $\gamma^A$. This degeneracy can be lifted by choosing the spacetime gamma matrices to be
\begin{align}
    \gamma^0 &= \begin{bmatrix}
        I & 0 \\
        0 & -I
    \end{bmatrix}, \,\, \gamma^1 = \begin{bmatrix}
        0 & -I \\
        I & 0
    \end{bmatrix} \,\,\mathrm{and}\,\,\gamma^A = \begin{bmatrix}
        0 & \hat{\gamma}^A \\
        \hat{\gamma}^A & 0
    \end{bmatrix},
    \label{eq:cliffordBlocks}
\end{align}
where $\hat{\gamma}^A$ are gamma matrices of the Riemannian manifold, $(S, h)$. In this representation, the spacetime spinor space can be viewed as a direct sum of the cross-section spinor space with itself. Now, $D_A^{(h)}\varepsilon_h = \frac{c}{2}\gamma_A\varepsilon_h$ can be written in a form that's truly intrinsic to the cross-section.
\begin{lemma}
    \label{thm:killingDecomposed}
    Let $\hat{D}_A^{(h)} = e_A^{(h)\alpha}\partial_\alpha - \frac{1}{4}\omega^{(h)}_{BCA}\hat{\gamma}^{BC}$ be the spin connection intrinsic to $h$. Then, the most general solution to $D_A^{(h)}\varepsilon_h = 0$ is
    \begin{align}
        \varepsilon_h &= \begin{bmatrix}
            \hat{\psi} \\
            \hat{\varphi}
        \end{bmatrix} \,\,\,\mathrm{with}\,\,\, \hat{D}_A^{(h)}\hat{\psi} = \hat{D}_A^{(h)}\hat{\varphi} = 0
    \end{align}
    and the most general solution to $D_A^{(h)}\varepsilon_h = \frac{1}{2}\gamma_A\varepsilon_h$ is
    \begin{align}
        \varepsilon_h &= \frac{1}{2}\begin{bmatrix}
        \hat{\varepsilon}_h^{(+)} + \hat{\varepsilon}_h^{(-)} \\
        \hat{\varepsilon}_h^{(+)} - \hat{\varepsilon}_h^{(-)}
    \end{bmatrix} \,\,\,\mathrm{with}\,\,\, \hat{D}_A^{(h)}\hat{\varepsilon}_h^{(\pm)} = \pm\frac{1}{2}\hat{\gamma}_A\hat{\varepsilon}_h^{(\pm)}.
    \end{align}
\end{lemma}
\begin{proof}
    Let $\varepsilon_h = (\hat{\psi}, \hat{\varphi})^T$ in the chosen representation. Then,
    \begin{align}
        D_A^{(h)}\varepsilon_h &= \left(e_A^{(h)\alpha}\partial_\alpha - \frac{1}{8}\omega_{BCA}^{(h)}\left(\begin{bmatrix}
        0 & \hat{\gamma}^A \\
        \hat{\gamma}^A & 0
    \end{bmatrix}\begin{bmatrix}
        0 & \hat{\gamma}^B \\
        \hat{\gamma}^B & 0
    \end{bmatrix} - \begin{bmatrix}
        0 & \hat{\gamma}^B \\
        \hat{\gamma}^B & 0
    \end{bmatrix}\begin{bmatrix}
        0 & \hat{\gamma}^A \\
        \hat{\gamma}^A & 0
    \end{bmatrix}\right)\right)\begin{bmatrix}
        \hat{\psi} \\ \hat{\varphi}
    \end{bmatrix} \\
    &= \begin{bmatrix}
        \hat{D}_A^{(h)}\hat{\psi} \\ \hat{D}_A^{(h)}\hat{\varphi}
    \end{bmatrix}.
    \end{align}
    Therefore the claim about $D_A^{(h)}\varepsilon_h = 0$ immediately follows. Meanwhile, since
    \begin{align}
        \gamma_A\varepsilon_h &= \begin{bmatrix}
        0 & \hat{\gamma}_A \\
        \hat{\gamma}_A & 0
    \end{bmatrix}\begin{bmatrix}
        \hat{\psi} \\ \hat{\varphi}
    \end{bmatrix} = \begin{bmatrix}
        \hat{\gamma}_A\hat{\varphi} \\ \hat{\gamma}_A\hat{\psi}
    \end{bmatrix},
\end{align}
it also follows that
\begin{align}
    0 &= D_A^{(h)}\varepsilon_h - \frac{1}{2}\gamma_A\varepsilon_h \iff \begin{bmatrix}
        \hat{D}_A^{(h)}\hat{\psi} - \frac{1}{2}\hat{\gamma}_A\hat{\varphi} \\ \hat{D}_A^{(h)}\hat{\varphi} - \frac{1}{2}\hat{\gamma}_A\hat{\psi}
    \end{bmatrix} = \begin{bmatrix}
        0 \\ 0
    \end{bmatrix}.
\end{align}
Hence, $\hat{\varepsilon}_h^{(\pm)} = \hat{\psi} \pm \hat{\varphi}$ proves the claim for $D_A^{(h)}\varepsilon_h = \frac{1}{2}\gamma_A\varepsilon_h$.
\end{proof}
In summary, $\varepsilon_k$ is built from the most general parallel or real Killing spinors (of either $+$ or $-$ orientation) on the cross-section. The simply connected, compact manifolds admitting such spinors have been classified \cite{Wang1989, Bar1993}. However, there are also non-simply connected manifolds which satisfy the required conditions \cite{Wang1995}; some such examples will be considered in section \ref{sec:examples}. 

Meanwhile, observe that if $\hat{D}_A^{(h)}\hat{\varepsilon}_h$ equals $\frac{1}{2}\hat{\gamma}_A\hat{\varepsilon}_h$, $-\frac{1}{2}\hat{\gamma}_A\hat{\varepsilon}_h$ or $0$, then
\begin{align}
    \hat{D}_A^{(h)}\left(\hat{\varepsilon}_h^\dagger\hat{\varepsilon}_h\right) = 0 \,\,\,\mathrm{and}\,\,\, \hat{D}_A^{(h)}\left(-\mathrm{i}\hat{\varepsilon}_h^\dagger\hat{\gamma}^A\hat{\varepsilon}_h\right) = \mathrm{i}\hat{\varepsilon}_h^\dagger\hat{\gamma}_{AB}\hat{\varepsilon}_h, \,\, -\mathrm{i}\hat{\varepsilon}_h^{\dagger}\hat{\gamma}_{AB}\hat{\varepsilon}_h\,\,\mathrm{or}\,\,0.
\end{align}
Therefore $-\mathrm{i}\hat{\varepsilon}_h^\dagger\hat{\gamma}^A\hat{\varepsilon}_h$ is a Killing vector of $h$ and $\hat{\varepsilon}_h^\dagger\hat{\varepsilon}_h$ is constant on $S$, allowing the following constructions.
\begin{definition}[``Conserved quantities" on the cross-section]
    \label{def:crossSectionConserved}
    In an asymptotically AdS spacetime with cross-section, $(S, h)$, given a Killing vector, $\hat{k}$, of $h$, define a ``conserved quantity" by
    \begin{align}
        Q_{\hat{k}} &= \frac{n - 1}{16\pi}\int_{S}p_A\hat{k}^A\mathrm{d}A(h) = \frac{n - 1}{16\pi}\int_{S}f_{(n-1)0\alpha}\hat{k}^\alpha\mathrm{d}A(h).
    \end{align}
\end{definition}
\begin{theorem}
    \label{thm:crossSectionPositiveEnergy0}
    For an asymptotically AdS spacetime with cross-section, $(S, h)$, $c = 0$, satisfying the Einstein equation and the dominant energy condition, define $\hat{k}^{A} = -\mathrm{i}\hat{\psi}^{\dagger}\hat{\gamma}^A\hat{\psi}$, where $\hat{\psi}$ solves $\hat{D}_A^{(h)}\hat{\psi} = 0$. Without loss of generality, scale $\hat{\psi}$ such that $\hat{\psi}^\dagger\hat{\psi} = 1$. Then, 
    \begin{align}
        E + Q_{\hat{k}} \geq 0.
    \end{align}
\end{theorem}
\begin{proof}
    Choosing $\mathcal{A}_a = 0$ in theorem \ref{thm:generalPositveEnergy} reduces the non-negativity condition on $\mathbb{M}$ to the dominant energy condition on $T_{ab}$. Then, substituting the $c = 0$ case of theorem \ref{thm:crossSectionKilling} yields 
    \begin{align}
        0 &\leq \frac{n - 1}{2}\int_{S}p_M\varepsilon_h^\dagger \gamma^0\gamma^MP_1^-\varepsilon_h\mathrm{d}A(h).
    \end{align}
    From lemma \ref{thm:killingDecomposed} and the chosen gamma matrices,
    \begin{align}
        \varepsilon_h^\dagger\gamma^0\gamma^0P_1^-\varepsilon_h &= \begin{bmatrix}
            \hat{\psi}^\dagger & \hat{\varphi}^\dagger
        \end{bmatrix}^\mathrm{T}\frac{1}{2}\begin{bmatrix}
            I & \mathrm{i}I \\
            -\mathrm{i}I & I
        \end{bmatrix}\begin{bmatrix}
            \hat{\psi} \\ \hat{\varphi}
        \end{bmatrix} = \frac{1}{2}\left(\hat{\psi}^\dagger\hat{\psi} + \mathrm{i}\hat{\psi}^\dagger\hat{\varphi} - \mathrm{i}\hat{\varphi}^\dagger\hat{\psi} + \hat{\varphi}^\dagger\hat{\varphi}\right) \,\,\,\mathrm{and} \\
        \varepsilon_h^\dagger\gamma^0\gamma^AP_1^-\varepsilon_h &= \begin{bmatrix}
            \hat{\psi}^\dagger & \hat{\varphi}^\dagger
        \end{bmatrix}^\mathrm{T}\begin{bmatrix}
            I & 0 \\
            0 & -I
        \end{bmatrix}\begin{bmatrix}
            0 & \hat{\gamma}^A \\
            \hat{\gamma}^A & 0
        \end{bmatrix}\frac{1}{2}\begin{bmatrix}
            I & \mathrm{i}I \\
            -\mathrm{i}I & I
        \end{bmatrix}\begin{bmatrix}
            \hat{\psi} \\ \hat{\varphi}
        \end{bmatrix} \\
        &= \frac{1}{2}\left(-\mathrm{i}\hat{\psi}^\dagger\hat{\gamma}^A\hat{\psi} + \hat{\psi}^\dagger\hat{\gamma}^A\hat{\varphi} - \hat{\varphi}^\dagger\hat{\gamma}^A\hat{\psi} - \mathrm{i}\hat{\varphi}^\dagger\hat{\gamma}^A\hat{\varphi}\right).
    \end{align}
    Putting both parts together,
    \begin{align}
        0 &\leq \frac{n - 1}{4}\int_{S}\left(\hat{\psi}^\dagger -\mathrm{i}\hat{\varphi}^\dagger\right)(p_0I - \mathrm{i}p_A\hat{\gamma}^A)\left(\hat{\psi} + \mathrm{i}\hat{\varphi}\right)\mathrm{d}A(h).
    \end{align}
    Then, defining a new parallel spinor, say $\hat{\psi}^\prime = \hat{\psi} + \mathrm{i}\hat{\varphi}$, and scaling (by a constant) so that $\hat{\psi}^{\prime\dagger}\hat{\psi}^\prime = 1$ proves the claim.
\end{proof}
\begin{theorem}
    \label{thm:crossSectionPositiveEnergy3}
    For an asymptotically AdS spacetime with cross-section, $(S, h)$, $c = 1$, satisfying the Einstein equation and the dominant energy condition, let $\hat{\varepsilon}_h^{(\pm)}$ solve $\hat{D}_A^{(h)}\hat{\varepsilon}_h^{(\pm)} = \pm\frac{1}{2}\hat{\gamma}_A\hat{\varepsilon}_h^{(\pm)}$ and define $\hat{k}^{(\pm)A} = -\mathrm{i}\hat{\varepsilon}_h^{(\pm)\dagger}\hat{\gamma}^A\hat{\varepsilon}_h^{(\pm)}$. Without loss of generality, scale $\hat{\varepsilon}_h^{(\pm)}$ so that $\hat{\varepsilon}_h^{(\pm)\dagger}\hat{\varepsilon}_h^{(\pm)} = \hat{\delta}^{(\pm)}$, where $\hat{\delta}^{(\pm)} = 1$ if a non-trivial $\hat{\varepsilon}_h^{(\pm)}$ exists and $\hat{\delta}^{(\pm)} = 0$ otherwise\footnote{Note that the theorem only applies if at least one of $\hat{\delta}^{(+)}$ or $\hat{\delta}^{(-)}$ is non-zero.}. Then, if $h$ is not the round metric on a sphere,
    \begin{align}
        E\big(\hat{\delta}^{(+)} + \hat{\delta}^{(-)}\big) + Q_{\hat{k}^{(+)}} + Q_{\hat{k}^{(-)}} \geq 0.
    \end{align}
\end{theorem}
\begin{proof}
    Let $\hat{s}^{(\pm)} = \hat{\varepsilon}_h^{(\mp)\dagger}\hat{\varepsilon}_h^{(\pm)}$ and $\hat{\xi}^{(\pm)A} = \hat{\varepsilon}_h^{(\mp)\dagger}\hat{\gamma}^A\hat{\varepsilon}_h^{(\pm)}$. From the Killing spinor equation,
    \begin{align}
        \hat{D}_A^{(h)}\hat{s}^{(\pm)} &= \pm\hat{\xi}_A \,\,\,\mathrm{and}\,\,\, \hat{D}_A^{(h)}\hat{D}_B^{(h)}\hat{s}^{(\pm)} = \pm \hat{D}_A^{(h)}\hat{\xi}^{(\pm)}_B = -\delta_{AB}\hat{s}^{(\pm)}.
    \end{align}
    By Obata's theorem \cite{Obata1962}, the round sphere is the only compact, Riemannian manifold admitting non-trivial solutions to $\hat{D}_A^{(h)}\hat{D}_B^{(h)}\hat{s}^{(\pm)} = -\delta_{AB}\hat{s}^{(\pm)}$. As that case has been explicitly excluded in this theorem, it must be that $\hat{s}^{(\pm)} = 0$ and consequently $\hat{\xi}^{(\pm)A} = 0$.

    Now consider theorem \ref{thm:generalPositveEnergy} with $c = 1$, $\mathcal{A}_a = 0$ and the $\varepsilon_k$ from theorem \ref{thm:crossSectionKilling}. Since there is a leading $\mathrm{e}^{-r}$ factor in equation \ref{eq:qBoundary}, it suffices to retain only the $\mathrm{e}^{r/2}$ term in equation \ref{eq:killingSpinorCrossSection}, i.e. 
    \begin{align}
        \varepsilon_k \to \mathrm{e}^{r/2}P_1^-\left(\mathrm{e}^{\mathrm{i}\gamma^0t/2} - \mathrm{i}\mathrm{e}^{-\mathrm{i}\gamma^0t/2}\right)\varepsilon_h = (1 + \mathrm{i})\mathrm{e}^{r/2}P_1^-\left(\cos(t/2)I - \sin(t/2)\gamma^0\right)\varepsilon_h.
    \end{align}
    In summary, theorem \ref{thm:generalPositveEnergy} reduces to saying 
    \begin{align}
        0 &\leq \int_{S}p_M\varepsilon_h^\dagger\left(\cos(t/2)I - \sin(t/2)\gamma^0\right)\gamma^0\gamma^MP_1^-(\cos(t/2)I - \sin(t/2)\gamma^0)\varepsilon_h\,\mathrm{d}A(h).
    \end{align}
    Then, using the gamma matrices in equation \ref{eq:cliffordBlocks} and the $\varepsilon_h$ from lemma \ref{thm:killingDecomposed} eventually yields
    \begin{align}
        0 &\leq \int_{S}\left(\hat{\varepsilon}_h^{(+)\dagger} + \mathrm{ie}^{\mathrm{i}t}\hat{\varepsilon}_h^{(-)\dagger}\right)\left(p_0I - \mathrm{i}p_A\hat{\gamma}^A\right)\left(\hat{\varepsilon}_h^{(+)} - \mathrm{ie}^{-\mathrm{i}t}\hat{\varepsilon}_h^{(-)}\right)\mathrm{d}A(h).
    \end{align}
    Since $\hat{s}^{(\pm)} = 0$ and $\hat{\xi}^{(\pm)A} = 0$, this inequality further reduces to 
    \begin{align}
        0 &\leq \int_{S}\left(p_0\hat{\varepsilon}_h^{(+)\dagger}\hat{\varepsilon}_h^{(+)} + p_0\hat{\varepsilon}_h^{(-)\dagger}\hat{\varepsilon}_h^{(-)} - \mathrm{i}p_A\hat{\varepsilon}_h^{(+)\dagger}\hat{\gamma}^A\hat{\varepsilon}_h^{(+)} - \mathrm{i}p_A\hat{\varepsilon}_h^{(-)\dagger}\hat{\gamma}^A\hat{\varepsilon}_h^{(-)}\right)\mathrm{d}A(h),
    \end{align}
    which is just the desired inequality.
\end{proof}
To better understand the full physical significance of the inequalities just derived, it's worth considering the geometric invariance of $E$ and $Q_{\hat{k}}$. In particular, $\mathcal{I}$ is defined through a conformal compactification and choosing a different compactification could change $f_{(0)}$ or $f_{(n-1)}$ through a change to the Fefferman-Graham parameter, $r$. It will be easiest to see this by swapping $r$ for $z = \mathrm{e}^{-r}$. Then,
\begin{align}
    g = \frac{1}{z^2}\left(\mathrm{d}z\otimes\mathrm{d}z + (f_{(0)mn} + zf_{(1)mn} + \cdots)\mathrm{d}x^m\otimes\mathrm{d}x^n\right)
\end{align}
and the conformal factor for compactification is $\Omega = z$. The most general transformations preserving the Fefferman-Graham form are known as Penrose-Brown-Henneaux (PBH) transformations and were derived in \cite{Imbimbo2000}. In particular, the new coordinates are defined by 
\begin{align}
    z = z^\prime\mathrm{e}^{-\sigma(x^\prime)} \,\,\,\mathrm{and}\,\,\, x^m = x^{\prime m} + \xi^m(x^\prime, z^\prime)
    \label{eq:zPBH}
\end{align}
for an arbitrary function, $\sigma(x^\prime)$. Then, $\xi^m$ is completely determined; namely it takes the form,
\begin{align}
    \xi^m(x^\prime, z^\prime) &= \xi^m(x^\prime, 0) + \partial^\prime_n(\sigma)\int_0^{z^\prime}\zeta f^{mn}(x^\prime, \zeta)\mathrm{d}\zeta.
\end{align}
Practically, it suffices to consider $\sigma(x^\prime)$ (and hence $\xi^m(x^\prime, z^\prime)$) to be infinitesimal. Then, the metric changes by
\begin{align}
    \delta f &= \mathcal{L}_\xi f + 2\sigma f - \sigma z\partial_zf.
\end{align}
Most importantly for the present discussion, the boundary metric changes as
\begin{align}
    \delta f_{(0)} &= \mathcal{L}_{\xi|_{z = 0}}f_{(0)} + 2\sigma f_{(0)}.
    \label{eq:f0PBH}
\end{align}
Therefore, in general, the boundary metric can change in an arbitrary way. Furthermore, by expanding $\xi^m$ in powers of $z$, one can also find a complicated transformation law for $\delta f_{(n-1)}$ (separately for each value of $n$). In summary, the boundary data, $(f_{(0)}, f_{(n-1)})$, could change in a very complicated way and therefore affect the expressions for $E$ and $Q_{\hat{k}}$.

However, this is an issue that affects even the standard asymptotically AdS spacetimes (i.e. with round sphere cross-section) and doesn't seem to have been considered in any detail previously in the literature. Firstly, it's unclear how physically meaningful this question is. For example, even an asymptotically flat end can be written in a poor choice of coordinates, which then leads to strange results for ADM mass \cite{Chrusciel2010}. Therefore, it would seem reasonable to insist on a particular conformal class representative for $f_{(0)}$. Indeed, this is what is done in section 5.3 of \cite{Papadimitriou2005} when studying the first law of black hole mechanics in this context.

A more tractable, and physically well-motivated, problem regarding geometric invariance is to study the question considered in \cite{Chrusciel2001} for the energy in asymptotically AdS spacetimes\footnote{The analogous problems for for asymptotically hyperbolic Riemannian manifolds and the ADM energy are studied in \cite{Wang2001} and \cite{Bartnik1986, Chrusciel1986} respectively.}. Specifically, the task is to quantify how $E$ and $Q_{\hat{k}}$ change under asymptotic symmetries preserving both $\Sigma_t$ and the Fefferman-Graham gauge. 

An asymptotic symmetry is determined by a vector field, $\xi$, such that $(\mathcal{L}_\xi g)|_{r = \infty} = 0$. Meanwhile, following section 3 of \cite{Chrusciel2001}, to preserve $\Sigma_t$, choose $\xi = \xi^i\partial_i = \xi^1\partial_r + \xi^\alpha\partial_\alpha$. Then, 
\begin{align}
    (\mathcal{L}_\xi g)_{\mu\nu} &= \xi^1\partial_rg_{\mu\nu} + g_{1\nu}\partial_\mu\xi^1 + g_{\mu 1}\partial_\nu\xi^1 + \xi^\alpha\partial_\alpha g_{\mu\nu} + g_{\alpha\nu}\partial_\mu\xi^\alpha + g_{\mu\alpha}\partial_\nu\xi^\alpha.
\end{align}
The various components then reduce to
\begin{align}
    (\mathcal{L}_\xi g)_{11} &= 2\partial_r\xi^1, \\
    (\mathcal{L}_\xi g)_{1m} &= \partial_m\xi^1 + \mathrm{e}^{2r}f_{m\alpha}\partial_r\xi^\alpha \,\,\,\mathrm{and} \\
    (\mathcal{L}_\xi g)_{mn} &= \xi^1\partial_r(\mathrm{e}^{2r}f_{mn}) + \mathrm{e}^{2r}(\xi^\alpha\partial_\alpha f_{mn} + f_{\alpha n}\partial_m\xi^\alpha + f_{m\alpha}\partial_n\xi^\alpha).
\end{align}
$\xi$ should be a physical spacetime vector that extends to the boundary. Hence, in the spirit of the Fefferman-Graham expansion, let
\begin{align}
    \xi = \xi_{(0)} + \mathrm{e}^{-r}\xi_{(1)} + \mathrm{e}^{-2r}\xi_{(2)} + \cdots
\end{align}
for $\xi_{(k)}$ being $r$-independent. Then, for a spacetime that's asymptotically AdS with cross-section, $(S, h)$, one finds 
\begin{align}
    (\mathcal{L}_\xi g)_{11} &= -2(\mathrm{e}^{-r}\xi_{(1)} + \cdots) \\
    (\mathcal{L}_\xi g)_{1m} &= -\mathrm{e}^rf_{(0)m\alpha}\xi^\alpha_{(1)} + \partial_m(\xi^1_{(0)}) - 2f_{(0)m\alpha}\xi^\alpha_{(2)} + \cdots \,\,\,\mathrm{and} \\
    (\mathcal{L}_\xi g)_{mn} &= \mathrm{e}^{2r}\left(2\xi^1_{(0)}f_{(0)mn} + \xi^\alpha_{(0)}\partial_\alpha f_{(0)mn} + f_{(0)\alpha n}\partial_m\xi^\alpha_{(0)} + f_{(0)m\alpha}\partial_n\xi^\alpha_{(0)}\right) + 2\mathrm{e}^r\xi^1_{(1)}f_{(0)mn} \nonumber \\
    &\,\,\,\,\,\,\, + \xi^\alpha_{(0)}\partial_\alpha f_{(2)mn} + f_{(2)\alpha n}\partial_m\xi^\alpha_{(0)} + f_{(2)m\alpha}\partial_n\xi^\alpha_{(0)} \nonumber \\
    &\,\,\,\,\,\,\, + \xi^\alpha_{(2)}\partial_\alpha f_{(0)mn} + f_{(0)\alpha n}\partial_m\xi^\alpha_{(2)} + f_{(0)m\alpha}\partial_n\xi^\alpha_{(2)} + \cdots.
\end{align}
Therefore, $(\mathcal{L}_\xi g)_{11}|_{r = \infty}$ holds automatically. Next, $(\mathcal{L}_\xi g)_{1m}|_{r = \infty} = 0$ implies $f_{(0)m\alpha}\xi^\alpha_{(1)} = 0$ and $\partial_m(\xi^1_{(0)}) - 2f_{(0)m\alpha}\xi^\alpha_{(2)} = 0$. Since $f_{(0)mn}$ is invertible, the former condition says $\xi^\alpha_{(1)} = 0$. Meanwhile, the latter condition says $\partial_t\xi^1_{(0)} = 0$ and $2\xi_{(2)\alpha} = \partial_\alpha\xi^1_{(0)}$, where the index has been lowered using $h$. Finally, consider $(\mathcal{L}_\xi g)_{mn}|_{r = \infty}$. From the $\mathrm{e}^r$ component, it follows that $\xi^1_{(1)} = 0$, meaning $\xi_{(1)} = 0$ when combined with $\xi^\alpha_{(1)} = 0$ from earlier. From the $\mathrm{e}^{2r}$ and $r$-independent components, setting $(m, n) = (\alpha, \beta)$ implies
\begin{align}
    \xi^1_{(0)}h = -\frac{1}{2}\mathcal{L}_{\xi_{(0)}}h \,\,\,\mathrm{and}\,\,-\frac{c}{2}\mathcal{L}_{\xi_{(0)}}h + \mathcal{L}_{\xi_{(2)}}h = 0,
    \label{eq:xi0ConformalKilling}
\end{align}
where the vector field argument of the Lie derivatives only considers the part tangential to $S$. Re-writing the Lie derivatives in terms of $\hat{D}^{(h)}$ and applying $2\xi_{(2)\alpha} = \partial_\alpha\xi^1_{(0)}$ from above, it then follows that 
\begin{align}
    \hat{D}^{(h)}_\alpha\hat{D}^{(h)}_\beta\xi^1_{(0)} = -c\xi^1_{(0)}h_{\alpha\beta}.
    \label{eq:xi1Obata}
\end{align}
Since $(S, h)$ is assumed to not be the round sphere\footnote{The case of the round sphere will be considered later in section \ref{sec:sphere}.} in theorems \ref{thm:crossSectionPositiveEnergy0} and \ref{thm:crossSectionPositiveEnergy3}, by the Obata theorem \cite{Obata1962}, it must be that $\xi^1_{(0)} = 0$. Substituting back into equation \ref{eq:xi0ConformalKilling} implies $\mathcal{L}_{\xi_{(0)}}h = 0$, i.e. $\xi_{(0)}$ is a just a Killing vector of $h$.

Therefore $\delta f_{(0)\alpha\beta} = \delta h_{\alpha\beta} = 0$ (by virtue of being an asymptotic symmetry) and $\mathcal{L}_{\xi_{(0)}}h = 0$. Hence, $\sigma = 0$ by equation \ref{eq:f0PBH}. Being an asymptotic symmetry only determines the leading order behaviour of $\xi$. However, there is still the additional requirement of staying in Fefferman-Graham gauge. From equation \ref{eq:zPBH} with $\sigma = 0$, it follows that $\xi = \xi^\alpha_{(0)}\partial_\alpha$.

In summary, the only asymptotic symmetries in this context are coordinate transformations generated by the Killing vectors of $(S, h)$. These don't affect $f_{(0)}$ by definition. Moreover, since there is no $z$ dependence, they also don't affect the Taylor series split for $f_{mn}$; in particular $f_{(n-1)mn}$ just transforms as a tensor under coordinate transformations. Since the integrands of $E$ and $Q_{\hat{k}}$ are scalars, they remain invariant under these transformations.

\section{Example boundaries}
\label{sec:examples}
This section presents various examples of $(S, h)$ which satisfy the requirements in section \ref{sec:boundaryGeometry} and thus allow applications of the various versions of the positive energy theorem discussed there. These examples are not exhaustive. While the simply-connected, compact Riemannian manifolds with parallel or Killing spinors have been classified \cite{Wang1989, Bar1993}, the general list of non-simply-connected examples remains open \cite{Wang1995}. The examples below include both simply-connected and non-simply-connected $(S, h)$ and have been ordered by computational complexity. Throughout this section the assumed decay on $\mathbb{M}$ in definition \ref{def:setup} translates to a running assumption that $T^{0a}$ decays quicker than $O(\mathrm{e}^{-(n-1)r})$.
\subsection{\texorpdfstring{Squashed $S^7$}{}}
The simplest metric that admits Killing spinors is the round sphere. However, it comes with additional subtleties and is thereby postponed to subsection \ref{sec:sphere}. The simplest deformation to a sphere is squashing, yet this typically destroys all Killing spinors. A rare exception is a particular 7D squashed sphere \cite{Duff1986} with metric,
\begin{align}
    h &= \frac{9}{20}\bigg(\mathrm{d}a\otimes\mathrm{d}a + \frac{1}{4}\sin^2(a)b_i\otimes b_i + \frac{1}{20}(c_i+\cos(a)b_i)\otimes(c_i+\cos(a)b_i)\bigg),
    \label{eq:squashedS7}\\ 
    \mathrm{where}\,\,\, b_i &= \sigma_i - \Sigma_i, \,\,\, c_i = \sigma_i + \Sigma_i, \\
    \sigma_1 &= \cos(\psi)\mathrm{d}\theta + \sin(\psi)\sin(\theta)\mathrm{d}\phi, \,\,\, \sigma_2 = -\sin(\psi)\mathrm{d}\theta + \cos(\psi)\sin(\theta)\mathrm{d}\phi, \nonumber \\
    \sigma_3 &= \mathrm{d}\psi + \cos(\theta)\mathrm{d}\phi
\end{align}
and $\Sigma_i$ are defined identically to $\sigma_i$, but with $(\psi, \theta, \phi)$ replaced by some analogous coordinates, $(\psi^\prime, \theta^\prime, \phi^\prime)$. The squashing comes from the factor of $1/20$ in equation \ref{eq:squashedS7}. If that factor were also $1/4$, then $h$ would be the usual round sphere.

From \cite{Duff1986}, $h$ satisfies $R_{AB}^{(h)} = 6\delta_{AB}$ and admits exactly one linearly independent Killing spinor, namely 
\begin{align}
    \hat{\varepsilon}^{(+)}_h &= \frac{1}{\sqrt{2}}\begin{bmatrix}
        0 & 1 & -1 & 0 & 0 & 0 & 0 & 0
    \end{bmatrix}^{\mathrm{T}},
    \label{eq:s7KillingSpinor}
\end{align}
where the components are in the basis chosen for $\hat{\gamma}^A$ in \cite{Duff1986}. Furthermore, from \cite{Bar1993}, this metric only admits $\hat{\varepsilon}^{(-)}_h = 0$.
\begin{theorem}
    For spacetimes asymptotically AdS with squashed $S^7$ cross-section, if the Einstein equation and dominant energy condition hold, then $E \geq 0$.
\end{theorem}
\begin{proof}
    The proof is simply applying theorem \ref{thm:crossSectionPositiveEnergy3} with $
    \hat{\varepsilon}_h^{(-)} = 0$ and equation \ref{eq:s7KillingSpinor}. The former means $\hat{\delta}^{(-)} = 0$ and $Q_{\hat{k}^{(-)}} = 0$. Meanwhile, it can be checked that for this particular $\hat{\varepsilon}_h^{(+)}$, $\hat{\varepsilon}_h^{(+)\dagger}\hat{\gamma}^A\hat{\varepsilon}_h^{(+)} = 0$ for all $A$. Thus $Q_{\hat{k}^{(+)}}$ is also zero and theorem \ref{thm:crossSectionPositiveEnergy3} reduces to $E \geq 0$.
\end{proof}
\subsection{Torus}
\label{sec:torus}
The simplest cross-section with $c = 0$ is the torus, $\mathbb{T}^{n-2} = S^1\times\cdots\times S^1$, which has metric,
\begin{align}
    h &= \mathrm{d}\theta^2\otimes\mathrm{d}\theta^2 + \cdots + \mathrm{d}\theta^{n-1}\otimes\mathrm{d}\theta^{n-1},
\end{align}
where $\theta^2, \cdots, \theta^{n-1}$ are the angles on each $S^1$ factor. It follows immediately that $k_\alpha = \partial_{\theta^\alpha}$ is a Killing vector for each angle. As per definition \ref{def:crossSectionConserved}, define
\begin{align}
    \mathbb{J}_A &= \frac{n-1}{16\pi}\int_{\mathbb{T}^{n-2}}p_A\mathrm{d}^{n-2}\theta
\end{align}
as the associated ``conserved quantities."
\begin{theorem}
    \label{thm:torusPositiveEnergy}
    For spacetimes asymptotically AdS with torus cross-section, if the Einstein equation and dominant energy condition hold, then
    \begin{align}
        E \geq \sqrt{\mathbb{J}_A\mathbb{J}^A}.
    \end{align}
\end{theorem}
\begin{proof}
    $h$ being locally flat implies the parallel spinors are just constant spinors, $\hat{\psi} = \hat{\psi}_0$. Therefore, theorem \ref{thm:crossSectionPositiveEnergy0} can be re-written as
    \begin{align}
        0 &\leq E + Q_{\hat{k}} = E + \frac{n - 1}{16\pi}\int_{\mathbb{T}^{n-2}}p_A\hat{k}^A\mathrm{d}(h) = \hat{\psi}_0^\dagger\left(EI -\mathrm{i}\mathbb{J}_A\hat{\gamma}^A\right)\hat{\psi}_0.
    \end{align}
    Since $EI -\mathrm{i}\mathbb{J}_A\hat{\gamma}^A$ has eigenvalues, $E \pm \sqrt{\mathbb{J}_A\mathbb{J}^A}$, the result follows.
\end{proof}
This reproduces equation 4.14 of \cite{Chrusciel2006} - the only difference is the Fefferman-Graham expansion here is Lorentzian, as opposed to the Riemannian asymptotics on $\Sigma_t$ required by the initial data point of view adopted by \cite{Chrusciel2006}. 

\subsection{Sphere}
\label{sec:sphere}
With a round sphere cross-section, the problem reduces to finding a positive energy theorem for asymptotically AdS spacetimes, essentially reproducing the results of \cite{Cheng2005, Chrusciel2006, Wang2015}. In this case, it will be easier to apply theorem \ref{thm:generalPositveEnergy} directly, instead of the cross-section treatment in section \ref{sec:boundaryGeometry}; the two approaches will be connected later. Furthermore, it will be easier to represent AdS in the Poincar\'{e} ball model, rather than a Fefferman-Graham expansion. In particular,
\begin{align}
    \bar{g}_{\mathrm{AdS}} = -\left(\frac{1+\rho^2}{1-\rho^2}\right)^2\mathrm{d}t\otimes\mathrm{d}t + \frac{4}{(1-\rho^2)^2}\delta_{IJ}\mathrm{d}x^I\otimes\mathrm{d}x^J,
    \label{eq:adsDisk}
\end{align}
where $x^I$ are Cartesian coordinates inside the unit ball, $\rho = \sqrt{x^Ix_I}$ and the indices on $x^I$ are understood to be raised and lowered by $\delta$. The natural vielbein is
\begin{align}
    e_0^\prime &= \frac{1-\rho^2}{1+\rho^2}\partial_t\,\,\,\mathrm{and}\,\,\, e_I^\prime = \frac{1 - \rho^2}{2}\partial_I.
\end{align}
It can be checked that in this frame the Killing spinors are 
\begin{align}
    \varepsilon^\prime_k &= \frac{1}{\sqrt{1-\rho^2}}\left(I - \mathrm{i}x_I\gamma^I\right)\mathrm{e}^{\mathrm{i}\gamma^0t/2}\varepsilon_0,
    \label{eq:diskKilling}
\end{align}
for any constant spinor, $\varepsilon_0$. The coordinates and frame chosen are related to the Fefferman-Graham version of equation \ref{eq:ads} by
\begin{align}
    r &= \ln(R + \sqrt{1 + R^2}) - \ln(2),\,\,\,\,R = \frac{2\rho}{1-\rho^2}, 
    \label{eq:rToRho} \\
    e_0 &= \frac{\mathrm{e}^{-r}}{1 + \frac{1}{4}\mathrm{e}^{-2r}}\partial_t, \,\,\,
    e_1 = \partial_r, \,\,\, e_A = \frac{\mathrm{e}^{-r}}{1 - \frac{1}{4}\mathrm{e}^{-2r}}e^{(s)\alpha}_A\partial_\alpha, \\
    e^\prime_0 &= e_0, \,\,\, \mathrm{and\,\,\,} e^\prime_I = \hat{x}_Ie_1 + \rho\frac{\partial\theta^\alpha}{\partial x^I}e^{(s)A}_\alpha e_A,
\end{align}
where $\hat{x}^I$ are unit vectors, i.e. $x^I = \rho\hat{x}^I$, $\theta^\alpha$ are coordinates on $S^{n-2}$ and $h = s$ is the round metric on $S^{n-2}$. Hence, the local Lorentz transformation relating $e_a$ and $e^\prime_a$, i.e. $e^\prime_a = \Lambda\indices{^b_a}(x)e_b$, is given by
\begin{align}
    \Lambda\indices{^a_0} = \delta\indices{^a_0}\,\,\,\mathrm{and}\,\,\,\Lambda\indices{^a_I} = \delta\indices{^a_1}\hat{x}_I + \delta\indices{^a_A}\rho\frac{\partial\theta^\alpha}{\partial x^I}e^{(s)A}_\alpha.
    \label{eq:lorentzTransformation}
\end{align}

The conserved quantities are more complicated now than in definition \ref{def:crossSectionConserved}, although some comparison will be provided later. For motivating arguments on their definition, see the discussion in \cite{Chrusciel2006}. 
\begin{definition}[``Conserved" quantities on the sphere]
    \label{def:momenta}
    Define the linear momentum, angular momentum and centre of mass position as
    \begin{align}
        P_I &= \frac{n-1}{16\pi}\int_{S^{n-2}}\hat{f}_{(0)}^{mn}(f_{(n-1)mn} - \bar{f}_{(n-1)mn})\hat{x}_I\mathrm{d}A(s) = \frac{n-1}{16\pi}\int_{S^{n-2}}p_0\hat{x}_I\mathrm{d}A(s), \\
        J_{IJ} &= \frac{n-1}{16\pi}\int_{S^{n-2}}f_{(n-1)0\alpha}\bigg(\hat{x}_I\frac{\partial\theta^\alpha}{\partial x^J}\bigg|_{\rho = 1} - \hat{x}_J\frac{\partial\theta^\alpha}{\partial x^I}\bigg|_{\rho = 1}\bigg)\mathrm{d}A(s) \,\,\,\mathrm{and} \\
        K_I &= \frac{n-1}{16\pi}\int_{S^{n-2}}f_{(n-1)0\alpha}\frac{\partial\theta^\alpha}{\partial x^J}\bigg|_{\rho = 1}\left(\delta\indices{^J_I} - \hat{x}^J\hat{x}_I\right)\mathrm{d}A(s).
    \end{align}
\end{definition}
It's now possible to state the main result of this section.
\begin{theorem}
    In an asymptotically AdS spacetime (i.e. with round sphere cross-section), if the Einstein equation and the dominant energy condition hold, then
    \begin{align}
        EI - \mathrm{i}P_I\gamma^I + \frac{\mathrm{i}}{2}J_{IJ}\gamma^0\gamma^{IJ} + K_I\gamma^0\gamma^I
    \end{align}
    is a non-negative definite matrix.
    \label{thm:adsPositiveEnergy}
\end{theorem}
\begin{proof}
    Theorem \ref{thm:generalPositveEnergy} was derived in a vielbein adapted to the Fefferman-Graham expansion. However, since $\bar{\varepsilon}_k\gamma^a\varepsilon_k$ in theorem \ref{thm:generalPositveEnergy} transforms as a Lorentz vector, it suffices to apply $\bar{\varepsilon}_k\gamma^a\varepsilon_k = \Lambda\indices{^a_b}\bar{\varepsilon}_k^\prime\gamma^b\varepsilon_k^\prime$. From equation \ref{eq:diskKilling}, one finds 
    \begin{align}
        \bar{\varepsilon}_k^\prime\gamma^0\varepsilon^\prime_k &= \frac{1}{1-\rho^2}\varepsilon_0^\dagger\mathrm{e}^{-\mathrm{i}\gamma^0t/2}((1 + \rho^2)I - 2\mathrm{i}x_I\gamma^I)\mathrm{e}^{\mathrm{i}\gamma^0t/2}\varepsilon_0 \,\,\,\mathrm{and} \\
        \bar{\varepsilon}_k^\prime\gamma^I\varepsilon^\prime_k &= \frac{1}{1-\rho^2}\varepsilon_0^\dagger\mathrm{e}^{-\mathrm{i}\gamma^0t/2}((1 + \rho^2)\gamma^0\gamma^I - 2\mathrm{i}x_J\gamma^0\gamma^{IJ} - 2x^Ix_J\gamma^0\gamma^J)\mathrm{e}^{\mathrm{i}\gamma^0t/2}\varepsilon_0.
    \end{align}
    By equation \ref{eq:rToRho}, as $r \to \infty$, $\rho \to 1$ and $\frac{1}{1 - \rho^2} \to \frac{1}{2}\mathrm{e}^r$. Therefore, 
    \begin{align}
        \bar{\varepsilon}_k^\prime\gamma^0\varepsilon^\prime_k &\to \mathrm{e}^r\varepsilon_0^\dagger\mathrm{e}^{-\mathrm{i}\gamma^0t/2}(I - \mathrm{i}\hat{x}_I\gamma^I)\mathrm{e}^{\mathrm{i}\gamma^0t/2}\varepsilon_0 \,\,\,\mathrm{and} \\
        \bar{\varepsilon}_k^\prime\gamma^I\varepsilon^\prime_k &\to \mathrm{e}^r\varepsilon_0^\dagger\mathrm{e}^{-\mathrm{i}\gamma^0t/2}(\gamma^0\gamma^I - \mathrm{i}\hat{x}_J\gamma^0\gamma^{IJ} - \hat{x}^I\hat{x}_J\gamma^0\gamma^J)\mathrm{e}^{\mathrm{i}\gamma^0t/2}\varepsilon_0.
    \end{align}
    Substituting these into theorem \ref{thm:generalPositveEnergy}, applying equation \ref{eq:lorentzTransformation}, definition \ref{def:momenta} and
    \begin{align}
        p_A = e_A^{(f_{(0)})m}n_{(0)}^n(f_{(n-1)mn} - \bar{f}_{(n-1)mn}) = e_A^{(s)\alpha}f_{(n-1)0\alpha}
    \end{align}
    ultimately yields
    \begin{align}
        0 &\leq Q(\varepsilon) = 8\pi\varepsilon_0^\dagger\mathrm{e}^{-\mathrm{i}\gamma^0t/2}\bigg(EI - \mathrm{i}P_I\gamma^I + K_I\gamma^0\gamma^I + \frac{\mathrm{i}}{2}J_{IJ}\gamma^0\gamma^{IJ}\bigg)\mathrm{e}^{\mathrm{i}\gamma^0t/2}\varepsilon_0.
    \end{align}
    Since $\varepsilon_0$ is an arbitrary constant spinor and $\mathrm{e}^{\mathrm{i}\gamma^0t/2}$ is a constant, unitary matrix on $\Sigma_t$, this inequality is satisfied if and only if $EI - \mathrm{i}P_I\gamma^I + K_I\gamma^0\gamma^I + \frac{\mathrm{i}}{2}J_{IJ}\gamma^0\gamma^{IJ}$ is non-negative definite.
\end{proof}
The result produced is formally identical to \cite{Chrusciel2006} - in particular, see the unnamed equation directly above their equation 3.12 on their page 11 - the only difference being that \cite{Chrusciel2006} considered Riemannian asymptotics on an initial data slice. As such, they don't have any $t$ dependence in their Killing spinors. Adding this dependence was essentially the main result of \cite{Wang2015}. Theorem \ref{thm:adsPositiveEnergy} differs from their work in that the $t$ dependence has been extracted from the physical quantities; in contrast, their analogues of $P_I$ and $K_I$ have explicit $t$ dependence in their definition. The presentation in theorem \ref{thm:adsPositiveEnergy} also justifies some claims which were already used in section 6 of \cite{Rallabhandi2025}.

The eigenvalues of the matrix in theorem \ref{thm:adsPositiveEnergy} can't be found analytically in general and thus there is no concrete inequality, like in theorem \ref{thm:torusPositiveEnergy} for example. However, more progress can be made in specific examples. For example, if $n = 4$ and $K_I = P_I = 0$, then\footnote{As explained later, it is typically possible to choose a frame in which $P_I = 0$ by performing a conformal isometry of the round sphere boundary \cite{Chrusciel2006}.} the eigenvalues of $EI + \frac{\mathrm{i}}{2}J_{IJ}\gamma^0\gamma^{IJ}$ are $E \pm \sqrt{\frac{1}{2}J_{IJ}J^{IJ}} = E \pm |J|$, leading to the familiar 4D BPS inequality,
\begin{align}
    E \geq |J|.
\end{align}
Similarly, consider the case where $n = 5$ and $K_I = P_I = 0$. There are two independent rotation planes in 4 space dimensions. Without loss of generality, suppose the coordinates are oriented so that the angular momentum is in the 1-2 and 3-4 planes. Then, since the eigenvalues of $EI + \mathrm{i}J_{12}\gamma^0\gamma^1\gamma^2 + \mathrm{i}J_{34}\gamma^0\gamma^3\gamma^4$ are $E \pm J_{12} \pm J_{34}$ and $E \pm J_{12} \mp J_{34}$, one gets the familiar 5D BPS inequality,
\begin{align}
    E \geq |J_1| + |J_2|,
    \label{eq:5dVacuumBPS}
\end{align}
where $J_1$ and $J_2$ describe the independent rotations.

For a different type of example, suppose $K_I = 0$ and $J_{IJ} = 0$. Then, the eigenvalues of $EI - \mathrm{i}P_I\gamma^I$ are $E \pm \sqrt{P_IP^I} = E \pm |P|$, leading to $E \geq |P|$, which was the result in \cite{Gibbons1983}. 

Many more permutations can be chosen like this - see \cite{Chrusciel2006} for a detailed analysis.

Theorem \ref{thm:adsPositiveEnergy} can also be derived from the cross-section point of view as follows. To deal with $S^{n-2}$ for arbitrary $n$, the only practical coordinates are the ``nested spheres." In particular,
\begin{align}
    x^I &= \rho\begin{bmatrix}
        \cos(\theta_2) \\
        \sin(\theta_2)\cos(\theta_3) \\ \vdots \\ \sin(\theta_2)\cdots\sin(\theta_{n-2})\cos(\theta_{n-1}) \\
        \sin(\theta_2)\cdots\sin(\theta_{n-2})\sin(\theta_{n-1})
    \end{bmatrix} \,\,\,\mathrm{and} 
    \label{eq:nestedSphere} \\
    h &= \rho^2(\mathrm{d}\theta_2\otimes\mathrm{d}\theta_2 + \sin^2(\theta_2)\mathrm{d}\theta_3\otimes\mathrm{d}\theta_3 + \cdots + \sin^2(\theta_2)\cdots\sin^2(\theta_{n-2})\mathrm{d}\theta_{n-1}\otimes\mathrm{d}\theta_{n-1}).
\end{align}
The natural vielbein to use on the unit sphere is thus
\begin{align}
    e^{(s)2} = \mathrm{d}\theta_2, \,\, e^{(s)3} = \sin(\theta_2)\mathrm{d}\theta_3, \,\, \cdots, \,\, e^{(s)n-1} = \sin(\theta_2)\cdots\sin(\theta_{n-2})\mathrm{d}\theta_{n-1}.
\end{align}
In this frame, the most general solution to $D_A^{(h)}\varepsilon_h = \frac{1}{2}\gamma_A\varepsilon_h$ on the unit sphere is \cite{Lu1999}
\begin{align}
    \varepsilon_h &= \mathrm{e}^{\theta_2\gamma^2/2}\mathrm{e}^{\theta_3\gamma^3\gamma^2/2}\cdots\mathrm{e}^{\theta_{n-1}\gamma^{n-1}\gamma^{n-2}/2}\varepsilon_0
    \label{eq:sphereKillingSpinor}
\end{align}
for a constant spinor, $\varepsilon_0$. Now, defining $\varepsilon_k$ as per theorem \ref{thm:crossSectionKilling} and applying theorem \ref{thm:generalPositveEnergy} should give the same result as the method based on $\varepsilon_k^\prime$ used in theorem \ref{thm:adsPositiveEnergy}. However, seeing this requires performing a spinorial change of frame\footnote{The details of the calculations will be omitted as they're not particularly relevant to the main idea.}.

Define $\mathfrak{o(n - 1)}$ through the generators,
\begin{align}
    (M_{IJ})_{KL} = \delta_{IK}\delta_{JL} - \delta_{IL}\delta_{JK}.
\end{align}
Then, the Lorentz transformation matrix, $\Lambda\indices{^I_J}$, of equation \ref{eq:lorentzTransformation} can be checked to be
\begin{align}
    &\begin{bmatrix}
    \cos(\theta_2) & \sin(\theta_2)\cos(\theta_3) & \cdots & \sin(\theta_2)\cdots\sin(\theta_{n-2})\cos(\theta_{n-1}) &
    \sin(\theta_2)\cdots\sin(\theta_{n-2})\sin(\theta_{n-1}) \\
    -\sin(\theta_2) & \cos(\theta_2)\cos(\theta_3) & \cdots & \cos(\theta_2)\sin(\theta_3)\cdots\sin(\theta_{n-2})\cos(\theta_{n-1}) & \cos(\theta_2)\sin(\theta_3)\cdots\sin(\theta_{n-1}) \\
    0 & -\sin(\theta_3) & \cdots & \cos(\theta_3)\sin(\theta_4)\cdots\sin(\theta_{n-2})\cos(\theta_{n-1}) & \cos(\theta_3)\sin(\theta_4)\cdots\sin(\theta_{n-1}) \\
    \vdots & \vdots & \ddots & \vdots & \vdots \\
    0 & 0 & \cdots & -\sin(\theta_{n-1}) & \cos(\theta_{n-1})
    \end{bmatrix}
    \label{eq:lorentzMatrixFull}\\
    &=\mathrm{e}^{\theta_2M_{12}}\cdots\mathrm{e}^{\theta_{n-1}M_{n-2, n-1}}.
    \label{eq:lorentzMatrixExponential}
\end{align}
Therefore, upon the change of frame,
\begin{align}
    \varepsilon^\prime_k &= \frac{1}{\sqrt{1-\rho^2}}\left(I - \mathrm{i}x_I\gamma^I\right)\mathrm{e}^{\mathrm{i}\gamma^0t/2}\varepsilon_0 \\
    \implies \varepsilon_k &= \frac{1}{\sqrt{1-\rho^2}}\left(I - \mathrm{i}\Lambda\indices{^I_J}x^J\gamma_I\right)\mathrm{e}^{\mathrm{i}\gamma^0t/2}\mathrm{e}^{\theta_2\gamma^2\gamma^1/2}\cdots\mathrm{e}^{\theta_{n-1}\gamma^{n-1}\gamma^{n-2}/2}\varepsilon_0.
\end{align}
It can be shown that applying equation \ref{eq:lorentzMatrixFull} \& \ref{eq:rToRho}, redefining $\varepsilon_0$ as $\frac{1}{\sqrt{2}}(I - \gamma^1)\varepsilon_0$ and then simplifying reduces $\varepsilon_k$ to the form in theorem \ref{thm:crossSectionKilling}.

Theorem \ref{thm:adsPositiveEnergy} could have been derived by a method even more intrinsic to the cross-section by deploying $\hat{\varepsilon}_h^{(\pm)}$, like in theorem \ref{thm:crossSectionPositiveEnergy3}. From \cite{Lu1999},
\begin{align}
    \hat{\varepsilon}_h^{(\pm)} &= \mathrm{e}^{\pm\theta_2\hat{\gamma}^2/2}\mathrm{e}^{\theta_3\hat{\gamma}^3\hat{\gamma}^2/2}\cdots\mathrm{e}^{\theta_{n-1}\hat{\gamma}^{n-1}\hat{\gamma}^{n-2}/2}\hat{\varepsilon}_0.
\end{align}
While theorem \ref{thm:crossSectionPositiveEnergy3} doesn't apply to round sphere cross-sections, the proof still holds until
\begin{align}
    0 &\leq \int_{\Sigma_{t, \infty}}\left(\hat{\varepsilon}_h^{(+)\dagger} + \mathrm{ie}^{\mathrm{i}t}\hat{\varepsilon}_h^{(-)\dagger}\right)\left(p_0I - \mathrm{i}p_A\hat{\gamma}^A\right)\left(\hat{\varepsilon}_h^{(+)} - \mathrm{ie}^{-\mathrm{i}t}\hat{\varepsilon}_h^{(-)}\right)\mathrm{d}A(h),
\end{align}
which is effectively the result that would arise by combining appendix F and section 8 of \cite{Cheng2005}.

However, the individual terms in the integrand can be analysed further. When $h = s$, the cross-terms between $\hat{\varepsilon}_h^{(+)}$ and $\hat{\varepsilon}_h^{(-)}$ are no longer zero because the Obata theorem no longer applies. The $p_0(\hat{\varepsilon}_h^{(+)\dagger}\hat{\varepsilon}_h^{(+)} + \hat{\varepsilon}_h^{(-)\dagger}\hat{\varepsilon}_h^{(-)})$ terms produce energy as before, but this time the $p_0(\mathrm{ie}^{\mathrm{i}t}\hat{\varepsilon}_h^{(-)\dagger}\hat{\varepsilon}_h^{(+)} - \mathrm{ie}^{-\mathrm{i}t}\hat{\varepsilon}_h^{(+)\dagger}\hat{\varepsilon}_h^{(-)})$ terms produce linear momentum, the $-\mathrm{i}p_A(\hat{\varepsilon}_h^{(+)\dagger}\hat{\varepsilon}_h^{(+)} + \hat{\varepsilon}_h^{(-)\dagger}\hat{\varepsilon}_h^{(-)})$ terms produce angular momentum and the $p_A(\mathrm{e}^{\mathrm{i}t}\hat{\varepsilon}_h^{(-)\dagger}\hat{\gamma}^A\hat{\varepsilon}_h^{(+)} - \mathrm{e}^{-\mathrm{i}t}\hat{\varepsilon}_h^{(+)\dagger}\hat{\gamma}^A\hat{\varepsilon}_h^{(-)})$ terms produce the centre of mass position. This approach with $\hat{\varepsilon}_h^{(\pm)}$ wasn't followed earlier for computational ease, but also because the interpretation of the physical quantities is much easier in the approach taken earlier.

Finally, consider asymptotic symmetries for these metrics, following the discussion at the end of section \ref{sec:boundaryGeometry}. Again, the Obata theorem doesn't hold. Therefore, this time equation \ref{eq:xi1Obata} no longer implies $\xi^1_{(0)} = 0$. Consequently, equation \ref{eq:xi0ConformalKilling} implies $\xi^\alpha_{(0)}\partial_\alpha$ is a conformal Killing vector of the sphere with $\hat{D}_\alpha^{(h)}\xi^\alpha_{(0)} = -(n-2)\xi^1_{(0)}$. The net effect is the asymptotic symmetry group is equal to the conformal group of $S^{(n-2)}$, namely $O(n - 1, 1)$. This reproduces the result of \cite{Chrusciel2001}. Furthermore, \cite{Chrusciel2001} goes on to show that under these asymptotic symmetries, $E$ \& $P_I$ can be combined into an object, $P_a \equiv (E, P_I)$, transforming as a vector under $O(n - 1, 1)$ and likewise $K_I$ \& $J_{IJ}$ can be combined into an object, $J_{ab}$ with $J_{0I} = K_I$, transforming as a 2-form under $O(n - 1, 1)$. Therefore, $E$, $P_I$, $J_{IJ}$ and $K_I$ are not individually geometric invariants when the cross-section is a round sphere - for example, if $P_a$ is timelike, one could always perform an $O(n -1, 1)$ transformation to set $P_I$ to zero. As explained in \cite{Chrusciel2001, Chrusciel2006}, this is fine because various different combinations of $E$, $P_I$, $J_{IJ}$ and $K_I$ are invariants, e.g. $E^2 - P_IP^I$, and theorem \ref{thm:adsPositiveEnergy} remains true despite the action of asymptotic symmetries.

\subsection{\texorpdfstring{$L(p, 1)$}{}}
\label{sec:lens}
View $S^3$ as the level set, 
\begin{align}
    \{(z_1, z_2) \in \mathbb{C}^2 \,\,|\,\, |z_1|^2+|z_2|^2 = 1\}.
\end{align}
Then, the lens space, $L(p, 1)$, is defined as the quotient of $S^3$ by the $\mathbb{Z}_p$ action,
\begin{align}
    (z_1, z_2) \to (z_1\mathrm{e}^{2\pi\mathrm{i}/p}, z_2\mathrm{e}^{2\pi\mathrm{i}/p}).
\end{align}
It will be easiest to work in coordinates, $(\theta, \phi_1, \phi_2) \in [0, \pi]\times[0, 2\pi)\times[0, 2\pi)$, defined by
\begin{align}
    x_1 &= \cos(\theta/2)\cos(\phi_1), \,\,\, x_2 = \cos(\theta/2)\sin(\phi_1),\,\,\, z_1 = x_1 + \mathrm{i}x_2, \\
    x_3 &= \sin(\theta/2)\cos(\phi_2), \,\,\, x_4 = \sin(\theta/2)\sin(\phi_2) \,\,\,\mathrm{and}\,\,\, z_2 = x_3 + \mathrm{i}x_4,
\end{align}
where $x_I$ are the Cartesian coordinates from the standard embedding of $S^3$ in $\mathbb{R}^4$.

The metric on $L(p, 1)$ is locally isometric to the round metric on $S^3$; in the chosen coordinates it reads
\begin{align}
    h &= \frac{1}{4}\mathrm{d}\theta\otimes\mathrm{d}\theta + \cos^2(\theta/2)\mathrm{d}\phi_1\otimes\mathrm{d}\phi_1 + \sin^2(\theta/2)\mathrm{d}\phi_2\otimes\mathrm{d}\phi_2.
\end{align}
\begin{lemma}
    \label{thm:lensKillingVectors}
    The Killing vectors of $L(p, 1)$ are spanned by 
    \begin{align}
        k_1 &= \frac{\partial}{\partial\phi_1} + \frac{\partial}{\partial\phi_2}, \,\,\,
        k_2 = \frac{\partial}{\partial\phi_1} - \frac{\partial}{\partial\phi_2} , \\
        k_3 &= \tan(\theta/2)\sin(\phi_1 - \phi_2)\frac{\partial}{\partial\phi_1} + 2\cos(\phi_1 - \phi_2)\frac{\partial}{\partial\theta} + \cot(\theta/2)\sin(\phi_1 - \phi_2)\frac{\partial}{\partial\phi_2} \,\,\,\mathrm{and} \\
        k_4 &= \tan(\theta/2)\cos(\phi_1 - \phi_2)\frac{\partial}{\partial\phi_1} - 2\sin(\phi_1 - \phi_2)\frac{\partial}{\partial\theta} + \cot(\theta/2)\cos(\phi_1 - \phi_2)\frac{\partial}{\partial\phi_2}.
    \end{align}
\end{lemma}
\begin{proof}
    $L(p, 1)$ is locally isometric to $S^3$ and hence $L(p, 1)$'s Killing vectors are the subspace of $S^3$'s Killing vectors which survive the $\mathbb{Z}_p$ quotient. The Killing vectors of a sphere are known to be spanned by
    \begin{align}
        v_{IJ} = \bigg(\hat{x}_I\frac{\partial\theta^\alpha}{\partial x^J}\bigg|_{\rho = 1} - \hat{x}_J\frac{\partial\theta^\alpha}{\partial x^I}\bigg|_{\rho = 1}\bigg)\partial_\alpha.
    \end{align}
    Re-writing everything in terms of $(\theta, \phi_1, \phi_2)$ and taking the following invertible linear transformation of Killing vector basis, one finds
        \begin{align}
        k_1 &= v_{12} + v_{34} = \frac{\partial}{\partial\phi_1} + \frac{\partial}{\partial\phi_2}, \\
        k_2 &= v_{12} - v_{34} = \frac{\partial}{\partial\phi_1} - \frac{\partial}{\partial\phi_2}, \\
        k_3 &= v_{24} + v_{13} \\
        &= \tan(\theta/2)\sin(\phi_1 - \phi_2)\frac{\partial}{\partial\phi_1} + 2\cos(\phi_1 - \phi_2)\frac{\partial}{\partial\theta} + \cot(\theta/2)\sin(\phi_1 - \phi_2)\frac{\partial}{\partial\phi_2}, \\
        k_4 &= v_{14} - v_{23} \\
        &= \tan(\theta/2)\cos(\phi_1 - \phi_2)\frac{\partial}{\partial\phi_1} - 2\sin(\phi_1 - \phi_2)\frac{\partial}{\partial\theta} + \cot(\theta/2)\cos(\phi_1 - \phi_2)\frac{\partial}{\partial\phi_2}, \\
        k_5 &= v_{13} - v_{24} \\
        &= \tan(\theta/2)\sin(\phi_1 + \phi_2)\frac{\partial}{\partial\phi_1} + 2\cos(\phi_1 + \phi_2)\frac{\partial}{\partial\theta} - \cot(\theta/2)\sin(\phi_1 + \phi_2)\frac{\partial}{\partial\phi_2}\,\,\,\mathrm{and} \\
        k_6 &= v_{14} + v_{23} \\
        &= -\tan(\theta/2)\cos(\phi_1 + \phi_2)\frac{\partial}{\partial\phi_1} + 2\sin(\phi_1 + \phi_2)\frac{\partial}{\partial\theta} + \cot(\theta/2)\cos(\phi_1 + \phi_2)\frac{\partial}{\partial\phi_2}.
    \end{align}
    $k_1$, $k_2$, $k_3$ and $k_4$ are manifestly well-defined on $L(p, 1)$ while any linear combination involving $k_5$ and $k_6$ is not well-defined on $L(p, 1)$.
\end{proof}
As the lens space is 3D, it will be convenient to choose the cross-section gamma matrices as $\hat{\gamma}^2 = \mathrm{i}\sigma_1$, $\hat{\gamma}^3 = \mathrm{i}\sigma_2$ and $\hat{\gamma}^4 = \mathrm{i}\sigma_3$.
\begin{lemma}
    \label{thm:lensKillingSpinors}
    The most general solution to $\hat{D}_A^{(h)}\hat{\varepsilon}^{(\pm)}_h = \pm\frac{1}{2}\hat{\gamma}_A\hat{\varepsilon}^{(\pm)}_h$ on $L(p, 1)$ is
    \begin{align}
        \hat{\varepsilon}^{(+)}_h = 0 \,\,\,\mathrm{and}\,\,\,\hat{\varepsilon}^{(-)}_h &= \mathrm{e}^{-\mathrm{i}\theta\sigma_2/4}\mathrm{e}^{-\mathrm{i}(\phi_1 - \phi_2)\sigma_1/2}\hat{\varepsilon}_0,
    \end{align}
    where $\hat{\varepsilon}_0$ is an arbitrary constant spinor and the chosen vielbein on $L(p, 1)$ is
    \begin{align}
        e^2 &= \cos(\theta/2)\mathrm{d}\phi_1, \,\, e^3 = \frac{1}{2}\mathrm{d}\theta\,\, \mathrm{and}\,\, e^4 = \sin(\theta/2)\mathrm{d}\phi_2.
    \end{align}
\end{lemma}
\begin{proof}
    For this tetrad, the connection one-forms are
    \begin{align}
        \omega_{32}^{(h)} = \tan(\theta/2)e^2, \,\,\omega_{34}^{(h)} = -\cot(\theta/2)e^4\,\, \mathrm{and}\,\, \omega_{24}^{(h)} = 0.
    \end{align}
    Therefore, $0 = e^{(h)\alpha}_A\partial_\alpha\hat{\varepsilon}^{(+)}_h - \frac{1}{4}\omega_{BCA}^{(h)}\hat{\gamma}^{BC}\hat{\varepsilon}^{(+)}_h - \frac{1}{2}\hat{\gamma}_A\hat{\varepsilon}^{(+)}_h$ reduces to the three equations,
    \begin{align}
        \partial_\theta\hat{\varepsilon}^{(+)}_h &= \frac{\mathrm{i}}{4}\sigma_2\hat{\varepsilon}^{(+)}_h, \\
        \partial_{\phi_1}\hat{\varepsilon}^{(+)}_h &= \frac{\mathrm{i}}{2}\sin(\theta/2)\sigma_3\hat{\varepsilon}^{(+)}_h + \frac{\mathrm{i}}{2}\cos(\theta/2)\sigma_1\hat{\varepsilon}^{(+)}_h \,\,\,\mathrm{and} \\
        \partial_{\phi_2}\hat{\varepsilon}^{(+)}_h &= \frac{\mathrm{i}}{2}\cos(\theta/2)\sigma_1\hat{\varepsilon}^{(+)}_h + \frac{\mathrm{i}}{2}\sin(\theta/2)\sigma_3\hat{\varepsilon}^{(+)}_h.
    \end{align}
    The first equation immediately integrates to $\hat{\varepsilon}^{(+)}_h = \mathrm{e}^{\mathrm{i}\theta\sigma_2/4}\hat{\varepsilon}_\theta$, for a spinor, $\hat{\varepsilon}_\theta$, that doesn't depend on $\theta$. Using $\mathrm{e}^{\mathrm{i}\theta\sigma_2} = \cos(\theta)I + \mathrm{i}\sin(\theta)\sigma_2$, the other two equations simplify to
    \begin{align}
        \partial_{\phi_1}\hat{\varepsilon}_\theta &= \frac{\mathrm{i}}{2}\sigma_1\hat{\varepsilon}_\theta \,\,\,\mathrm{and}\,\,\, \partial_{\phi_2}\hat{\varepsilon}_\theta = \frac{\mathrm{i}}{2}\sigma_1\hat{\varepsilon}_\theta.
    \end{align}
    These equations simultaneously integrate to $\hat{\varepsilon}_\theta = \mathrm{e}^{\mathrm{i}(\phi_1 + \phi_2)\sigma_1/2}\hat{\varepsilon}_0$ for some constant spinor, $\hat{\varepsilon}_0$.

    Proceeding completely analogously for $\hat{\varepsilon}_h^{(-)}$ yields
    \begin{align}
        \hat{\varepsilon}^{(\pm)}_h &= \mathrm{e}^{\pm\mathrm{i}\theta\sigma_2/4}\mathrm{e}^{\pm\mathrm{i}(\phi_1 \pm \phi_2)\sigma_1/2}\hat{\varepsilon}_0^{(\pm)}.
    \end{align}
    While these spinors are well-defined on $S^3$, to be well-defined on $L(p, 1)$, they must be invariant under the $\mathbb{Z}_p$ quotient. Thus,
    \begin{align}
        \hat{\varepsilon}^{(\pm)}_h &\to \mathrm{e}^{\pm\mathrm{i}\theta\sigma_2/4}\mathrm{e}^{\pm\mathrm{i}((\phi_1 + 2\pi/p) \pm (\phi_2 + 2\pi/p))\sigma_1/2}\hat{\varepsilon}_0^{(\pm)}.
    \end{align}
    Choosing the $-$ in $\pm$ means the $2\pi/p$ factors immediately cancel and the spinor is left invariant. Hence, every $\hat{\varepsilon}_h^{(-)}$ of $S^3$ is also a $\hat{\varepsilon}_h^{(-)}$ of $L(p, 1)$. Meanwhile, in the $+$ case,
    \begin{align}
        \hat{\varepsilon}^{(+)}_h &\to \mathrm{e}^{\mathrm{i}\theta\sigma_2/4}\mathrm{e}^{\mathrm{i}(\phi_1 + \phi_2 + 4\pi/p)\sigma_1/2}\hat{\varepsilon}_0 
        = \mathrm{e}^{\mathrm{i}\theta\sigma_2/4}\mathrm{e}^{\mathrm{i}(\phi_1 + \phi_2)\sigma_1/2}\mathrm{e}^{2\pi\mathrm{i}\sigma_1/p}\hat{\varepsilon}_0.
    \end{align}
    Since $\mathrm{e}^{\mathrm{i}\theta\sigma_2/4}\mathrm{e}^{\mathrm{i}(\phi_1 + \phi_2)\sigma_1/2}$ is invertible, $\hat{\varepsilon}_h^{(+)}$ remains invariant if and only if $\mathrm{e}^{2\pi\mathrm{i}\sigma_1/p}\hat{\varepsilon}_0 = \hat{\varepsilon}_0$. However, $\mathrm{e}^{2\pi\mathrm{i}\sigma_1/p}$ has eigenvalues $\cos(2\pi/p) \pm \mathrm{i}\sin(2\pi/p)$, meaning $\hat{\varepsilon}_h^{(+)}$ is never invariant.
\end{proof}
One could consider $L(p, q)$ more generally. However, a similar procedure to the previous lemma shows non-trivial Killing spinors only exist for $q = \pm 1$. $L(p, -1)$ differs from $L(p, 1)$ only by a discrete isometry though. 
\begin{definition}[Angular momenta on $L(p, 1)$]
    \label{def:lensConserved}
    For each Killing vector, $k_I$, on $L(p, 1)$, define an ``angular momentum,"
    \begin{align}
        J_I &= \frac{1}{4\pi}\int_{L(p, 1)}f_{(4)0\alpha}\,k_I^\alpha\,\mathrm{d}A(h).
    \end{align}
\end{definition}
\noindent Note that these $J_I$s are identical to the ``conserved quantities" of definition \ref{def:crossSectionConserved}.
\begin{theorem}
    \label{thm:lensPositiveEnergy}
    For spacetimes asymptotically AdS with $L(p, 1)$ cross-section, if the Einstein equation and the dominant energy condition hold, then
    \begin{align}
        E &\geq \sqrt{J_2^2 + J_3^2 + J_4^2}.
    \end{align}
\end{theorem}
\noindent Observe that $J_1$ does not appear in the theorem. $J_1$ is distinguished because its generator, $k_1$, is the symmetry along the circle direction when $S^3$ is viewed as a Hopf fibration over $S^2$.
\begin{proof}
    Since $\hat{\varepsilon}_h^{(+)} = 0$ by lemma \ref{thm:lensKillingSpinors}, theorem \ref{thm:crossSectionPositiveEnergy3} reduces to $0 \leq E + Q_{\hat{k}^{(-)}}$. Furthermore, from lemma \ref{thm:lensKillingSpinors}, by direct evaluation one finds
    \begin{align}
        p_A\hat{k}^{(-)A} &= - \mathrm{i}p_A\hat{\varepsilon}_h^{(-)\dagger}\hat{\gamma}^A\hat{\varepsilon}_h^{(-)} \\
        &= \hat{\varepsilon}_0^\dagger\big((p_2\sin(\theta/2)\sin(\phi_1-\phi_2) + p_3\cos(\phi_1-\phi_2) + p_4\cos(\theta/2)\sin(\phi_1-\phi_2))\sigma_2 \nonumber \\
        &\,\,\,\,\,\,\, + (p_2\sin(\theta/2)\cos(\phi_1-\phi_2) -p_3\sin(\phi_1-\phi_2) + p_4\cos(\theta/2)\cos(\phi_1-\phi_2))\sigma_3 \nonumber \\
        &\,\,\,\,\,\,\, + (p_2\cos(\theta/2) - p_4\sin(\theta/2))\sigma_1\big)\hat{\varepsilon}_0.
    \end{align}
    From lemma \ref{thm:lensKillingVectors} and the choice of vielbein, this expression can be re-written as
    \begin{align}
        p_A\hat{k}^{(-)A} &=\hat{\varepsilon}_0^\dagger\left(f_{(4)0\alpha}k_2^\alpha\sigma_1 + f_{(4)0\alpha}k_3^\alpha\sigma_2 + f_{(4)0\alpha}k_4^\alpha\sigma_3\right)\hat{\varepsilon}_0.
    \end{align}
    Since $\hat{\varepsilon}_0$ and $\sigma_i$ are both constants, it follows that
    \begin{align}
        Q_{\hat{k}^{(-)}} &= 4\pi \hat{\varepsilon}_0^\dagger\left(J_2\sigma_1 + J_3\sigma_2 + J_4\sigma_3\right)\hat{\varepsilon}_0.
    \end{align}
    Then, given the normalisation of $\hat{\varepsilon}_0$, the positive energy inequality reduces to
    \begin{align}
        0 \leq \hat{\varepsilon}_0^\dagger\left(EI + J_2\sigma_1 + J_3\sigma_2 + J_4\sigma_3\right)\hat{\varepsilon}_0.
    \end{align}
    The eigenvalues of the matrix inbetween $\hat{\varepsilon}_0^\dagger$ and $\hat{\varepsilon}_0$ are $E \pm \sqrt{J_2^2 + J_3^2 + J_4^2}$ and hence the theorem follows.
\end{proof}

\subsection{Compatibility of spin structures}
The entire formalism explored in this work relies on the existence of background Killing spinors near $\mathcal{I}$. By construction, all the examples in this section satisfy that requirement.

However, there are more subtle issues which may arise. Let $\overbar{M}$ be an open neighbourhood of $\mathcal{I}$ and let $\bar{g}$ be the background metric. The main problem is that $(\overbar{M}, \bar{g})$ may admit multiple spin structures and the spin structure which admits a non-zero solution, $\varepsilon_k$, may not be compatible with the spin structure on $(M, g)$. The classic example of this is the AdS soliton \cite{Horowitz1998}, which is a vacuum solution of the Einstein equation with $\Lambda < 0$ and has metric, 
\begin{align}
    g &= -r^2\mathrm{d}t\otimes\mathrm{d}t + \frac{\mathrm{d}r\otimes\mathrm{d}r}{r^2(1 - 1/r^{n-1})} \nonumber \\
    &\,\,\,\,\,\,\, + r^2\left(\left(1 - \frac{1}{r^{n-1}}\right)\mathrm{d}\phi_2\otimes\mathrm{d}\phi_2 + \mathrm{d}\phi_3\otimes\mathrm{d}\phi_3 + \cdots \mathrm{d}\phi_{n-1}\otimes\mathrm{d}\phi_{n-1}\right).
\end{align}
The angles are identified by $\phi_2 \sim \phi_2 + \frac{4\pi}{n-1}$ and $\phi_3 \sim \phi_3 + a_3,\, \cdots, \phi_{n-1} \sim \phi_{n-1} + a_{n-1}$ for arbitrary $a_3,\, \cdots, a_{n-1}$. The manifold is therefore asymptotically AdS with toroidal cross-section. However, its energy is 
\begin{align}
    E &= -\frac{1}{4}a_3\cdots a_{n-1} < 0,
\end{align}
seemingly contradicting theorem \ref{thm:torusPositiveEnergy}. The reason for this is that a circle admits two inequivalent spin structures - periodic and anti-periodic. Thus, $(\overbar{M}, \bar{g})$ admits $2^{n-2}$ inequivalent spin structures. However, the constant spinor, $\hat{\psi}_0$, used in the proof of theorem \ref{thm:torusPositiveEnergy} is manifestly periodic around each circle of $T^{n-2}$. Therefore, for theorem \ref{thm:torusPositiveEnergy} to apply, this particular spin structure must extend beyond $(\overbar{M}, \bar{g})$ to all of $(M, g)$. It turns out the global topology of the AdS soliton, namely $\mathbb{R}^3\times T^{n-3}$, requires anti-periodicity around $\phi_2$. Hence, the AdS soliton's spin structure is incompatible with the spin structure required for theorem \ref{thm:torusPositiveEnergy}. There doesn't appear to be any spinorial remedy to this issue; however the modified positive energy conjecture of \cite{Horowitz1998} has recently been confirmed by \cite{Brendle2024} using very different means. 

Similar issues can arise for the lens space cross-section used in theorem \ref{thm:lensPositiveEnergy}. Firstly, $L(p, 1)$ admits a unique spin structure for odd $p$ and two inequivalent spin structures for even $p$ \cite{Franc1987}. It happens that there is an explicitly known soliton solution with $L(p, 1)$ asymptotics \cite{Clarkson2006}. However, it has negative energy, in violation of theorem \ref{thm:lensPositiveEnergy}. The situation is similar to the AdS soliton. This time, depending on $p$ the solution in \cite{Clarkson2006} either has no spin structure or a different spin structure to the one required by theorem \ref{thm:lensPositiveEnergy}.

\section{BPS inequalities}
\label{sec:BPS}
This section will discuss BPS inequalities for the bosonic sector of minimal, gauged supergravity in four and five dimensions. Although the final results are in line with previous partial results in the literature, a number of subtleties are discussed regarding magnetic fields\footnote{Some of the material in this section overlaps with \cite{McSharry2025}, which came out while this paper was in preparation. Specifically, their analysis of 4D minimal, gauged supergravity uses the same $\nabla_a$ and Lichnerowicz identity as the present work.}. The theories considered are described by the actions,
\begin{align}
    S = \frac{1}{16\pi}\int_M\left(R - 2\Lambda - F_{ab}F^{ab}\right)\mathrm{d}\mu(g)
\end{align}
in 4D and
\begin{align}
    S = \frac{1}{16\pi}\int_M\left(R - 2\Lambda -F_{ab}F^{ab} - \frac{2}{3\sqrt{3}}\varepsilon^{abcde}F_{ab}F_{cd}A_e\right)\mathrm{d}\mu(g)
\end{align}
in 5D. In both cases, $F = \mathrm{d}A$ is an electromagnetic field. Hence, the 4D theory is simply Einstein-Maxwell theory with a negative cosmological constant, while the 5D theory has an additional Chern-Simons term\footnote{In both cases, one can also add additional matter terms linearly coupled to $A_a$ and use them to introduce source charges and currents in the Maxwell equations. The results below go through in almost identical fashion with the main modification being to energy conditions.}. The equations of motion are
\begin{align}
    R_{ab} - \frac{1}{2}Rg_{ab} + \Lambda g_{ab} &= 2\bigg(F\indices{_a^c}F_{bc} - \frac{1}{4}g_{ab}F^{cd}F_{cd}\bigg), \\
    D_{b}F^{ba} &= 0 \,\,\,\&\,\,\,
    D_{[a}F_{bc]} = 0
\end{align}
in the 4D case and
\begin{align}
    R_{ab} - \frac{1}{2}Rg_{ab} + \Lambda g_{ab} &= 2\bigg(F\indices{_a^c}F_{bc} - \frac{1}{4}g_{ab}F^{cd}F_{cd}\bigg), \\
    D_{b}F^{ba} &= - \frac{1}{2\sqrt{3}}\varepsilon^{abcde}F_{bc}F_{de} \,\,\,\&\,\,\,
    D_{[a}F_{bc]} = 0
\end{align}
in the 5D case. These equations are assumed to hold throughout this section.

It will be very natural in what follows to split $F_{ab}$ into separate electric and magnetic fields.
\begin{definition}[Electric \& magnetic components and electric charge]
    Given a Maxwell field, $F_{ab}$, the electric and magnetic components with respect to $\Sigma_t$ will be defined as $E_I = F_{I0}$ and $F_{IJ}$ respectively. In Fefferman-Graham form, the electric charge is then 
    \begin{align}
        q_e = \frac{1}{4\pi}\int_{\Sigma_{t, \infty}}\star F = \frac{1}{4\pi}\int_{\Sigma_{t, \infty}}E_1\mathrm{d}A.
    \end{align}
\end{definition}
Unlike sections \ref{sec:boundaryGeometry} and \ref{sec:examples}, $\mathcal{A}_a$ can no longer be chosen as 0. Instead, $\mathcal{A}_a$ is chosen so that $\nabla_a\Psi$ is the gravitino transformation in the gauged supergravity. Therefore,
\begin{align}
    \mathcal{A}_a^{(4)} &= - \frac{1}{4}F_{bc}\gamma^{bc}\gamma_a + \mathrm{i}A_aI = \frac{1}{2}E_I\gamma^0\gamma^I\gamma_a - \frac{1}{4}F_{IJ}\gamma^{IJ}\gamma_a + \mathrm{i}A_a I
    \label{eq:4DSupergravityAMu}
\end{align}
for the 4D theory and
\begin{align}
    \mathcal{A}_a^{(5)} &= - \frac{1}{4\sqrt{3}}F_{bc}\gamma^{bc}\gamma_a - \frac{1}{2\sqrt{3}}F_{ab}\gamma^b + \mathrm{i}\sqrt{3}A_aI \\
    &= - \frac{1}{2\sqrt{3}}E_I\gamma^{I}\gamma^0\gamma_a - \frac{1}{4\sqrt{3}}F_{IJ}\gamma^{IJ}\gamma_a - \frac{1}{2\sqrt{3}}F_{ab}\gamma^b + \mathrm{i}\sqrt{3}A_aI
    \label{eq:5DSupergravityAMu}
\end{align}
for the 5D theory \cite{Freedman1977, Gunaydin1984}.
\begin{lemma}
    \label{thm:supergravityPositiveEnergy}
    Assume $E_I$, $F_{IJ}$ and $A_I$ decay as $O(\mathrm{e}^{-(n-2)r})$, $O(\mathrm{e}^{-(n-1)r})$ and $O(\mathrm{e}^{-(n-1)r})$ respectively\footnote{These assumptions are there to ensure convergence of the integrals to follow.}. Then, In the 4D theory, 
    theorem \ref{thm:generalPositveEnergy} implies
    \begin{align}
        Q(\varepsilon) &= \frac{3}{2}\int_{\Sigma_{t, \infty}}\mathrm{e}^{-r}p_M\bar{\varepsilon}_k\gamma^M\varepsilon_k\sqrt{\mathrm{det}(f_{(0)\alpha\beta})}\,\mathrm{d}^{n-2}x - 8\pi q_e\bar{\varepsilon}_k\varepsilon_k \nonumber \\
        &\,\,\,\,\,\,\, - 2\int_{\Sigma_{t, \infty}}F_{23}\varepsilon^\dagger_k\gamma^1\gamma^2\gamma^3\varepsilon_k\mathrm{d}A + 2\mathrm{i}\int_{\Sigma_{t, \infty}}A_A\varepsilon^\dagger_k\gamma^1\gamma^A\varepsilon_k\mathrm{d}A 
        \label{eq:qBulk4DBPS} \\
        &\geq 0
    \end{align}
    while in the 5D theory it implies
    \begin{align}
        Q(\varepsilon) &= 2\int_{\Sigma_{t, \infty}}\mathrm{e}^{-r}p_M\bar{\varepsilon}_k\gamma^M\varepsilon_k\sqrt{\mathrm{det}(f_{(0)\alpha\beta})}\,\mathrm{d}^{n-2}x -4\pi\sqrt{3}q_e\bar{\varepsilon}_k\varepsilon_k \nonumber \\
        &\,\,\,\,\,\,\, - \frac{\sqrt{3}}{2}\int_{\Sigma_{t, \infty}}F_{AB}\varepsilon^\dagger_k\gamma^1\gamma^{AB}\varepsilon_k\mathrm{d}A + 2\mathrm{i}\sqrt{3}\int_{\Sigma_{t, \infty}}A_A\varepsilon^\dagger_k\gamma^1\gamma^A\varepsilon_k\mathrm{d}A 
        \label{eq:qBulk5DBPS} \\
        &\geq 0.
    \end{align}
\end{lemma}
\begin{proof}
    To apply theorem \ref{thm:generalPositveEnergy}, it must first be checked that the assumptions in definition \ref{def:setup} hold. First, with the chosen $\mathcal{A}_a$,
    \begin{align}
        \gamma^{IJ}\mathcal{A}_J^{(4)} &= -E^I\gamma^0 - \frac{1}{2}F_{JK}\gamma^{IJK} + \mathrm{i}A_J\gamma^{IJ} \,\,\,\mathrm{and} \\
        \gamma^{IJ}\mathcal{A}_J^{(5)} &= -\frac{\sqrt{3}}{2}E^I\gamma^0 + \frac{\sqrt{3}}{4}F_{JK}\gamma^I\gamma^{JK} - \frac{\sqrt{3}}{2}F_{JK}\gamma^{IJ}\gamma^K + \mathrm{i}\sqrt{3}A_J\gamma^{IJ},
    \end{align}
    both of which are hermitian, as required. Next, with the present energy-momentum tensor,
    \begin{align}
        \mathbb{M} &= 4\pi T^{0a}\gamma_0\gamma_a + \gamma^{IJ}D_I\mathcal{A}_J + \frac{\mathrm{i}(n-2)}{2}(\gamma^I\mathcal{A}_I + \mathcal{A}_I^\dagger\gamma^I) - \mathcal{A}^\dagger_I\gamma^{IJ}\mathcal{A}_J \\
        &= \frac{1}{2}E^IE_II + \frac{1}{4}F^{IJ}F_{IJ}I - F_{IJ}E^J\gamma^0\gamma^I + \gamma^{IJ}D_I\mathcal{A}_J + \frac{\mathrm{i}(n-2)}{2}(\gamma^I\mathcal{A}_I + \mathcal{A}_I^\dagger\gamma^I) \nonumber \\
        &\,\,\,\,\,\,\, - \mathcal{A}^\dagger_I\gamma^{IJ}\mathcal{A}_J.
        \label{eq:mSupergravity}
    \end{align}
    Consider the 4D theory first. Using the equations of motion,
    \begin{align}
        \gamma^{IJ}D_I\mathcal{A}^{(4)}_J &= -D_I(E^I)\gamma^0 - \frac{1}{2}D_{[I}F_{JK]}\gamma^{IJK} + \mathrm{i}D_{[I}A_{J]}\gamma^{IJ} = 0 - 0 + \frac{\mathrm{i}}{2}F_{IJ}\gamma^{IJ}.
    \end{align}
    Next, one finds
    \begin{align}
        \gamma^I\mathcal{A}^{(4)}_I &= -\frac{1}{2}E_I\gamma^0\gamma^I - \frac{1}{4}F_{IJ}\gamma^{IJ} + \mathrm{i}A_I\gamma^I \,\,\,\mathrm{and\,\,therefore}
        \label{eq:gammaA4}\\
        \gamma^I\mathcal{A}^{(4)}_I + \mathcal{A}^{(4)\dagger}_I\gamma^I &= \gamma^I\mathcal{A}^{(4)}_I - (\gamma^I\mathcal{A}^{(4)}_I)^\dagger = -\frac{1}{2}F_{IJ}\gamma^{IJ}.
    \end{align}
    The most tedious to simplify is $\mathcal{A}^{(4)\dagger}_I\gamma^{IJ}\mathcal{A}_J^{(4)}$. After a somewhat lengthy calculation, albeit one that simply applies gamma matrix identities repeatedly, it can be shown
    \begin{align}
        \mathcal{A}^{(4)\dagger}_I\gamma^{IJ}\mathcal{A}_J^{(4)} &= \frac{1}{2}E^IE_II + E^IF_{IJ}\gamma^0\gamma^J + \frac{1}{4}F^{IJ}F_{IJ}I.
    \end{align}
    Substituting these results back into equation \ref{eq:mSupergravity} reveals $\mathbb{M} = 0$, which trivially satisfies all the required assumptions.

    Proceeding completely analogously, in the 5D theory the individual terms are
    \begin{align}
        \gamma^{IJ}D_I\mathcal{A}^{(5)}_J &= -\frac{1}{4}\varepsilon^{IJKL}F_{IJ}F_{KL}\gamma^0 + \frac{\mathrm{i}\sqrt{3}}{2}F_{IJ}\gamma^{IJ}, \\
        \gamma^I\mathcal{A}^{(5)}_I &= \frac{1}{2\sqrt{3}}E_I\gamma^I\gamma^0 - \frac{1}{2\sqrt{3}}F_{IJ}\gamma^{IJ} + \mathrm{i}\sqrt{3}A_I\gamma^I, 
        \label{eq:gammaA5}\\
        \gamma^I\mathcal{A}^{(5)}_I + \mathcal{A}_I^{(5)\dagger}\gamma^I &= -\frac{1}{\sqrt{3}}F_{IJ}\gamma^{IJ}, \\
        \mathcal{A}^{(5)\dagger}_I\gamma^{IJ}\mathcal{A}^{(5)}_J &= \frac{1}{2}E^IE_II + E^IF_{IJ}\gamma^0\gamma^J + \frac{1}{4}F_{IJ}F_{JK}\varepsilon^{IJKL}\gamma^1\gamma^2\gamma^3\gamma^4 + \frac{1}{4}F^{IJ}F_{IJ}I.
    \end{align}
    In 5D, there are two inequivalent, irreducible representations of the Clifford algebra; they have $\gamma^4 = \pm\gamma^0\gamma^1\gamma^2\gamma^3$ respectively. Choosing the representation in which $\gamma^4 = +\gamma^0\gamma^1\gamma^2\gamma^3$ once again leads to $\mathbb{M} = 0$. 

    The only remaining assumptions are on $\mathcal{A}$'s decay rate. These transfer to decay rates on the fields; they ensure all boundary integrals are convergent and are stronger than the decay rates required for results in appendix \ref{sec:invertDirac}.

    Having established that theorem \ref{thm:generalPositveEnergy} is valid in the present context, all that remains is to evaluate the $\mathcal{A}$ dependent boundary terms in the theorem. After some gamma matrix algebra, one finds
    \begin{align}
        \gamma^1\gamma^A\mathcal{A}^{(4)}_A &= -E_1\gamma^0 -F_{23}\gamma^1\gamma^2\gamma^3 + \mathrm{i}A_A\gamma^1\gamma^A \,\,\,\mathrm{and} \\
        \gamma^1\gamma^A\mathcal{A}^{(5)}_A &= -\frac{\sqrt{3}}{2}E_1\gamma^0 - \frac{\sqrt{3}}{4}F_{AB}\gamma^1\gamma^{AB} + \mathrm{i}\sqrt{3}A_A\gamma^1\gamma^A.
    \end{align}
    These are hermitian already, so $\gamma^1\gamma^A\mathcal{A}_A + \mathcal{A}_A^\dagger\gamma^A\gamma^1 = 2\gamma^1\gamma^A\mathcal{A}_A$. In theorem \ref{thm:generalPositveEnergy}, this matrix is inbetween $\varepsilon_k^\dagger$ and $\varepsilon_k$. In the case of the electric field term, that produces $E_1\bar{\varepsilon}_k\varepsilon_k$ as an integrand. However, the Killing spinor equation implies that $\overbar{D}_a(\bar{\varepsilon}_k\varepsilon_k) = 0$. Since $\bar{\varepsilon}_k\varepsilon_k$ is a scalar (and all derivatives act identically on a scalar), it must be that $\bar{\varepsilon}_k\varepsilon_k$ is constant. Hence, $\bar{\varepsilon}_k\varepsilon_k$ can be pulled out of the integral, leaving a term proportional to $\int_{\Sigma_{t, \infty}}E_1\mathrm{d}A = 4\pi q_e$ and thus the claimed result.
\end{proof}
\begin{corollary}
    \label{thm:magneticCancellation}
    If the extrinsic curvature, $K_{IJ}$, of $\Sigma_t$ is $o(\mathrm{e}^{-r})$ near $\Sigma_{t, \infty}$, then
    \begin{align}
        Q(\varepsilon) &= \frac{3}{2}\int_{\Sigma_{t, \infty}}\mathrm{e}^{-r}p_M\bar{\varepsilon}_k\gamma^M\varepsilon_k\sqrt{\mathrm{det}(f_{(0)\alpha\beta})}\,\mathrm{d}^{n-2}x - 8\pi q_e\bar{\varepsilon}_k\varepsilon_k \geq 0
    \end{align}
    in the 4D theory, while in the 5D theory 
    \begin{align}
        Q(\varepsilon) &= 2\int_{\Sigma_{t, \infty}}\mathrm{e}^{-r}p_M\bar{\varepsilon}_k\gamma^M\varepsilon_k\sqrt{\mathrm{det}(f_{(0)\alpha\beta})}\,\mathrm{d}^{n-2}x - 4\pi\sqrt{3}q_e\bar{\varepsilon}_k\varepsilon_k \geq 0.
    \end{align}
\end{corollary}
Note that in this corollary, the decay rates enforced on $F_{IJ}$ and $A_I$ can be relaxed to merely those required to ensure $\nabla_I\varepsilon_k \in L^2$, namely an $o(\mathrm{e}^{-(n-1)r/2})$ decay.
\begin{proof}
    The objective here is to show the magnetic and gauge field integrals in the theorem cancel. By construction, the $t$ coordinate is chosen such that the boundary has $f_{(0)0\alpha} = 0$. Therefore, to leading order, $e\indices{_I^\mu} = \delta\indices{^\mu_i}e^{(\sigma)i}_I$, where $\sigma$ is the metric on $\Sigma_t$. Thus,
    \begin{align}
        F_{JK} \to e^{(\sigma)j}_Je^{(\sigma)k}_KF_{jk} = e^{(\sigma)j}_Je^{(\sigma)k}_K(\partial_jA_k - \partial_kA_j) \to D_J^{(\sigma)}A_K - D_K^{(\sigma)}A_J.
    \end{align}
    The decay conditions assumed mean only the leading order contributions survive the integral. Let $r^I \equiv \partial_r$ be the normal to constant $r$ surfaces. Then,
    \begin{align}
        \frac{1}{2}\int_{\Sigma_{t, \infty}}r_IF_{JK}\varepsilon_k^\dagger\gamma^{IJK}\varepsilon_k\mathrm{d}A &= \int_{\Sigma_{t, \infty}}r_ID^{(\sigma)}_J(A_K)\varepsilon_k^\dagger\gamma^{IJK}\varepsilon_k\mathrm{d}A \\
        &= \int_{\Sigma_{t, \infty}}r_ID^{(\sigma)}_J(A_K\varepsilon_k^\dagger\gamma^{IJK}\varepsilon_k)\mathrm{d}A \nonumber \\
        &\,\,\,\,\,\, - \int_{\Sigma_{t, \infty}}r_IA_K\left(D_J^{(\sigma)}(\varepsilon_k)^\dagger\gamma^{IJK}\varepsilon_k + \varepsilon_k^\dagger\gamma^{IJK}D_J^{(\sigma)}\varepsilon_k\right)\mathrm{d}A.
        \label{eq:magneticByParts}
    \end{align}
    The first term in equation \ref{eq:magneticByParts} vanishes by applying Stokes' theorem in the same way as lemma \ref{thm:antisymmetricDerivative}. As for the other two terms, the Levi-Civita connection of $\Sigma_t$ and $M$ are related by 
    \begin{align}
        D_I^{(\sigma)}\varepsilon_k = D_I\varepsilon_k + \frac{1}{2}K_{IJ}\gamma^J\gamma^0\varepsilon_k
    \end{align}
    when acting on spinors. Since $K_{IJ}$ is assumed to decay quicker than $O(\mathrm{e}^{-r})$, $D_I^{(\sigma)}\varepsilon_k \to D_I\varepsilon_k$ to leading order. Then, since the metric also approaches the background to leading order, $D_I^{(\sigma)}\varepsilon_k \to -\frac{\mathrm{i}}{2}\gamma_I\varepsilon_k$. Hence, equation \ref{eq:magneticByParts} reduces to
    \begin{align}
        \frac{1}{2}\int_{\Sigma_{t, \infty}}r_IF_{JK}\varepsilon_k^\dagger\gamma^{IJK}\varepsilon_k\mathrm{d}A &= \frac{\mathrm{i}}{2}\int_{\Sigma_{t, \infty}}r_IA_K\left(\varepsilon_k^\dagger\gamma_J\gamma^{IJK}\varepsilon_k + \varepsilon_k^\dagger\gamma^{IJK}\gamma_J\varepsilon_k\right)\mathrm{d}A \\
        &= \mathrm{i}(n-3)\int_{\Sigma_{t, \infty}}A_A\varepsilon_k^\dagger\gamma^1\gamma^A\varepsilon_k\mathrm{d}A,
    \end{align}
    which means the magnetic and gauge field integrals do indeed cancel.
\end{proof}
The most subtle point in the previous proofs is the implicit assumption that $F = \mathrm{d}A$ everywhere in an open neighbourhood of $\Sigma_{t, \infty}$ in corollary \ref{thm:magneticCancellation}. If $\Sigma_{t, \infty}$ has a topology where $H^2_{\mathrm{dR}}$ is trivial, then this assumption is fine. However, there are many examples - including the most standard example of the $S^2$ cross-section - where this is not true. As such, it becomes impossible to incorporate magnetic charge into the discussion.

Even if $F$ wasn't assumed exact, magnetic charge doesn't arise as naturally from the equations as it does when $\Lambda = 0$. In the 4D Einstein-Maxwell theory \cite{Gibbons1982}, one still gets a term,
\begin{align}
    - 2\int_{\Sigma_{t, \infty}}F_{23}\varepsilon^\dagger_k\gamma^1\gamma^2\gamma^3\varepsilon_k\mathrm{d}A,
    \label{eq:magneticDecay}
\end{align}
in the analogue of $Q(\varepsilon)$. However, in that case $\varepsilon_k$ is just a constant, meaning $\varepsilon^\dagger_k\gamma^1\gamma^2\gamma^3\varepsilon_k$ can be pulled out of the integral, leaving the standard magnetic charge integral. But, in the present situation, $\varepsilon^\dagger_k\gamma^1\gamma^2\gamma^3\varepsilon_k$ is non-constant. Furthermore, it grows as $\mathrm{e}^r$. Hence, a convergent integral requires $F_{23}$ to decay as $O(\mathrm{e}^{-3r})$, which is faster than the $O(\mathrm{e}^{-2r})$ decay required to get non-zero magnetic charge. To some extent, this reflects the breakdown in electric-magnetic duality when $\Lambda \neq 0$.

Subtleties of $F$'s exactness also arise tacitly when solving the Dirac equation, $\gamma^I\nabla_I\varepsilon = 0$. The connection, $\nabla_a$, is constructed to be gauge covariant; in particular, under $A_a \to A_a + D_a\lambda$, if $\varepsilon \to \mathrm{e}^{\mathrm{i}\lambda}\varepsilon$, then $\nabla_a\varepsilon \to \mathrm{e}^{\mathrm{i}\lambda}\nabla_a\varepsilon$. However, on a manifold, changing from one coordinate patch to another also transforms $A_a$ in a formally identical way. Thus, $\varepsilon$ must also transform by a phase to keep $\nabla_a$ covariant. Hence, in principle one could get merely a $\mathrm{spin}^c$ structure, rather than a spin structure. This appears to be incompatible with Witten's method though because $\gamma^I\nabla_I\varepsilon = 0$ is solved subject to the boundary condition, $\varepsilon \to \varepsilon_k$, in which $\varepsilon_k$ is a true spinor, not a section of a non-trivial $\mathrm{spin}^c$ bundle. Therefore, imposing the required boundary condition breaks the gauge covariance. If $F$ were exact though, the issue is avoided. It remains open whether Witten's method can be adjusted in any way to accommodate $\mathrm{spin}^c$ structures.

The only previous work on classical BPS inequalities in these two gauged supergravity theories appears to be in \cite{London1995, Kostelecky1996, Wang2015b, Nozawa2014b}. \cite{London1995} doesn't consider magnetic fields in their positive energy theorem at all, so the issues discussed don't arise. 

Meanwhile, the result found here disagrees with \cite{Kostelecky1996}, who explicitly have magnetic charge in their equation 27. However, as explained in \cite{McSharry2025}, this is likely due to some error relating to the fact their analogue of $\nabla_a$ - see their equation 24 - omits the $\mathrm{i}A_aI$ term and is therefore not gauge covariant. Furthermore, the present results are consistent with subsequent work in \cite{Caldarelli1999, Hristov2011}, where the BPS limit is explicitly required to have zero magnetic charge. Another subtlety explained by \cite{Hristov2011} is that the 4D minimal supergravity has two distinct vacuum states, one of which has a non-zero magnetic charge determined by the cosmological constant. 

\cite{Wang2015b} uses the same connection as \cite{Kostelecky1996} and studies it in much greater detail. Despite using the same connection, their main result - their theorem 1.1 - differs from the main result of \cite{Kostelecky1996} - their equation 27. Since $\mathrm{i}A_aI$ is omitted, the analogue of $\mathbb{M}$ found in \cite{Wang2015b} is only non-negative definite when a modified dominant energy condition holds - their equation 1.2. However, their condition is very unnatural; for example, it can be seen from the constraint equations that in a purely electrovacuum spacetime with non-zero magnetic field, their modified dominant energy condition never holds. 

The closest paper to the present work is \cite{Nozawa2014b} - see their section 3.1 in particular. They have the same connection as used here and make similar observations about the decay rates in equation \ref{eq:magneticDecay}. However, their main result - their equation 22 - does not include the fourth term in lemma \ref{thm:supergravityPositiveEnergy} because they relied on \cite{Kostelecky1996} to get their result. A heuristic argument is given describing corrections to \cite{Kostelecky1996}, leading to the same conclusions as corollary \ref{thm:magneticCancellation} (for $S^2$ cross-section topology) and equation \ref{eq:4dEQJ} below, but the analytic steps required - including the Dirac equation analysis in appendix \ref{sec:invertDirac} - are not given in \cite{Nozawa2014b}.

Anyhow, bearing in mind all these subtleties, since the boundary geometries considered in section \ref{sec:examples} are time symmetric, they all satisfy the extra assumptions in corollary \ref{thm:magneticCancellation}. Hence, borrowing from the work there, the following results hold.
\begin{theorem}
    In an asymptotically AdS spacetime (i.e. with round sphere cross-section),
    \begin{align}
        &EI - \mathrm{i}P_I\gamma^I + \frac{\mathrm{i}}{2}J_{IJ}\gamma^0\gamma^{IJ} + K_I\gamma^0\gamma^I - q_e\gamma^0
        \label{eq:4dAdSBPSInequality}
    \end{align}
    is a non-negative definite matrix in the 4D theory and 
    \begin{align}
        &EI - \mathrm{i}P_I\gamma^I + \frac{\mathrm{i}}{2}J_{IJ}\gamma^0\gamma^{IJ} + K_I\gamma^0\gamma^I - \frac{\sqrt{3}}{2}q_e\gamma^0
        \label{eq:5dAdSBPSInequality}
    \end{align}
    is a non-negative definite matrix in the 5D theory.
    \label{thm:adsBPS}
\end{theorem}
\begin{proof}
    With the $\varepsilon_k$ chosen in section \ref{sec:sphere}, one finds $\bar{\varepsilon}_k\varepsilon_k = \varepsilon_0^\dagger\gamma^0\varepsilon_0$. Meanwhile, the first term in corollary \ref{thm:magneticCancellation} was already calculated in theorem \ref{thm:adsPositiveEnergy}. Borrowing from the calculation there, corollary \ref{thm:magneticCancellation} says
    \begin{align}
        0 &\leq 8\pi\varepsilon_0^\dagger\mathrm{e}^{-\mathrm{i}\gamma^0t/2}\bigg(EI - \mathrm{i}P_I\gamma^I + K_I\gamma^0\gamma^I + \frac{\mathrm{i}}{2}J_{IJ}\gamma^0\gamma^{IJ} - q_e\gamma^0\bigg)\mathrm{e}^{\mathrm{i}\gamma^0t/2}\varepsilon_0 \,\,\,\mathrm{in\,\,4D\,\,and} \\
        0 &\leq 8\pi\varepsilon_0^\dagger\mathrm{e}^{-\mathrm{i}\gamma^0t/2}\bigg(EI - \mathrm{i}P_I\gamma^I + K_I\gamma^0\gamma^I + \frac{\mathrm{i}}{2}J_{IJ}\gamma^0\gamma^{IJ} - \frac{\sqrt{3}}{2}q_e\gamma^0\bigg)\mathrm{e}^{\mathrm{i}\gamma^0t/2}\varepsilon_0 \,\,\,\mathrm{in\,\,5D.}
    \end{align}
    Since $t$ is constant on $\Sigma_{t, \infty}$ and $\varepsilon_0$ is an arbitrary constant spinor, the conclusion follows.
\end{proof}
Like theorem \ref{thm:adsPositiveEnergy}, the matrix in theorem \ref{thm:adsBPS} doesn't have closed form eigenvalues in general. Again, more progress can be made in specific cases. For example, if the $n = 4$ and $K_I = P_I = 0$, the eigenvalues are $E \pm q_e \pm |J|$ and $E \pm q_e \mp |J|$, leading to the familar BPS inequality,
\begin{align}
    E \geq |q_e| + |J|.
    \label{eq:4dEQJ}
\end{align}
Likewise, suppose $n = 5$ and $K_I = P_I = 0$. Like in the derivation of equation \ref{eq:5dVacuumBPS}, suppose the independent rotations are in the 1-2 and 3-4 planes. Then, from the eigenvalues of \newline $EI + \mathrm{i}J_1\gamma^0\gamma^1\gamma^2 + \mathrm{i}J_2\gamma^0\gamma^3\gamma^4 - \frac{\sqrt{3}}{2}q_e\gamma^0$, it follows that
\begin{align}
    E - \frac{\sqrt{3}}{2}q_e \geq |J_1 + J_2| \,\,\,\mathrm{and} \,\,\, E + \frac{\sqrt{3}}{2}q_e \geq |J_1 - J_2|.
    \label{eq:5dBPS}
\end{align}
The Chern-Simons term means the equations of motion are not invariant under $F \to -F$ in the 5D theory and this example illustrates that one must in fact keep track of the relative signs between the charge and angular momentum. Inequalities \ref{eq:5dBPS} agree with the BPS relations derived in section 3 of \cite{Cvetic2005} by assuming the supersymmetry algebra\footnote{Also note that inequalities \ref{eq:5dBPS} are not equivalent to the $E \geq |J_1| + |J_2| + \frac{\sqrt{3}}{2}|q_e|$ that \cite{Gutowski2004} claim results from modifying the results of \cite{London1995}; in fact, the intended modification would produce exactly inequalities \ref{eq:5dBPS}. Moreover, $E \geq |J_1| + |J_2| + \frac{\sqrt{3}}{2}|q_e|$ can only be concluded from a matrix with eight eigenvalues (to cover all possible combinations of $\pm$), which is impossible to achieve using the $4\times 4$ gamma matrices used in Witten's method when $n = 5$.}. Furthermore, these inequalities are saturated by the supersymetric solutions in \cite{Gutowski2004, Chong2005} and \cite{Klemm2001} respectively.

For a different type of example, suppose $n = 4$, $K_I = 0$ and $J_{IJ} = 0$; then the eigenvalues are $E \pm \sqrt{P_IP^I + q_e^2}$, which implies $\sqrt{E^2 - P^2} \geq |q_e|$.
\begin{theorem}
    \label{thm:lensBPS}
    For spacetimes asymptotically AdS with $L(p, 1)$ cross-section,
    \begin{align}
        E &\geq -\frac{\sqrt{3}}{2}q_e + \sqrt{J_2^2 + J_3^2 + J_4^2}
        \label{eq:5dLensInequality}
    \end{align}
\end{theorem}
\begin{proof}
    $\varepsilon_k$ can be chosen identically to section \ref{sec:lens}, namely
    \begin{align}
        \varepsilon_k &= \mathrm{e}^{r/2}P_1^-\left(\mathrm{e}^{\mathrm{i}\gamma^0t/2} - \mathrm{i}\mathrm{e}^{-\mathrm{i}\gamma^0t/2}\right)\varepsilon_h + \frac{1}{2}\mathrm{e}^{-r/2}P_1^+\left(\mathrm{e}^{\mathrm{i}\gamma^0t/2} + \mathrm{i}\mathrm{e}^{-\mathrm{i}\gamma^0t/2}\right)\varepsilon_h, \\
        \mathrm{with} \,\, \varepsilon_h &= \frac{1}{2}\begin{bmatrix}
            \hat{\varepsilon}_h^{(-)} \\
            - \hat{\varepsilon}_h^{(-)}
        \end{bmatrix} \,\,\, \mathrm{and} \,\,\,\hat{\varepsilon}_h^{(-)} = \mathrm{e}^{-\mathrm{i}\theta\sigma_2/4}\mathrm{e}^{-\mathrm{i}(\phi_1 - \phi_2)\sigma_1/2}\hat{\varepsilon}_0.
    \end{align}
    Then, from the proof of theorem \ref{thm:lensPositiveEnergy}, it follows the first term in corollary \ref{thm:magneticCancellation} is
    \begin{align}
        4\pi\hat{\varepsilon}_0^\dagger(EI + J_2\sigma_1 + J_3\sigma_2 + J_4\sigma_3)\hat{\varepsilon}_0.
    \end{align}
    Meanwhile, direct evaluation shows $\overline{\varepsilon}_k\varepsilon_k = -\hat{\varepsilon}_0\hat{\varepsilon}_0$ for the present Killing spinor. Hence, corollary \ref{thm:magneticCancellation} reduces to saying
    \begin{align}
        0 \leq 4\pi\hat{\varepsilon}_0^\dagger\left(EI + J_2\sigma_1 + J_3\sigma_2 + J_4\sigma_3 + \frac{\sqrt{3}}{2}q_eI\right)\hat{\varepsilon}_0.
    \end{align}
    As in theorem \ref{thm:lensPositiveEnergy}, the eigenvalues of the matrix in between $\hat{\varepsilon}_0^\dagger$ and $\hat{\varepsilon}_0$ prove the theorem.
\end{proof}
It turns out an explicit locally supersymmetric solution - see section 2.3.1 of \cite{Durgut2023} or appendix B of \cite{Lucietti2021} -  is known to 5D minimal gauged supergravity with $\mathbb{R}\times L(p, 1)$ conformal infinity. However, that solution satisfies a different BPS-like equation, namely
\begin{align}
    E = \frac{\sqrt{3}}{2}q_e + J_1.
\end{align}
The solution evades theorem \ref{thm:lensBPS} in much the same way as the AdS soliton evaded theorem \ref{thm:torusPositiveEnergy}. It turns out for even $p$, the spin structure required is the opposite to the spin structure required for the Killing spinors in theorem \ref{thm:lensBPS}. Meanwhile for odd $p$, the solution turns out to have no spin structure at all, but merely a $\mathrm{spin}^c$ structure.

Finally, the only other relevant boundary considered in section \ref{sec:examples} is the torus. In this case, theorem \ref{thm:torusPositiveEnergy} is unchanged by the electromagnetic fields in either theory because $\bar{\varepsilon}_k\varepsilon_k$ is simply zero for the required $\varepsilon_k$.

\section{Acknowledgements}
Most of all, I would like to thank my supervisor, James Lucietti, for guidance and many discussions throughout this project. I would also like to thank Jelle Hartong, Hari Kunduri, Harvey Reall and especially Piotr Chrusciel for comments and discussions based on earlier versions of this paper. Finally, I would like to thank the University of Edinburgh’s School of Mathematics for my PhD studentship funding.

\begin{appendices}
\section{Inverting the Dirac operator}
\label{sec:invertDirac}
This section is dedicated to showing that given $\varepsilon_k$, it is possible to solve $\gamma^I\nabla_I\varepsilon = 0$ with $\varepsilon \to \varepsilon_k$, as required for theorem \ref{thm:generalPositveEnergy}. The presentation here will be heavily based on \cite{Bartnik2005, ChruscielBartnik2003, Chrusciel2010}. As the procedure is largely well known, only the essential steps will be sketched.

Let $C_c^\infty$ denote the set of Dirac spinors of $(M, g)$ which are smooth and have compact support when restricted to $\Sigma_t$. Then, there exists a natural inner product on $C_c^\infty$ which is adapted to the application at hand.
\begin{lemma}
    An inner product can be defined on $C_c^\infty$ by
    \begin{align}
        \langle\psi, \chi\rangle_{C_c^\infty} &= \int_{\Sigma_t}\left((\nabla_I\psi)^\dagger\nabla^I\chi + \psi^\dagger\mathbb{M}\chi\right)\mathrm{d}V.
        \label{eq:Cc1InnerProduct}
    \end{align}
\end{lemma}
\begin{proof}
    Conjugate symmetry and linearity in the second argument are manifest. Since $\mathbb{M}$ is assumed to be non-negative definite, $\langle\psi, \psi\rangle_{C_c^\infty} \geq 0$ is also immediate. It only remains to check that $\langle\psi, \psi\rangle_{C_c^\infty} = 0 \implies \psi = 0$.

    From equation \ref{eq:Cc1InnerProduct}, $\langle\psi, \psi\rangle_{C_c^\infty} = 0 \implies \nabla_I\psi = 0$. Since $\psi$ is compactly supported, $\exists p \in \Sigma_t$ such that $\psi|_p = 0$. Now consider any other point, $q \in \Sigma_t$, and a curve, $\Gamma(s)$, joining $p$ and $q$. Let $s^I$ be tangent to the curve. Then, it follows that $s^I\nabla_I\psi = 0$ is an ODE along the curve subject to the initial value, $\psi|_p = 0$. Therefore $\psi = 0$ everywhere along the curve because $\psi$ is smooth and thus 1st order, linear, homogeneous ODEs have unique solution. Since $q$ was arbitrary, it must be that $\psi = 0$ everywhere on $\Sigma_t$.
\end{proof}
\begin{definition}[$\mathfrak{D}$]
    Define the Dirac operator, $\mathfrak{D} : C_c^\infty \to L^2$, by $\mathfrak{D} : \psi \mapsto \gamma^I\nabla_I\psi$.
    \label{def:G}
\end{definition}
\begin{lemma}
    For any antisymmetric spacetime tensor, $M^{ab}$, 
    \begin{align}
        n_aD_bM^{ba} &= \widetilde{D}_b(n_aM^{ba}),
    \end{align}
    where $\widetilde{D}$ is the induced covariant derivative on $\Sigma_t$ and $n_a$ is its unit normal.
    \label{thm:antisymmetricDerivative}
\end{lemma}
\begin{proof}
    See \cite{Cheng2005} or lemma 2.5 of \cite{Rallabhandi2025}.
\end{proof}
\begin{lemma}
    $\langle \psi, \chi\rangle_{C_c^\infty} = \langle\mathfrak{D}(\psi), \mathfrak{D}(\chi)\rangle_{L^2}$.
    \label{thm:licnerowiczInnerProduct}
\end{lemma}
\begin{proof}
    Apply theorem \ref{thm:lichnerowicz}, lemma \ref{thm:antisymmetricDerivative}, compact support of elements in $C_c^\infty$ and the polarisation identity for relating norms and inner products.
\end{proof}
\begin{definition}[$\mathcal{H}$]
    Define $\mathcal{H}$ to be the completion of $C_c^\infty$ under $\langle\cdot, \cdot\rangle_{C_c^\infty}$.
    \label{def:h}
\end{definition}
In general, the completion of a metric space has elements being equivalence classes of Cauchy sequences. However, in this case, the elements, $\psi \in \mathcal{H}$, can be represented as elements of the more familiar Sobolev space, $H^1_{\mathrm{loc}}$, as follows.

The key technical requirement is a weighted Poincar\'{e} inequality. For the chosen asymptotics, $\Sigma_t$ satisfies definition 9.9 of \cite{ChruscielBartnik2003} to be a weakly, asymptotically hyperboloidal end\footnote{The $x$ in their definition is $\mathrm{e}^{-r}$ here, their $h$ is the pullback of $f_{mn}$ to $\Sigma_t$ here and their $\mathcal{N}$ is $S$ here.}. Then, one can apply proposition 8.3 of \cite{ChruscielBartnik2003} to deduce $\exists \,w \in L_{\mathrm{loc}}^1$ such that
\begin{align}
    \int_{\Sigma_t}\psi^\dagger\psi w\,\mathrm{d}V \leq \int_{\Sigma_t}(\nabla_I\psi)^\dagger\nabla^I(\psi)\mathrm{d}V,
\end{align}
for any $\psi \in C_c^\infty$. Thus,
\begin{align}
    \int_{\Sigma_t}(\psi_m - \psi_n)^\dagger(\psi_m - \psi_n)w\,\mathrm{d}V &\leq \int_{\Sigma_t}\nabla_I(\psi_m - \psi_n)^\dagger\nabla^I(\psi_m - \psi_n)\mathrm{d}V \leq ||\psi_m - \psi_n||_{C_c^\infty}.
\end{align}
Therefore, for any Cauchy sequence, $\{\psi_m\}_{m = 0}^\infty \subseteq \mathcal{H}$, $\{\nabla_I\psi_m\}_{m = 0}^\infty$ and $\{\sqrt{w}\psi_m\}_{m = 0}^\infty$ are Cauchy in $L^2$. Since $w \in L^1_{\mathrm{loc}}$, it finally follows that $\psi \in H^1_{\mathrm{loc}}$.
\begin{lemma}
    $\mathfrak{D}$ extends to a continuous (i.e. bounded) linear operator from $\mathcal{H}$ to $L^2$ such that $\langle\psi, \chi\rangle_{\mathcal{H}} = \langle \mathfrak{D}(\psi), \mathfrak{D}(\chi)\rangle_{L^2}$.
    \label{thm:GExtension}
\end{lemma}
\begin{proof}
    This can be proven identically to lemma 3.6 of \cite{Rallabhandi2025}. In particular, given a Cauchy sequence, $\{\psi_m\}_{m = 0}^\infty \subseteq C_c^\infty$, with limit, $\psi \in \mathcal{H}$, $\mathfrak{D}(\psi)$ is defined to be $\lim_{m \to \infty}\mathfrak{D}(\psi_m) \in L^2$.
\end{proof}
\begin{theorem}
    $\mathfrak{D}$ is a continuous, linear isomorphism between $\mathcal{H}$ and $L^2$.
    \label{thm:gIsomorphism}
\end{theorem}
Most saliently, the theorem implies $(\gamma^I\nabla_I)^{-1} : L^2 \to \mathcal{H}$ exists.
\begin{proof}
    Linearity is by construction and continuity has already been shown by lemma \ref{thm:GExtension}. Next, suppose $\mathfrak{D}(\psi) = 0$. Then, by lemma \ref{thm:GExtension}, $0 = ||\mathfrak{D}(\psi)||_{L^2} = ||\psi||_{\mathcal{H}} \implies \psi = 0$ and therefore $\mathfrak{D}$ is injective. It remains to prove surjectivity.
    
    The surjectivity proof follows an index theory argument based on the analysis in \cite{Reall2024, McSharry2025}. Let $\hat{\nabla}_a = D_a + \frac{\mathrm{i}}{2}\gamma_a$, i.e. don't include the $\mathcal{A}_a$ term. Therefore, $\mathfrak{D} - \hat{\mathfrak{D}} = \gamma^I\mathcal{A}_I$. If $\gamma^I\mathcal{A}_I$ were a compact operator, then the index (dimension of kernel minus dimension of cokernel) of $\mathfrak{D}$ and $\hat{\mathfrak{D}}$ would coincide. In particular, if $\hat{\mathfrak{D}}$ was invertible, then it would follow that $\mathrm{index}(\mathfrak{D}) = 0$. Since $\mathfrak{D}$ has already been shown to have trivial kernel above, it would follow that $\mathfrak{D}$ is invertible.

Even though $\mathcal{A}_a$ is just a (sufficiently regular) matrix valued function, it's not clear in general whether $\gamma^I\mathcal{A}_I$ is actually compact. This would be the case though if the underlying space, $\Sigma_t$, were itself compact - a fact leveraged in \cite{Reall2024, McSharry2025}.

Similarly, given an asymptotic end of the form $\mathbb{R} \times S$, consider the compact subset\footnote{Here, $r$ is the Fefferman-Graham coordinate, but in general it could be any coordinate for the $\mathbb{R}$ factor in $\mathbb{R}\times S$ such that the conformal boundary is at $r = \infty$.}, $\Sigma(r_0) = \Sigma\backslash\{r > r_0\}$. With the appropriate boundary conditions for spinors on $\partial\Sigma(r_0) = \{r = r_0\}$, it's known - e.g. from section 3 of \cite{Rallabhandi2025} or appendix B of \cite{McSharry2025} - that $\hat{\mathfrak{D}}$ is invertible on the compact set. Therefore, by the index theory argument above, $\mathfrak{D}$ is invertible on $\Sigma(r_0)$.

Since $C_c^\infty$ is dense in $L^2$, for any $\Phi \in L^2$, $\exists$ a Cauchy sequence, $\{\phi_m\}_{m = 0}^\infty \subseteq C_c^\infty$, such that $\lim_{m \to \infty}\phi_m = \Phi$ in $L^2$. For each $\phi_m$ choose $r_m$ so that $\mathrm{supp}(\phi_m) \subseteq \mathrm{int}(\Sigma(r_m))$, i.e. choose $r_m$ large enough. Then, the choice of boundary conditions on $\partial\Sigma(r_m)$ doesn't matter and $\exists \psi_m$ such that $\mathfrak{D}(\psi_m) = \phi_m$. 
Theorem 6.4 in \cite{Bartnik2005, ChruscielBartnik2003} shows $\mathfrak{D}$ has an ``elliptic regularity" property on $\Sigma(r_m)$ meaning $\phi_m \in C_c^\infty$ and the metric's smoothness imply $\psi_m \in C_c^\infty$\footnote{Theorem 6.4 as stated in \cite{Bartnik2005, ChruscielBartnik2003} only says $\psi_m \in H^1_{\mathrm{loc}}(\Sigma(r_m))$. However, this is because of the very low regularity assumed in \cite{Bartnik2005, ChruscielBartnik2003} for the metric coefficients and the analogue of $\phi_m$. If one assumes additional regularity - in this case $g$'s smoothness and $\phi_m \in C_c^\infty$ - then $\psi_m$ inherits this additional regularity by the same proof. I'd like to thank Piotr Chrusciel for confirming that this is indeed the case.}. 

Furthermore, $\{\psi_m\}_{m = 0}^\infty \subseteq C_c^\infty$ is a Cauchy sequence by lemma \ref{thm:licnerowiczInnerProduct}. By theorem \ref{thm:GExtension}, the limit, $\lim_{m \to \infty}\psi_m = \psi \in \mathcal{H}$ satisfies
\begin{align}
    \mathfrak{D}(\psi) = \lim_{m \to \infty}\mathfrak{D}(\psi_m) = \lim_{m \to \infty}\phi_m = \Phi,
\end{align}
thereby proving $\mathfrak{D}$ is surjective.
\end{proof}
The index theory proof of surjectivity presented works for any regular $\mathcal{A}_I$. This is different to the Lichnerowicz identity based proofs which have previously appeared in the literature. These proofs require additional assumptions on $\mathcal{A}_I$ which the index theory argument does not. For example, the standard surjectivity argument presented in \cite{Bartnik2005, ChruscielBartnik2003} proceeds as follows.

    Let $\theta$ be an arbitrary element of $L^2$ and define $F_\theta : \mathcal{H} \to \mathbb{C}$ by
    \begin{align}
        F_\theta(\psi) = \langle\theta, \mathfrak{D}(\psi)\rangle_{L^2}.
        \label{eq:FTheta}
    \end{align}
    $F_\theta$ is manifestly linear. It is also continuous/bounded because the Cauchy-Schwarz inequality and lemma \ref{thm:GExtension} imply $|F_\theta(\Psi)| = |\langle\theta, \mathfrak{D}(\Psi)\rangle_{L^2}| \leq ||\theta||_{L^2}||\mathfrak{D}(\Psi)||_{L^2} = ||\theta||_{L^2}||\Psi||_{\mathcal{H}}$. Therefore, by the Riesz representation theorem, $\exists\varphi \in \mathcal{H}$ such that $F_\theta(\psi) = \langle\varphi , \psi\rangle_{\mathcal{H}}$. Then, lemma \ref{thm:GExtension} and equation \ref{eq:FTheta} imply
    \begin{align}
       \langle \Phi, \mathfrak{D}(\psi)\rangle_{L^2} = 0 \,\,\forall \psi \in \mathcal{H},\,\, \mathrm{where}\,\, \Phi = \theta - \mathfrak{D}(\varphi).
       \label{eq:PhiProperty}
    \end{align}
    Using lemma \ref{thm:antisymmetricDerivative}, one can perform a formal integration by parts to get
    \begin{align}
        0 &= \int_{\Sigma_t}\psi^\dagger\left(\gamma^{I}D_I(\Phi) + \frac{\mathrm{i}}{2}(n-1)\Phi - \mathcal{A}_I^\dagger\gamma^I\Phi\right)\mathrm{d}V.
        \label{eq:formalIntegration}
    \end{align}
    Suppose $\exists\, \tilde{\mathcal{A}}_a$ such that $(\gamma^I\mathcal{A}_I)^\dagger = \gamma^I\tilde{\mathcal{A}}_I$. Define a new connection on spinors, $\widetilde{\nabla}_a$, by $\widetilde{\nabla}_a = D_a -\frac{\mathrm{i}}{2}\gamma_a + \tilde{\mathcal{A}}_a$ where  and a new Dirac operator, $\widetilde{\mathfrak{D}} = \gamma^I\widetilde{\nabla}_I$. Then, equation \ref{eq:formalIntegration} can be re-written as
    \begin{align}
        0 &= \int_{\Sigma_t}\psi^\dagger\widetilde{\mathfrak{D}}(\Phi)\mathrm{d}V.
    \end{align}
    Since $\psi$ could be any compactly supported spinor, it follows that $\Phi$ is a weak solution to $\widetilde{\mathfrak{D}}(\Phi) = 0$. Thus, the surjectivity proof reduces to showing $\widetilde{\mathfrak{D}}$ has trivial kernel. Continuing with the new connection, suppose $\gamma^{IJ}\tilde{\mathcal{A}}_J$ is hermitian and 
    \begin{align}
        \widetilde{\mathbb{M}} = 4\pi T^{0a}\gamma_0\gamma_a + \gamma^{IJ}D_I\tilde{\mathcal{A}}_J - \frac{\mathrm{i}(n-2)}{2}(\gamma^I\tilde{\mathcal{A}}_I + \tilde{\mathcal{A}}_I^\dagger\gamma^I) - \tilde{\mathcal{A}}^\dagger_I\gamma^{IJ}\tilde{\mathcal{A}}_J
    \end{align}
    is non-negative definite. Then, following the same steps as theorem \ref{thm:lichnerowicz}, one finds if $\psi \in C_c^\infty$, then
    \begin{align}
        ||\widetilde{\mathfrak{D}}(\psi)||_{L^2}^2 &= \int_{\Sigma_t}\left((\widetilde{\nabla}_I\psi)^\dagger\widetilde{\nabla}^I\psi + \psi^\dagger\widetilde{\mathbb{M}}\psi\right)\mathrm{d}V.
        \label{eq:adjointLichnerowicz}
    \end{align}
    As $\widetilde{\mathbb{M}}$ is assumed to be non-negative definite, this is formally identical to the original Lichnerowicz identity. Then, one can proceed analogously to the analysis of $\mathfrak{D}$ itself by defining an inner product analogous to equation \ref{eq:Cc1InnerProduct} and thereby concluding that $\widetilde{\mathfrak{D}}$ is injective. Hence, $\widetilde{\mathfrak{D}}(\Phi) = 0$ implies $\Phi = \theta - \mathfrak{D}(\varphi) = 0$ and therefore $\mathfrak{D}$ is surjective.

While this proof is more in the spirit of the rest of the material because it is built on a Lichnerowicz identity, it has the disadvantage that there may be choices of $\mathcal{A}_a$ for which no simple $\tilde{\mathcal{A}}_a$ exists. Furthermore, even if $\tilde{\mathcal{A}}_a$ exists, it may be that $\widetilde{\mathbb{M}}$ is not non-negative definite. 

In fact, this is exactly the situation for the electromagnetic examples of section \ref{sec:BPS}. For the 4D and 5D theories studied there,
\begin{align}
    \mathcal{A}_a^{(4)} &= - \frac{1}{4}F_{bc}\gamma^{bc}\gamma_a + \mathrm{i}A_aI = \frac{1}{2}E_I\gamma^0\gamma^I\gamma_a - \frac{1}{4}F_{IJ}\gamma^{IJ}\gamma_a + \mathrm{i}A_a I \,\,\,\mathrm{and} \\
    \mathcal{A}_a^{(5)} &= - \frac{1}{4\sqrt{3}}F_{bc}\gamma^{bc}\gamma_a - \frac{1}{2\sqrt{3}}F_{ab}\gamma^b + \mathrm{i}\sqrt{3}A_aI \\
    &= - \frac{1}{2\sqrt{3}}E_I\gamma^{I}\gamma^0\gamma_a - \frac{1}{4\sqrt{3}}F_{IJ}\gamma^{IJ}\gamma_a - \frac{1}{2\sqrt{3}}F_{ab}\gamma^b + \mathrm{i}\sqrt{3}A_aI.
\end{align}
Using equations \ref{eq:gammaA4} and \ref{eq:gammaA5}, one finds $\tilde{\mathcal{A}}^{(4)}$ \& $\tilde{\mathcal{A}}^{(5)}$ exist and they are identical to $\mathcal{A}^{(4)}$ \& $\mathcal{A}^{(5)}$ except that $F_{IJ} \to -F_{IJ}$. In both cases, by following similar steps to the proof of lemma \ref{thm:supergravityPositiveEnergy}, one finds
\begin{align}
    \widetilde{\mathbb{M}} &= -2S_I\gamma^0\gamma^I,
\end{align}
where $S^I = F\indices{^I_J}E^J$ is the Poynting vector. Thus, $\widetilde{\mathbb{M}}$ has eigenvalues, $\pm 4\sqrt{S_IS^I}$, and is therefore not non-negative definite\footnote{This issue was also recently pointed out in \cite{Reall2024} and exists even when $\Lambda = 0$. An analogous proof to theorem \ref{thm:gIsomorphism} presented in this work therefore also fixes the issue in theorem 11.9 of \cite{ChruscielBartnik2003}, where $\Lambda = 0$.}.

Unlike the proof of theorem \ref{thm:gIsomorphism} presented above, it appears the Licherowicz identity based proof can only be rectified in a few specific cases where further assumptions are made. The simplest, but most unsatisfying, assumption would be to restrict to electromagnetic fields which have vanishing Poynting vector. 

As the following argument\footnote{I am very grateful to Piotr Chrusciel for providing this argument.} shows, a much less obvious assumption that also works is to restrict to hypersurfaces, $\Sigma_t$, which have $K = \delta^{IJ}K_{IJ} = 0$. Instead of the spacetime Levi-Civita connection, $D$, one could re-write the argument in terms of the connection intrinsic to $\Sigma_t$, say $D^{(\sigma)}$, where $\sigma$ is the metric on $\Sigma_t$. Then, for any Dirac spinor, $\Psi$, 
\begin{align}
    D_I\Psi = D_I^{(\sigma)}\Psi + \frac{1}{2}K_{IJ}\gamma^0\gamma^J \,\,\,\mathrm{and}\,\,\, \gamma^ID_I\Psi = \gamma^ID_I^{(\sigma)}\Psi + \frac{1}{2}K\gamma^0.
\end{align}
Instead of the $\widetilde{\nabla}_a$ defined earlier, one could define another connection, say $\nabla^\prime_I$, differing from $\nabla_I$ in not just $F_{IJ} \to -F_{IJ}$ (and a $-\frac{\mathrm{i}}{2}\gamma_I$ term instead of a $+\frac{\mathrm{i}}{2}\gamma_I$ term), but also $K_{IJ} \to -K_{IJ}$. Since $(-K_{IJ}, E_I, -F_{IJ})$ satisfies the constraint equations (encoded in $T^{0a}$) whenever $(K_{IJ}, E_I, F_{IJ})$ does, this connection will have $\mathbb{M}^\prime = 0$, which is trivially non-negative definite. $\gamma^I\nabla^\prime_I$ and $\gamma^I\widetilde{\nabla}_I$ differ by a $K\gamma^0$ term, but this goes to zero if $K$ is assumed to vanish.

$K = 0$ is a ``maximal gauge" that is sometimes used in the study of the initial value problem \cite{Fournodavlos2019}. However, it would have to be shown $(M, g)$ admits such a foliation and that the foliation is compatible with the coordinates chosen on $\mathcal{I}$ to get $f_{(0)0\alpha} = 0$.
\begin{theorem}
    Suppose $\Phi$ is a spinor such that $\gamma^I\nabla_I\Phi \in L^2$ and $\Phi$ grows at most as $O(\mathrm{e}^{r/2})$. By theorem \ref{thm:gIsomorphism}, let $\Psi \in \mathcal{H}$ be the unique spinor such that $\gamma^I\nabla_I\Psi = \gamma^I\nabla_I\Phi$. Let $\mathcal{Z} = \Phi - \Psi$. Let $\{\psi_m\}_{m = 0}^\infty \in C_c^\infty$ be a Cauchy sequence whose limit is $\Psi$ and let $\mathcal{Z}_m = \mathcal{Z} - \psi_m$. Then, for the functional, $Q$, in definition \ref{def:setup}, $\lim_{m \to \infty} Q(\mathcal{Z}_m) = Q(\mathcal{Z})$.
    \label{thm:bulkConvergence}
\end{theorem}
\begin{proof}
    As before, theorem \ref{thm:lichnerowicz} implies that for any spinor, $\chi$, 
    \begin{align}
        Q(\Phi - \chi) &=  2\int_{\Sigma_t}\Big(\nabla_I(\Phi - \chi)^\dagger\nabla^I(\Phi - \chi) - (\gamma^I\nabla_I(\Phi - \chi))^\dagger\gamma^J\nabla_J(\Phi - \chi) \nonumber \\
        &\,\,\,\,\,\,\,\,\,\,\,\,\,\,\,\,\,\,\,\,\,\,\,\, + (\Phi - \chi)^\dagger\mathbb{M}(\Phi - \chi)\Big)\mathrm{d}V.
    \end{align}
    Hence, by equation \ref{eq:Cc1InnerProduct}, definition \ref{def:h}, definition \ref{def:G}, lemma \ref{thm:GExtension} and $\gamma^I\nabla_I\Phi \in L^2$,
    \begin{align}
        \frac{1}{2}(Q(\mathcal{Z}) - Q(\mathcal{Z}_m)) &= \frac{1}{2}(Q(\Phi - \Psi) - Q(\Phi - \psi_m)) \\
        &= ||\Psi||_\mathcal{H}^2 - ||\psi_m||_\mathcal{H}^2 - ||\mathfrak{D}(\Psi)||_{L^2}^2 + ||\mathfrak{D}(\psi_m)||_{L^2}^2 + \langle \mathfrak{D}(\Psi - \psi_m), \gamma^I\nabla_I\Phi\rangle_{L^2} \nonumber \\
        &\,\,\,\,\,\,\, + \langle\gamma^I\nabla_I\Phi, \mathfrak{D}(\Psi - \psi_a)\rangle_{L^2} - \int_{\Sigma_t}(\nabla_I(\Psi - \psi_m))^\dagger\nabla^I(\Phi)\mathrm{d}V \nonumber \\
        &\,\,\,\,\,\,\, - \int_{\Sigma_t}\nabla^I(\Phi)^\dagger\nabla_I(\Psi - \psi_m)\mathrm{d}V - \int_{\Sigma_t}(\Psi - \psi_m)^\dagger\mathbb{M}\Phi\,\mathrm{d}V\nonumber \\
        &\,\,\,\,\,\,\, - \int_{\Sigma_t}\Phi^\dagger\mathbb{M}(\Psi - \psi_m)\,\mathrm{d}V.
    \end{align}
    Since inner products and $\mathfrak{D}$ are both continuous, it immediately follows that 
    \begin{align}
        \lim_{m\to\infty}\frac{1}{2}(Q(\mathcal{Z}) - Q(\mathcal{Z}_m)) &= \lim_{a\to\infty}\bigg(- \int_{\Sigma_t}(\nabla_I(\Psi - \psi_m))^\dagger\nabla^I(\Phi)\mathrm{d}V - \int_{\Sigma_t}\nabla^I(\Phi)^\dagger\nabla_I(\Psi - \psi_m)\mathrm{d}V\nonumber \\
        &\,\,\,\,\,\,\, - \int_{\Sigma_t}(\Psi - \psi_m)^\dagger\mathbb{M}\Phi\,\mathrm{d}V - \int_{\Sigma_t}\Phi^\dagger\mathbb{M}(\Psi - \psi_m)\,\mathrm{d}V\bigg).
    \end{align}
    Since the inner product on $\mathcal{H}$ is $\langle\psi, \chi\rangle_{\mathcal{H}} = \int_{\Sigma_t}\left((\nabla_I\psi)^\dagger\nabla^I\chi + \psi^\dagger\mathbb{M}\chi\right)\mathrm{d}V$ (with limits of Cauchy sequences taken appropriately when $\psi$ or $\chi$ is in $\mathcal{H}\backslash C_c^\infty$) and $\mathbb{M}$ is non-negative definite,
    \begin{align}
        \int_{\Sigma_t}(\nabla_I\psi)^\dagger\nabla^I(\psi)\,\mathrm{d}V \leq ||\psi||_\mathcal{H}^2 < \infty.
    \end{align}
    Hence $\nabla_I\psi \in L^2$ and $\psi \mapsto \nabla_I\psi$ is a continuous (i.e. bounded) linear operator. Consequently,
    \begin{align}
        \lim_{m\to\infty}\int_{\Sigma_t}(\nabla_I(\Psi - \psi_m))^\dagger\nabla^I(\Phi)\mathrm{d}V &= \lim_{m\to\infty}\langle\nabla_I(\Psi - \psi_m), \nabla^I\Phi\rangle_{L^2} = 0
    \end{align}
    and likewise for $\int_{\Sigma_t}\nabla^I(\Phi)^\dagger\nabla_I(\Psi - \psi_m)\mathrm{d}V$. That leaves
    \begin{align}
        \lim_{m\to\infty}\frac{1}{2}(Q(\varepsilon) - Q(\varepsilon_m)) = \lim_{m\to\infty}\bigg(- \int_{\Sigma_t}(\Psi - \psi_m)^\dagger\mathbb{M}\Phi\,\mathrm{d}V - \int_{\Sigma_t}\Phi^\dagger\mathbb{M}(\Psi - \psi_m)\,\mathrm{d}V\bigg).
    \end{align}
    Because it's assumed $\mathbb{M}$ is non-negative definite, $||\mathbb{M}||_0$ decays faster than $O(\mathrm{e}^{-(n-1)r})$ near $\Sigma_{t, \infty}$ and $\Phi$ grows at $O(\mathrm{e}^{r/2})$ near $\Sigma_{t, \infty}$,
    \begin{align}
        \bigg|\int_{\Sigma_t}\Phi^\dagger\mathbb{M}\Phi\,\mathrm{d}V\bigg| = \int_{\Sigma_t}\Phi^\dagger\mathbb{M}\Phi\,\mathrm{d}V \leq \int_{\Sigma_t}\Phi^\dagger\Phi||\mathbb{M}||_0\,\mathrm{d}V < \infty.
    \end{align}
    Therefore $\Phi\sqrt{||\mathbb{M}||_0} \in L^2$. Likewise, $(\Psi - \psi_m)\sqrt{||\mathbb{M}||_0} \in L^2$ because
    \begin{align}
        \int_{\Sigma_t}(\Psi - \psi_m)^\dagger(\Psi - \psi_m)||\mathbb{M}||_0\,\mathrm{d}V \leq \int_{\Sigma_t}(\Psi - \psi_m)^\dagger\mathbb{M}(\Psi - \psi_m)\mathrm{d}V \leq ||\Psi - \psi_m||_{\mathcal{H}}^2 < \infty.
    \end{align}
    Hence, by the Cauchy-Schwartz inequality and continuity of norms, 
    \begin{align}
        \lim_{m\to\infty}\bigg|\int_{\Sigma_t}(\Psi - \psi_m)^\dagger\mathbb{M}\Phi\,\mathrm{d}V\bigg| &\leq \lim_{m\to\infty}||(\Psi - \psi_m)\sqrt{||\mathbb{M}||_0}||_{L^2}||\Phi\sqrt{||\mathbb{M}||_0}||_{L^2} = 0.
    \end{align}
    Analogously, $\lim_{m\to\infty}\int_{\Sigma_t}\Phi^\dagger\mathbb{M}(\Psi - \psi_m)\,\mathrm{d}V = 0$ too, leaving $\lim_{m\to\infty}Q(\mathcal{Z}_m) = Q(\mathcal{Z})$.
\end{proof}

\end{appendices}
\newpage

\bibliographystyle{unsrt}
\bibliography{references}

\end{document}